\documentclass{article}[11pt]
\usepackage[utf8]{inputenc}
\usepackage{style}

\usepackage{framed}
\usepackage[parfill]{parskip}
\allowdisplaybreaks

\title{A Birthday Repetition Theorem and \\ Complexity of Approximating Dense CSPs}

\author{
  Pasin Manurangsi\thanks{University of California, Berkeley. Email: \texttt{pasin@berkeley.edu}.}
  \and
  Prasad Raghavendra\thanks{University of California, Berkeley. Email: \texttt{prasad@berkeley.edu}.}
}

\begin{document}

\maketitle

\begin{abstract}

A \emph{$(k \times l)$-birthday repetition} $\cG^{k \times l}$ of a two-prover game $\cG$ is a game in which the two provers are sent random sets of questions from $\cG$ of sizes $k$ and $l$ respectively. These two sets are sampled independently uniformly among all sets of questions of those particular sizes. We prove the following \emph{birthday repetition theorem}: when $\cG$ satisfies some mild conditions, $val(\cG^{k \times l})$ decreases exponentially in $\Omega(kl/n)$ where $n$ is the total number of questions. Our result positively resolves an open question posted by Aaronson, Impagliazzo and Moshkovitz~\cite{AIM}.

As an application of our birthday repetition theorem, we obtain new
fine-grained hardness of approximation results for dense CSPs.
Specifically, we establish a tight trade-off between running time and
approximation ratio for dense CSPs by showing conditional lower
bounds, integrality gaps and approximation algorithms.  In particular, for any sufficiently large $i$ and for every $k \geq 2$, we show the following results:

\begin{itemize}
\item We exhibit an $O(q^{1/i})$-approximation algorithm for dense {\sc Max $k$-CSP}s with alphabet size $q$ via $O_k(i)$-level of Sherali-Adams relaxation.
\item Through our birthday repetition theorem, we obtain an integrality gap of $q^{1/i}$ for $\tilde
  \Omega_k(i)$-level Lasserre relaxation for fully-dense {\sc
  Max $k$-CSP}.
\item Assuming that there is a constant $\epsilon > 0$ such that
  {\sc Max 3SAT} cannot be approximated to within $(1-\epsilon)$ of the optimal
  in sub-exponential time, our birthday repetition theorem implies that
  any algorithm that approximates fully-dense
  {\sc Max $k$-CSP} to within a $q^{1/i}$ factor takes
  $(nq)^{\tilde \Omega_k(i)}$ time, almost tightly matching the algorithmic result based on Sherali-Adams relaxation.
\end{itemize}

As a corollary of our approximation algorithm for dense {\sc Max $k$-CSP}, we give a new approximation algorithm for {\sc Densest $k$-Subhypergraph}, a generalization of {\sc Densest $k$-Subgraph} to hypergraphs. In particular, when the input hypergraph is $O(1)$-uniform and the optimal $k$-subhypergraph has constant density, our algorithm finds a $k$-subhypergraph of density $\Omega(n^{-1/i})$ in time $n^{O(i)}$ for any integer $i > 0$.

\end{abstract}

\newpage

\tableofcontents

\newpage

\section{Introduction}

Polynomial-time reductions between computational problems are among the central tools in complexity theory.
The rich and vast theory of hardness of approximation emerged out of the celebrated PCP Theorem~\cite{AroraLMSS} and the intricate web of polynomial-time reductions developed over the past two decades.
During this period, an extensive set of reduction techniques such as parallel repetition and
long-codes have been proposed and a variety of mathematical tools
including discrete harmonic analysis, information theory and Gaussian
isoperimetry have been applied towards analyzing these reductions.
These developments have led to an almost complete understanding of the
approximability of many fundamental combinatorial optimization problems like
{\sc Set Cover} and {\sc Max 3SAT}.  Yet, there are a few central problems such as computing approximate
Nash equlibria, the {\sc Densest $k$-Subgraph} problem and the {\sc Small Set Expansion} problem, that remain out of reach of the web of polynomial-time reductions.

A promising new line of work proposes to understand the complexity of
these problems through the lens of {\it sub-exponential time
reductions}. Specifically, the idea is to construct a sub-exponential
time reduction from {\sc 3SAT} to the problem at hand, say, the Approximate
Nash Equilibrium problem. Assuming that {\sc 3SAT} does not admit
sub-exponential time algorithms (also known as the Exponential Time
Hypothesis (ETH)~\cite{IP01}), this would rule out polynomial time
algorithms for the Approximate Nash Equilibrium problem.

At the heart of this line of works, lies the so-called {\it birthday  repetition} of two-prover games. To elaborate on this, we begin by formally defining the notion of two-prover games.

\begin{definition} (Two-prover game)
  A two prover game $\cG$ consists of
  \begin{itemize}
  \item A finite set of questions $X,Y$  and
    corresponding answer sets $\Sigma_X, \Sigma_Y$.
  \item A distribution $\cQ$ over pairs of questions $X
    \times Y$.
  \item A verification function $P: X \times Y \times \Sigma_X
    \times \Sigma_Y \to \{0,1\}$.
  \end{itemize}
  The value of the game is the maximum over all strategies $\phi : X
  \cup Y \to \Sigma_X \cup \Sigma_Y$ of the output of the verification
  function, i.e.,
  $val(\cG) = \max_{ \phi: X \cup Y \to \Sigma_X \cup \Sigma_Y}
  \E_{(x,y) \sim \cQ} [P(x,y,\phi(x),\phi(y))]$.
\end{definition}

Two prover games earn their name from the following interpretation of
the above definition:
The game $\cG$ is played between a verifier $V$ and two cooperating
provers $Merlin_1$ and $Merlin_2$ who have agreed upon a common
strategy, but cannot communicate with each other during the game.  The
verifier samples two questions $(x,y) \sim \cQ$ and sends $x$
to $Merlin_1$ and $y$ to $Merlin_2$.  The provers respond with answers
$\phi(x)$ and $\phi(y)$, which the verifier accepts or rejects based on
the value of the verifiaction function $P(x,y,\phi(x),\phi(y))$.

Two-prover games and, more specifically, a special class of two-prover
games known as the {\sc Label Cover} problem are the starting points for
reductions in a large body of hardness of approximation results.
The PCP theorem implies that for some absolute
constant $\epsilon_0$, approximating the value of a two prover game to within an additive $\epsilon_0$ is
$\mathbf{NP}$-hard.  However, this hardness result on its own is inadequate to
construct reductions to other combinatorial optimization problems.
To this end, this hardness result can be
strengthened to imply that it is $\mathbf{NP}$-hard to approximate the value of
two-prover games to any constant factor,  using the {\it parallel
  repetition theorem}.

For an integer $k$, the $k$-wise parallel repetition $\cG^{\otimes k}$
of a game $\cG$ can be described as follows.  The question and
answer sets in $\cG^{\otimes k}$ consist of $k$-tuples of questions and
answers from $\cG$.  The distribution over questions
in $\cG^{\otimes k}$ is given by the product distribution $\cQ^k$.  The
verifier for $\cG^{\otimes k}$ accepts the answers if and only if the
verifier for $\cG$ accepts each of the $k$ individual answers.

Roughly speaking, the parallel repetition theorem asserts that the
value of the repeated game $\cG^k$ decays exponentially in
$k$.  Parallel repetition theorems form a key ingredient in obtaining tight hardness of
approximation results, and have aptly received considerable attention
in literature~\cite{Raz98,Hol09,Rao11,DS14,Mos14,BG15}.

Birthday repetition, introduced by Aaronson \etal~\cite{AIM}, is an alternate transformation on two-prover games defined as follows.

\begin{definition} (Birthday Repetition)
  The $(k \times l)$-birthday repetition of a two-prover game $\cG$
  consists of
  \begin{itemize}

  \item The set of questions in $\cG^{k \times l}$ are
    $\binom{X}{k}$ and $\binom{Y}{l}$ respectively,
    i.e., each question is a subset $S \subseteq X$
    of size $k$ and subset $T \subseteq Y$ of size $l$.
  \item The distribution over questions is the uniform product
    distribution over $\binom{X}{k} \times \binom{Y}{l}$.
  \item The verifier accepts
    only if, for every pair of $(x, y) \in S \times T$ such that $(x,
    y)$ form a valid pair of questions in $\cG$, i.e., $(x,y)
    \in \supp(\cQ)$, the answers to $x$ and $y$ are accepted in the
    original game $\cG$.

  \end{itemize}
\end{definition}

The basic idea of birthday repetition can be traced back to the work of
Aaronson \etal~\cite{ABDFS09} on quantum multiprover proof systems
$\mathbf{QMA}(k)$ for {\sc 3SAT}.
Subsequent work by Aaronson \etal~\cite{AIM} on the classical
analogue of $\mathbf{QMA}(k)$, namely $\AM(k)$, formally
defined birthday repetition for two-prover games, and set the stage
for applications in hardness of approximation.

Unlike parallel repetition, birthday repetition is only effective for
large values of $k$ and $l$.  In particular, if $k, l <
o(\sqrt{|X|+|Y|})$,
then, for most pairs of $S$ and $T$, there is no pair of questions $(x,y) \in
S \times T$ such that $(x,y)$ belongs to the support of the questions
in the original game.

However, if we pick $k = l = \omega(\sqrt{n})$ where $n = |X| + |Y|$, then by the
birthday paradox, with high probability the sets $S,T$ contain an edge
$(x,y)$ from the original game $\cG$.  Hence, for this choice of $k$
and $l$, the game played by the provers is seemingly at least as
difficult to succeed, as the original game $\cG$.
Aaronson \etal~\cite{AIM} confirmed this intuition by proving the following theorem.

\begin{theorem} \cite{AIM} \label{thm:aim-birthday}
  For any two-prover game $\cG$ such that $\cQ$ is uniform over its support, if the bipartite graph induced by $(X, Y, \supp(\cQ))$ is biregular, then $val(\cG^{k \times l}) \leq val(\cG) + O(\sqrt{\frac{n}{kl}})$.
\end{theorem}

On the one hand, birthday repetition is ineffective in that it has to
incur a blowup of $2^{\sqrt{n}}$ in the size, to even simulate the
original game $\cG$.  The distinct advantage of birthday repetition is
that the resulting game $\cG^{k,l}$ has a distinct structure -- in
that it is a {\it free game}.
\begin{definition} (Free game)
  A free game is a two-player game $\cG = (X, Y, \cQ, \Sigma_X, \Sigma_Y, P)$ such that $\cQ$ is the uniform distribution over $X \times Y$.
\end{definition}

The birthday repetition theorem of Aaronson \etal~\cite{AIM}
immediately implies a hardness of approximation for the value of free
games.
Specifically, they show that it is ETH-hard to approximate free games
to some constant ratio in almost quasi-polynomial time. Interestingly,
this lower bound is nearly tight in that free games admit a quasipolynomial time approximation scheme (QPTAS)~\cite{BHHS11,AIM}.

Following Aaronson \etal's work, birthday repetition has received
numerous applications, which can be broadly classified in to two main
themes.
On the one hand, there are problems such as computing approximate Nash
equilibria~\cite{BKW15,BPR16},  approximating free games \cite{AIM}, and approximate symmetric
signaling in zero sum games \cite{R15}, where the underlying problems
admit quasipolynomial-time algorithms \cite{Dughmi14, LMM03, FS97} and birthday repetition can be
used to show that such a running time is necessary, assuming ETH.
On the other hand, there are computational problems like Densest
$k$-Subgraph \cite{BKRW15}, injective tensor norms
\cite{ABDFS09,HM13,BBHKSZ12}, $2$-to-$4$-norms
\cite{ABDFS09,HM13,BBHKSZ12} wherein an
$\mathbf{NP}$-hardness of approximation result seems out of reach of current
techniques.  But the framework of birthday repetition can be employed to
show a quasi-polynomial hardness assuming ETH\footnote{Although the
	hardness results for injective tensor norms and
	$2$-to-$4$-norms build over quantum multiprover proof systems,
	the basic idea of birthday repetition~\cite{ABDFS09} lies at
the heart of these reductions.}.



Unlike the parallel repetition theorem, the birthday repetition
theorem of \cite{AIM} does not achieve any reduction in the value of
the game.
It is thus natural to ask whether birthday repetition can be used to
decrease the value of a game, just like parallel repetition. Aaronson
et al. conjectured not only that the value of the game deteriorates
with birthday repetition, but also that it decreases exponentially in $\Omega(kl/n)$. Notice that the expected number of edges between $S$ and $T$ in birthday repetition is $\Theta(kl/n)$.

Our main technical contribution is that we resolve the conjecture
positively by showing the following theorem.

\begin{theorem} (Birthday Repetition Theorem (informal); See Theorem~\ref{thm:birthday-general},\ref{thm:birthday-proj}) \label{thm:main}
  Let $\cG = (X, Y, \cQ, \Sigma_X, \Sigma_Y, P)$ be a two-prover game such that $\cQ$ is uniform over its support, $(X, Y, \supp(\cQ))$ is biregular and $|\Sigma_X|, |\Sigma_Y|$ are constant. If $val(\cG) = 1 - \varepsilon$, then
  \begin{align*}
    val(\cG^{k \times l}) \leq 2(1 - \varepsilon/2)^{\Omega(\varepsilon^5 kl/n)}.
  \end{align*}
\end{theorem}
We note here that our theorem is, in fact, more general than stated above and can handle non-regular graphs and non-constant alphabet sizes as well (see Theorem~\ref{thm:birthday-general}). Moreover, we can get a better bound if $\cG$ is a {\sc Label Cover} instance (see Theorem~\ref{thm:birthday-proj}).

By definition, the birthday repetition theorem almost immediately
implies a hardness of approximation result for the value of a free
game.
\begin{corollary}
  Unless ETH is false, no polynomial time algorithm can approximate
  the value of a free game to within a factor of $2^{\tilde \Omega(\log(nq))}$
  where $n$ is the number of questions and $q$ is the alphabet (answer set) size.
\end{corollary}

The above hardness result improves upon $\polylog (nq)$ ratio achieved in~\cite{AIM} and is tight up to a factor of $\polyloglog (nq)$ in the exponent since there exists a polynomial-time algorithm that achieves $O(q^\varepsilon)$ approximation for every constant $\varepsilon > 0$~\cite{AIM,MM15}.

\subsection*{Dense CSPs}

A free game can be considered an instance of 2-ary constraint satisfaction problems. From this perspective, free games are \emph{dense}, in that there are constraints between a constant fraction of all pairs of variables.
As an application of our birthday repetition theorem, we will show
almost-tight hardness of approximation results for dense CSPs.  To
this end, we begin by defining {\sc Max $k$-CSP} and its density.

\begin{definition} ({\sc Max $k$-CSP})
  A {\sc Max $k$-CSP} instance $\cG$ consists of
  \begin{itemize}
  \item A finite set of variables $V$ and a finite alphabet set $\Sigma$.
  \item A distribution $\cQ$ over $k$-tuple of variables $V^k$.
  \item A predicate $P: V^k \times \Sigma^k \to [0,1]$.
  \end{itemize}
  The value of the instance is the maximum over all assignments $\phi : V \to \Sigma$ of the output of the predicate, i.e.,
  $val(\cG) = \max_{\phi: V \to \Sigma} \E_{S \sim \cQ} [P(S, \phi|_S)]$
  where $\phi|_S$ is the restriction of the assignment to $S$.

  Finally, an instance is called $\Delta$-dense if $\Delta \cdot \cQ(S) \leq 1 / |V|^k$ for every $S \in V^k$. Fully-dense instances are defined to be simply the 1-dense instances.
\end{definition}

There has been a long line of works on approximating dense CSPs.
Arora, Karger and Karpinski were first to devise  a polynomial-time approximation scheme for the problem when alphabet size is constant~\cite{AKK95}. Since then, numerous algorithms have been invented for approximating dense CSPs; these algorithms use wide ranges of techniques such as combinatorial algorithms with exhaustive sampling~\cite{AKK95,VKKV05,MS08,Yar14,MM15,FLP16}, subsampling of instances~\cite{AVKK03, BHHS11}, regularity lemmas~\cite{FK96,CCF10} and linear and semidefinite program hierarchies~\cite{VM07,BRS11,GS11,YZ14}. Among the known algorithms, the fastest is that of Yaroslavtsev~\cite{Yar14} that achieves approximation ratio  $(1 + \varepsilon)$ in $q^{O_k(\log q/\varepsilon^2)} + (nq)^{O(1)}$ time\footnote{\cite{Yar14} states that the algorithm takes $q^{O_k(1/\varepsilon^2)} + (nq)^{O(1)}$ time but it in fact takes $q^{O_k(\log q / \varepsilon^2)} + (nq)^{O(1)}$ time~\cite{Yar16}.} where $n$ and $q$ denote the number of variables and alphabet size respectively.

Unfortunately, when $q$ is (almost-)polynomial in $n$, none of the mentioned algorithms run in polynomial time. CSPs in such regime of parameters have long been studied in hardness of approximation (e.g.~\cite{Bellare93,RS97,AS03,DFKRS11,MR10,M}) and have recently received more attention from the approximation algorithm standpoint, both in the general case~\cite{Peleg07,CHK,MM15a} and the dense case~\cite{MM15}. The approximabilities of these two cases are vastly different. In the general case, it is known that, for some constant $k > 0$, approximating {\sc Max $k$-CSP} to within a factor of $2^{\log^{1 - \varepsilon} (nq)}$ is $\mathbf{NP}$-hard for any constant $\varepsilon > 0$~\cite{DFKRS11}. Moreover, the long-standing Sliding Scale Conjecture of Bellare et al.~\cite{Bellare93} states there are constants $k, \varepsilon > 0$ such that it is $\mathbf{NP}$-hard to approximate {\sc Max $k$-CSP} to within a factor of $(nq)^{\varepsilon}$. On the other hand, aforementioned algorithms for dense CSPs rule out such hardnesses for the dense case.

While the gap between known approximation algorithms and inapproximability results in the general case is tiny ($2^{\log^\varepsilon (nq)}$ for any constant $\varepsilon > 0$), the story is different for the dense case, especially when we restrict ourselves to polynomial-time algorithms. Aaronson \etal's result only rules out, assuming ETH, $\polylog(nq)$ factor approximation for such algorithms~\cite{AIM}. However, for $k > 2$, no non-trivial polynomial time algorithm for dense {\sc Max $k$-CSP} on large alphabet is even known. In this paper, we settle down the complexity of approximating dense {\sc Max $k$-CSP} almost completely by answering the following fine-grained question: ``for each
$i \in \N$, what is the best approximation for dense {\sc Max $k$-CSP}, achievable by algorithms
running in time $(nq)^{i}$?''.

Manurangsi and Moshkovitz developed
an algorithm for dense {\sc Max 2-CSP} that, when the instance has
value $1 - \varepsilon$, can approximate the value to within a factor
of $O(q^{1/i}/(1 - \varepsilon)^i)$ in $(nq)^{O(i)}$
time~\cite{MM15}\footnote{Note that it is unclear whether Aaronson,
Impagliazzo and Moshkovitz's algorithm~\cite{AIM} that
achieves a similar guarantee for free games can be extended to
handle dense {\sc Max 2-CSP}.}. Due to the algorithm's
combinatorial nature, it is unclear whether the algorithm can be extended to handle dense {\sc Max $k$-CSP}s when $k > 2$.

Using a conditioning-based rounding technique developed in
\cite{BRS11, RT12, YZ14}, we show that the Sherali-Adams relaxation exhibits a similar approximation even when $k > 2$, as stated below.

\begin{theorem}(Informal; See Theorem~\ref{thm:alg-dense-csp}) \label{thm:approx-inf}
  For every $i > 0$ and any dense {\sc Max $k$-CSP} instance of value $1 - \varepsilon$, an $O_{k, \varepsilon}(i/\Delta)$-level of the Sherali-Adams relaxation yields an $O(q^{1/i})$-approximation for the instance.
\end{theorem}

Using our birthday repetition theorem, we show that it is impossible
to improve the above tradeoff between run-time and
approximation ratio using the sum-of-squares SDP hierarchy (aka the
Lasserre SDP hierarchy).  Specifically, we use birthday repetition on
the $\Omega(n)$-level Lasserre integrality gap for {\sc Max 3XOR} by
Schoenebeck~\cite{Sch08} to show the following.

\begin{lemma}(Informal; See Lemma~\ref{thm:lasserre-gap}) \label{cor:lasserre-gap}
  For every sufficiently large $i > 0$, there is a fully-dense {\sc Max $k$-CSP} instance of value $1/(nq)^{1/i}$ such that the value of $\tilde \Omega_k(i)$-level Lasserre relaxation is one.
\end{lemma}


Instead, if we assume that there exists a constant $\epsilon > 0$ so that
{\sc Max 3SAT} cannot be approximated to $1-\epsilon$ in sub-exponential
time (which we call the Exponential Time Hypothesis for Approximating {\sc Max 3SAT} (ETHA)), then we can arrive at the following hardness result.

\begin{lemma}(Informal; See Lemma~\ref{thm:hardness-linear})
  Unless ETHA is false, for every sufficiently large $i > 0$, no $(nq)^{\tilde O_k(i)}$-time algorithm can approximate fully-dense {\sc Max $k$-CSP} to within a factor of $(nq)^{1/i}$.
\end{lemma}

Thus, assuming ETHA, our hardness result and algorithm resolve complexity of approximating dense CSPs up to a factor of $\polylog i$ and a dependency on $k$ in the exponent of the running time.

\subsection*{{\sc Densest $k$-Subhypergraph}}

As a by-product of our approximation algorithm for dense {\sc Max $k$-CSP}, we will give a new approximation algorithm for {\sc Densest $k$-Subhypergraph}, the generalization of {\sc Densest $k$-Subgraph} to hypergraphs defined below, in the regime where the input hypergraph is $d$-uniform for some constant $d > 0$ and the optimal subhypergraph is sufficiently dense.

\begin{definition}[{\sc Densest $k$-Subhypergraph}]
  Given a hypergraph $(V, E)$ as an input, find a subset $S \subseteq V$ of $k$ vertices that maximizes the number of edges contained in the subhypergraph induced on $S$.
\end{definition}

When the input hypergraph is simply a graph, the problem becomes the {\sc Densest $k$-Subgraph} problem, which has been extensively studied from the approximation algorithm viewpoint dating back to the early '90s~\cite{KP93,FKP99,FS97,ST05,BCCFV}. On the other hand, {\sc Densest $k$-Subhypergraph} was first studied in 2006, when Hajiaghayi \etal~\cite{Haj06} proved that, if 3SAT $\notin \mathbf{DTIME}(2^{n^{3/4 + \varepsilon}})$ for some $\varepsilon > 0$, then there is no polynomial-time algorithm that approximates {\sc Densest $k$-Subhypergraph} to within a factor of $2^{\log^\delta n}$ for some $\delta > 0$. Later, Applebaum~\cite{App13} showed, under a cryptographic assumption, that, for sufficiently large $d$, {\sc Densest $k$-Subhypergraph} on $d$-uniform hypergraph is hard to approximate to a factor of $n^{\varepsilon}$ for some $\varepsilon > 0$. Chuzhoy \etal~\cite{Chu15} then used this result to establish a hardness of approximation for the $k$-route cut problem. More recently, Chlamt{\'{a}}c \etal~\cite{Ch16} provided the first non-trivial approximation algorithm for the problem; their algorithm works only on $3$-uniform hypergraph and achieves $O(n^{4(4 - \sqrt{3})/13 + \varepsilon})$-approximation for any constant $\varepsilon > 0$ in polynomial time.

Thanks to Charikar \etal's~\cite{CHK} reduction from {\sc Densest $k$-Subgraph} to {\sc Max 2-CSP}, which can be easily adapted to a reduction from {\sc Densest $k$-Subhypergraph} on $d$-uniform hypergraph to {\sc Max $d$-CSP}, Theorem~\ref{thm:approx-inf} immediately implies the following approximation algorithm for {\sc Densest $k$-Subhypergraph}.

\begin{corollary}(Informal; See Corollary~\ref{cor:dense-hypergraph})
  There is a randomized algorithm that, given a $d$-uniform hypergraph whose densest $k$-subhypergraph is $\Delta$-dense and an integer $i > 0$, runs in $n^{O_{\Delta, d}(i)}$ time and outputs a $k$-subhypergraph of density $\Omega_{\Delta, k}(n^{-1/i})$ with high probability.
\end{corollary}

Here we use $n$ to denote the number of vertices of the graph and we define the density of a $d$-uniform hypergraph to be simply $d! |E|/|V|^d$. We remark that our algorithm is incomparable to Chlamt{\'{a}}c \etal's~\cite{Ch16} as their algorithm works for any $\Delta$ while ours requires $\Delta$ to be sufficiently large.

Note also that the density condition required is on the optimal output not the input hypergraph. Moreover, when $\Delta$ and $d$ are constant, the above corollary gives an $n^{O(i)}$-time $O(n^{1/i})$-approximation algorithm for {\sc Densest $k$-Subhypergraph} for every $i > 0$. When $d = 2$, this matches exactly with the previously known approximation algorithms for {\sc Densest $k$-Subgraph}~\cite{FS97,ST05,MM15}.

\subsection*{Almost Optimal $\AM(2)$ Protocol for {\sc 3SAT}}

Another interpretation of our improved hardness of approximation of free games is as an improved $\AM(2)$ protocol for {\sc 3SAT}. The Arthur-Merlin ($\AM$) protocol~\cite{Bab85} is a protocol where Arthur (verifier) tosses some random coins and sends the results to Merlin (prover). The prover sends back a proof to Arthur who then decides whether to accept it. Motivated by quantum complexity class $\mathbf{QMA}(k)$, Aaronson \etal~\cite{AIM} proposes a multi-prover version of $\AM$ called $\AM(k)$ where there are $k$ non-communicating Merlins\footnote{$\AM(k)$ is not to be confused with $\AM[k]$ defined in~\cite{Bab85}. In $\AM[k]$, there is only one Merlin but Arthur and Merlin are allowed to engage in $k$ rounds of communication.}. Authur sends an independent random challenge to each Merlin who then sends an answer back to Arthur. Finally, Arthur decides to accept or reject based on the received answers. The protocol is formally defined below.

\begin{definition}\cite{AIM}
  An $\AM(k)$ protocol for a language $L \subseteq \{0, 1\}^*$ of length $p(n) = kq(n)$, completeness $c(n)$, and soundness $s(n)$ consists of a probabilistic polynomial-time verifier $V$ such that
  \begin{itemize}
  \item (Completeness) For every $x \in L$, there exists functions $m_1, \dots, m_k: \{0, 1\}^{q(n)} \rightarrow \{0, 1\}^{q(n)}$ such that $\Pr_{y_1, \dots, y_k \sim \{0, 1\}^{q(n)}}[V(x, y_1, \dots, y_k, m(y_1), \dots, m(y_k))] \geq c(n)$, and,
  \item (Soundness) For every $x \notin L$ and for every function $m_1, \dots, m_k: \{0, 1\}^{q(n)} \rightarrow \{0, 1\}^{q(n)}$, we have $\Pr_{y_1, \dots, y_k \sim \{0, 1\}^{q(n)}}[V(x, y_1, \dots, y_k, m(y_1), \dots, m(y_k))] \leq s(n)$
  \end{itemize}
\end{definition}

The complexity class $\AM_{p(n)}(k)$ is a set of all languages $L$ such that there exists an $\AM(k)$ protocol of length $p(n)$, completeness 1/3, and soundness 2/3. Finally, the class $\AM(k)$ is defined as $\bigcup_{c \in \mathbb{N}}$ $\AM_{n^c}(k)$.

Similar to the interpretation of a two-prover game as a two-prover protocol, a free game can be viewed as an $\AM(2)$ protocol. Under this view, inapproximabilities of free games translate to $\AM(2)$ protocols whereas approximation algorithms for free games translate to lower bounds on the lengths of $\AM(2)$ protocols.

With this viewpoint, Aaronson et al. constructed, via birthday repetition, an $\AM(2)$ protocol of length $n^{1/2 + o(1)}\poly(1/\delta)$ for {\sc 3SAT} with completeness 1 and soundness $\delta$ for every $\delta > 0$. They also showed a lower bound of $\Omega(\sqrt{n}\log(1/\delta))$ on the length of such protocol. Equipped with our birthday repetition theorem, we construct an $\AM(2)$ protocol whose length is optimal up to a factor of $\polylog n$.

\begin{lemma} \label{lem:am-protocol}
  For any $\delta > 0$, there is an $\AM(2)$ protocol for {\sc 3SAT} of length $\tilde O(\sqrt{n \log(1/\delta)})$ with completeness 1 and soundness $\delta$.
\end{lemma}

We note that, by picking $\delta = 1/3$, Lemma~\ref{lem:am-protocol} immediately imply {\sc 3SAT} $\in \AM_{\tilde O(\sqrt{n})}(2)$. Since every problem in $\mathbf{NTIME}(n)$ is reducible to a quasi-linear size {\sc 3SAT} instance~\cite{Cook88}, we arrive at the following corollary, resolving the first open question posted in~\cite{AIM}.

\begin{corollary}
  $\mathbf{NTIME}(n) \subseteq \AM_{\tilde O(\sqrt{n})}(2)$.
\end{corollary}

\subsection*{Organization of the Paper}

The rest of the paper is organized as follows. In the following section, we provide preliminaries and state notations that we use in the paper. Then, in Section~\ref{sec:birthday}, we prove our main theorems. Next, Section~\ref{sec:app} demonstrates applications of our birthday repetition theorem, including new hardnesses of approximation and Lasserre integrality gap for dense CSPs, and an almost optimal $\AM(2)$ protocol for {\sc 3SAT}. The algorithm for dense {\sc Max $k$-CSP} is described and its approximation guarantee is proved in Section~\ref{sec:alg}; the approximation algorithm for {\sc Densest $k$-Subhypergraph} is also given at the end of the section. Finally, we conclude by proposing open questions and future research directions in Section~\ref{sec:open}.

\section{Preliminaries and Notations} \label{sec:notation}

In this section, we define notations and state some well-known facts that will be used in the paper.

\subsection{Miscellaneous}

For any positive integer $n$, we use $[n]$ to denote the set $\{1, \dots, n\}$. For two sets $X$ and $S$, define $X^S$ to be the set of tuples $(x_s)_{s \in S}$ indexed by $S$ with $x_S \in X$. We sometimes view each tuple $(x_s)_{s \in S}$ as a function from $S$ to $X$. For a set $S$ and an integer $n \leq |S|$, we use $\binom{S}{n}$ to denote the collection of all subsets of $S$ of size $n$. For convenience, we let $\binom{S}{0} = \{\emptyset\}$. We use $\binom{S}{[n]}$ to denote $\binom{S}{0} \cup \cdots \cup \binom{S}{n}$. For any bipartite graph $(A, B, E)$ and any $S \subseteq A, T \subseteq B$, let $E(S, T)$ denote the set of all edges with one endpoint in $S$ and the other in $T$.

Throughout the paper, we use $\log$ to denote the natural logarithm. We write $\polylog n$ and $\polyloglog n$ as shorthands for $\log^c n$ and $(\log \log n)^c$ for some constant $c > 0$ respectively. Finally, $\tilde \Omega(f(n))$ and $\tilde O(f(n))$ are used to denote $\bigcup_{c \in \mathbb{N}} \Omega(f(n)/\log^c f(n))$ and $\bigcup_{c \in \mathbb{N}} O(f(n) \log^c f(n))$ correspondingly.

\subsection{Probability Theory and Information Theory}

Throughout the paper, we use calligraphic letters to denote probability distributions. Let $\mathcal{X}$ be a probability distribution over a finite probability space $\Theta$. We use $x \sim \mathcal{X}$ to denote a random variable $x$ sampled according to $\mathcal{X}$. Sometimes we use shorthand $x \sim \Theta$ to denote $x$ being drawn uniformly at random from $\Theta$. For each $\theta \in \Theta$, we denote $\Pr_{x \sim \mathcal{X}}[x = \theta]$ by $\mathcal{X}(\theta)$. The \emph{support} of $\mathcal{X}$ or $\supp(\mathcal{X})$ is the set of all $\theta \in \Theta$ such that $\mathcal{X}(\theta) \ne 0$. For any event $E$, we use $\mathds{1}[E]$ to denote the indicator variable for the event.

Let us define some information theoretic notions that will be useful in the analysis of our algorithm. The \emph{informational divergence} (aka \emph{Kullback-Leibler divergence}) between two probability distributions $\mathcal{X}$ and $\mathcal{Y}$ is $D_{KL}(\mathcal{X}\|\mathcal{Y}) = \sum_{\theta \in \supp(\mathcal{X})}\mathcal{X}(\theta)\log(\mathcal{X}(\theta)/\mathcal{Y}(\theta)).$
Note that, when $\supp(\mathcal{Y}) \not\subseteq \supp(\mathcal{X})$, we let $D_{KL}(\mathcal{X}\|\mathcal{Y}) = \infty$. It is well-known that $D_{KL}(\mathcal{X}\|\mathcal{Y}) \geq 0$ for any distributions $\mathcal{X}$ and $\mathcal{Y}$.

The \emph{entropy} of a random variable $x \sim \mathcal{X}$ is defined as $H(x) = -\sum_{\theta \in \supp(\mathcal{X})} \mathcal{X}(\theta)\log \mathcal{X}(\theta)$. For jointly distributed random variables $x_1, \dots, x_n$, the entropy of $x_1, \dots, x_n$ is defined similarly as $H(x_1, \dots, x_n) = -\sum_{(\theta_1, \dots, \theta_n) \in \supp(\mathcal{X}_{1, \dots, n})} \mathcal{X}_{1, \dots, n}(\theta)\log \mathcal{X}_{1, \dots, n}(\theta)$ where $\mathcal{X}_{1, \dots, n}$ is the joint distribution of $x_1, \dots, x_n$. The \emph{mutual information} of $x_1, \dots, x_n$ is defined as $I(x_1; \dots; x_n) = \sum_{S = \{i_1, \dots, i_m\} \subseteq [n] \atop S \ne \emptyset} (-1)^{m - 1} H(x_{i_1}, \dots, x_{i_m})$.

The conditional entropy $H(x_1, \dots, x_{n - 1} | x_n)$ is defined as $\E_{\theta \sim \supp(\mathcal{X}_n)} [H(x_1, \dots, x_{n - 1}) | x_n = \theta]$ where $\mathcal{X}_n$ is the marginal distribution of $x_n$. Similarly, the conditional mutual information $I(x_1; \dots; x_{n - 1} | x_n)$ is defined as $\E_{\theta \sim \supp(\mathcal{X}_n)} [I(x_1; \dots; x_{n - 1}) | x_n = \theta]$. The following identity is well-known and is, in fact, often used an a definition for mutual information.
\begin{lemma} \label{lem:cond-mutual-info}
  For any random variables $x_1, \dots, x_n$, we have $I(x_1; \dots; x_n) = I(x_1; \dots; x_{n - 1}) - I(x_1; \dots; x_{n - 1} | x_n)$.
\end{lemma}

Last information theoretic measure we will use is the \emph{total correlation} defined as $C(x_1; \dots; x_n) = D_{KL}(\mathcal{X}_{1, \dots, n} \| \mathcal{X}_1 \times \cdots \times \mathcal{X}_n)$ where $\mathcal{X}_{1, \dots, n}$ is the joint distribution of $x_1, \dots, x_n$ whereas $\mathcal{X}_1, \dots, \mathcal{X}_n$ are the marginal distributions of $x_1, \dots, x_n$ respectively. We note that the total correlation defined here is always non-negative whereas the mutual information can be negative.

The total correlation is related to entropies and mutual information as follows.

\begin{lemma} \label{lem:total-cor-entropy}
  For any random variables $x_1, \dots, x_n$, we have $C(x_1; \dots; x_n) = \sum_{i \in [n]} H(x_i) - H(x_1; \dots; x_n)$.
\end{lemma}

\begin{lemma} \label{lem:total-cor-mutual-info}
  For any random variables $x_1, \dots, x_n$, we have $C(x_1; \dots; x_n) = \sum_{S = \{i_1, \dots, i_m\} \subseteq [n] \atop |S| \geq 2} I(x_{i_1}; \dots; x_{i_m})$.
\end{lemma}

Finally, similar to conditional entropy and conditional mutual information, we define the conditional total correlation as $C(x_1; \dots; x_{n - 1} | x_n) = \E_{\theta \sim \supp(\mathcal{X}_n)} [C(x_1; \dots; x_{n - 1}) | x_n = \theta]$.

\subsection{Two-prover Game, Free Game and {\sc Max $k$-CSP}}

Two-prover games, free games, and {\sc Max $k$-CSP} are defined in similar manners as in the introduction. However, for convenience, we write the predicates as $P_S(\phi|_S)$ instead of $P(S, \phi|_S)$, and, when $\cQ$ is the uniform distribution on $\Theta$, we sometimes write the instance as $(V, \Theta, \{P_S\})$ instead of $(V, \cQ, \{P_S\})$. Moreover, for an assignment $\phi$ of a {\sc Max $k$-CSP} instance $\cG = (V, \cW, \{P_S\})$, we define its value as $val_{\cG}(\phi) = \E_{S \sim \cW}[P_S(\phi|_S)]$. When it is clear from the context, we will drop $\cG$ and write it simply as $val(\phi)$. Note that $val(\cG)$ is the maximum of $val_{\cG}(\phi)$ among all possible assignments $\phi$'s. We say that $\cG$ is \emph{satisfiable} if its value is one.

We use $n$ to denote the number of variables $|V|$, $q$ to denote the alphabet size $|\Sigma|$ and $N$ to denote the instance size $|\supp(\cW)|q^k$, the number of bits needed to encode the input if each predicate is a boolean function. Note that, when the instance is fully dense, $N$ is simply $(nq)^k$. Similar notations are also used for two-prover games and free games.

Finally, we define projection games (aka {\sc Label Cover}), two-prover games with ``projection'' predicates.
\begin{definition}
  A two-prover game $\cG = (X, Y, \cQ, \Sigma_X, \Sigma_Y, \{P_{(x, y)}\})$ is a projection game if, for each $(x, y) \in \supp(\cQ)$, there exists a function (or projection) $f_{(x, y)}: \Sigma_X \to \Sigma_Y$ such that, for all $\sigma_x \in \Sigma_X$ and $\sigma_y \in \Sigma_Y$, $P_{(x, y)}(\sigma_x, \sigma_y) = 1$ if and only if $f_{(x, y)}(\sigma_x) = \sigma_y$.
\end{definition}

\subsection{Parallel and Birthday Repetitions}

We have already described the parallel (or tensor) repetition and birthday repetition in the introduction. Below are the formal notations we use to refer to them throughout the paper.

\begin{definition}
The $k$-parallel repetition of a two-prover game $\cG = (X, Y, \cQ, \Sigma_X, \Sigma_Y, \{P_{(x, y)}\})$ is a two-prover game $\cG^{\otimes k} = (X^k, Y^k, \cQ^k, \Sigma_X^k, \Sigma_Y^k, \{P^k_{((x_1, \dots, x_k), (y_1, \dots, y_k))}\})$ defined as follows. $X^k, Y^k, \Sigma_X^k, \Sigma_Y^k$ are defined in the usual Cartesian product sense. $\cQ^k$ is defined by $\cQ^k((x_1, \dots, x_k), (y_1, \dots, y_k)) = \prod_{i=1}^{k} \cQ(x_i, y_i).$
Finally, the predicates are defined as
$P^k_{((x_1, \dots, x_k), (y_1, \dots, y_k))}((\sigma_{x_1}, \dots, \sigma_{x_k}), (\sigma_{y_1}, \dots, \sigma_{y_k})) = \prod_{i=1}^{k} P_{x_i, y_i}(\sigma_{x_i}, \sigma_{y_i}).$
\end{definition}

\begin{definition}
For any $k \leq |X|$ and $l \leq |Y|$, a $(k \times l)$-birthday repetition of a two-player game $\cG = (X, Y, \cQ, \Sigma, \{P_{x, y}\})$ is a two-prover game $\cG^{k \times l} = (\binom{X}{k}, \binom{Y}{l}, \mathcal{U}^{k \times l}, \Sigma_X^{k}, \Sigma_Y^{l}, \{P^{k \times l}_{(S, T)}\}_{S \in \binom{X}{k}, T \in \binom{Y}{l}})$ defined as follows. $\mathcal{U}^{k \times l}$ is simply the uniform distribution over $\binom{X}{k} \times \binom{Y}{l}$. $\Sigma_X^{k}, \Sigma_Y^{l}$ are defined in the usual Cartesian product sense. Lastly, $P^{k \times l}_{(S, T)}$  is defined as $P^{k \times l}_{(S, T)}((\sigma_x)_{x \in S}, (\sigma_{y})_{y \in T}) = \prod_{(x, y) \in (S \times T) \cap \supp(\cQ)} P_{(x, y)}(\sigma_x, \sigma_y).$

Note that an empty product is defined as one, i.e., if $(S \times T) \cap \supp(\cQ) = \emptyset$, then $P^{k \times l}_{(S, T)}$ is identically one.
\end{definition}

\subsection{Sherali-Adams and Lasserre Hierarchies}

In this paper, we consider two hierarchies of linear and semidefinite program relaxations of {\sc Max $k$-CSP}. For compactness, we only write down the relaxations of {\sc Max $k$-CSP} but do not describe the hierarchies in full details. For interested readers, we refer to Chlamt{\'{a}}c and Tulsiani's survey on the topic~\cite{CT12}.

The first hierarchy we consider is the Sherali-Adams (SA) hierarchy, introduced in~\cite{SA90}. An {\em $r$-level SA solution} of a {\sc Max $k$-CSP} instance $\cG = (V, \cW, \{P_S\})$ is a collection $\mu = \{\mathcal{X}_S\}_{|S| \leq t}$ of distributions $\mathcal{X}_S$ on $\Sigma^S$ for every subset $S$ of $V$ of size at most $r$ such that, for every $S, T \subseteq V$ of size at most $r$, the marginal probability of $\mathcal{X}_S$ and $\mathcal{X}_T$ on $\Sigma^{S \cap T}$ agrees. The value of an $r$-level SA solution $\mu$ for $r \geq k$ is defined to be $val_{SA}(\mu) = \E_{S \sim \cW}[\E_{x_S \sim \mu}[P_S(x_S)]]$ where $\E_{x_S \sim \mu}[P_S(x_S)]$ is a shorthand for $\E_{\phi_S \sim \mathcal{X}_{\{i_1, \dots, i_k\}}}[P_S(\phi_S)]$ when $S = (x_{i_1}, \dots, x_{i_k})$. The optimal of the $r$-level SA relaxation of $\cG$, $opt_{SA}^r(\cG)$, is defined as the maximum value among all the $r$-level SA solutions. It is easy to see that finding $opt_{SA}^r(\cG)$ can be formulated as a linear program with at most $(nq)^{O(r)}$ variables and inequalities and, thus, can be solved in $(nq)^{O(r)}$ time.

Another hierarchy we consider is the Lasserre hierarchy~\cite{Lass00}. Before stating the Lasserre relaxation for {\sc Max $k$-CSP}, we define additional notations regarding assignments. Two assignments $\phi_1 \in \Sigma^{S_1}, \phi_2 \in \Sigma^{S_2}$ are said to be \emph{consistent} if $\phi_1(x) = \phi_2(x)$ for all $x \in S_1 \cap S_2$. The two assignments are said to be \emph{inconsistent} otherwise. More than two assignments are consistent if every pair of the assignments is consistent; otherwise, they are said to be inconsistent.  Moreover, for two consistent assignments $\phi_1 \in \Sigma^{S_1}, \phi_2 \in \Sigma^{S_2}$, we define $\phi_1 \circ \phi_2 \in \Sigma^{S_1 \cap S_2}$ by $\phi_1 \circ \phi_2 (x) = \phi_1(x)$ if $x \in S_1$ and $\phi_1 \circ \phi_2(x) = \phi_2(x)$ otherwise.

An {\em $r$-level Lasserre solution} of an instance $\cG = (V, \cW, \{P_S\})$ is a collection $\{U_{(S, \phi_S)}\}_{|S| \leq r, \phi_S \in \Sigma^S}$ of vectors $U_{(S, \phi_S)}$ for all $S \subseteq V$ of size at most $r$ and assignments $\phi_S$ of $S$ satisfying the following constraints.
\begin{align*}
  \langle U_{(S_1, \phi_1)}, U_{(S_2, \phi_2)} \rangle &\geq 0 & \forall S_1, S_2, \phi_1 \phi_2 \\
  \langle U_{(S_1, \phi_1)}, U_{(S_2, \phi_2)} \rangle &= \langle U_{(S_3, \phi_3)}, U_{(S_4, \phi_4)} \rangle &  \forall S_1 \cup S_2 = S_3 \cup S_4 \text{ and } \phi_1 \circ \phi_2 &= \phi_3 \circ \phi_4 \\
  \langle U_{(S_1, \phi_1)}, U_{(S_2, \phi_2)} \rangle &= 0 &  \forall S_1, S_2, \phi_1 \phi_2 \text{ s.t. } \phi_1, \phi_2 \text{ are inconsistent} \\
   \sum_{\sigma \in \Sigma}\|U_{(x, \sigma)}\|^2 &= 1 &\forall x \in V \\
  \|U_{(\emptyset, \emptyset)}\| &= 1
\end{align*}
where $S_1, S_2, S_3, S_4$ are over all subset of $V$ of size at most $r$ and $\phi_1, \phi_2, \phi_3, \phi_4$ are over all assignments of $S_1, S_2, S_3, S_4$ respectively. The value of an $r$-level Lasserre solution $\{U_{(S, \phi_S)}\}$ is defined as $val_{Las}(\{U_{(S, \phi_S)}\}) = \E_{S \sim \cW} [\sum_{\phi_S \in \Sigma^S} \|U_{(S, \phi_S)}\|^2 P_S(\phi_S)]$. A Lasserre solution is called \emph{complete} if its value is one.

Note that we abuse the notation here as $S$ in $\{U_{(S, \phi_S)}\}$ is a set whereas $S$ in $\cW$ is a tuple. Here and elsewhere in the paper, when we write $U_{(S, \phi_S)}$ for some tuple $S = (x_{i_1}, \dots, x_{i_m})$, this simply refers to $U_{\{x_{i_1}, \dots, x_{i_m}\}, \phi_S}$ if the assignment $\phi_S$ does not assign the same variable to different values and the all zero vector otherwise. Finally, we use $opt_{Las}^r(\cG)$ to denote the maximum value among all $r$-level Lasserre solutions $\{U_{(S, \phi_S)}\}$.

It is not hard to see that finding $opt_{Las}^r(\cG)$ can be formulated as SDP with $(nq)^{O(r)}$ variables and, hence, can be approximated up to arbitrarily small error within $(nq)^{O(r)}$ time. Moreover, it is known that the $r$-level Lasserre relaxation is stronger than the $r$-level SA relaxation~\cite{Lau03}. In the case of {\sc Max $k$-CSP}, this can be easily seen since we can define an $r$-level SA solution $\mu = \{\mathcal{X}_S\}_{|S| \leq t}$ from an $r$-level Lasserre solution $\{U_{(S, \phi_S)}\}$ by $\mathcal{X}_S(\phi_S) = \|U_{(S, \phi_S)}\|^2$.

\subsection{Exponential Time Hypotheses}

Here we formally state the ETH and ETHA mentioned in the introduction.

\begin{conjecture}(Exponential Time Hypothesis for {\sc 3SAT} (ETH)~\cite{IP01})
  There exists a constant $c > 0$ such that no $O(2^{cn})$-time algorithm can solve {\sc 3SAT} where $n$ denote the number of clauses.
\end{conjecture}

\begin{conjecture}(Exponential Time Hypothesis for Approximating {\sc Max 3SAT} (ETHA))
  There exists a constant $\varepsilon, c > 0$ such that no algorithm running in time $O(2^{cn})$ can distinguish between a satisfiable {\sc 3SAT} formula from a {\sc 3SAT} formula whose at most $1 - \varepsilon$ fraction of the clauses is satisfiable. Again, here $n$ denotes the number of clauses.
\end{conjecture}

Note that the two conjectures remain equivalent even when $n$ denotes the number of variables. For ETH, this is due to the well-known sparsification lemma of Impagliazzo, Paturi and Zane~\cite{IPZ01}. For ETHA, this is implied by the following simple observation: if a {\sc 3SAT} instance of $m$ clauses has value at most $1 - \varepsilon$, then an instance created by subsampling $\Omega_{\varepsilon}(n)$ clauses has value at most $1 - \varepsilon/2$ with high probability. This can be proved via standard arguments involving Chernoff and Union bounds. (See, for example, the proof of Lemma 2.1 in~\cite{DKR16}, which contains a similar statement for 2-CSP.)

ETHA is also introduced independently as gap-ETH by Dinur~\cite{Dinur16} who uses it to provide a supporting evidence to the Sliding Scale Conjecture. We remark that an evidence supporting ETHA is that $\Omega(n)$-level of the Lasserre hierarchy, a powerful tool in approximating CSPs, cannot distinguish satisfiable {\sc 3SAT} formulae from those whose only $7/8 + \varepsilon$ fraction of clauses is satisfiable for any constant $\varepsilon > 0$~\cite{Sch08}. In fact, no subexponential time algorithm is even known for distinguishing a satisfiable {\sc 3SAT} formula from a random formula.

Regarding relations between ETH and ETHA, it is obvious that ETHA implies ETH. On the other hand, the reverse direction is not yet known. Dinur's PCP~\cite{Dinur07} implies only $2^{O(n / \polylog n)}$ time lower bound for approximating {\sc 3SAT} to within $1 - \varepsilon$ factor for some $\varepsilon > 0$. One possible way for ETH to imply ETHA is if there exists a linear-length constant-query PCP for {\sc 3SAT} (i.e. {\sc 3SAT} $\in \mathbf{PCP}_{\delta, 1}[\log n + O(1), O(1)]$ for some constant $\delta < 1$). However, such PCP is not currently known.

\subsection{Some Useful Bounds}

Finally, we list simple bounds and inequalities that will be used in our proofs. We start with a concentration bound on number of edges in a random subgraph of a bipartite graph.

\begin{lemma} \label{lem:random-num-edges}
  Let $(X, Y, E)$ be any bipartite graph where each vertex has degree at most $d_{max}$. For any non-negative integers $k \leq |X|$ and $l \leq |Y|$, let $s = \frac{kl|E|}{|X||Y|}$. For any non-negative number $\gamma < 1/2$, we have
  \begin{align*}
    \Pr_{S \sim \binom{X}{k}, T \sim \binom{Y}{l}}[|E(S, T)| \notin [(1 - \gamma)s, (1 + \gamma)s]] \leq 4\exp\left(-\frac{\gamma^2 s}{54 d_{max}}\right).
  \end{align*}
\end{lemma}

For completeness, we give a proof of Lemma~\ref{lem:random-num-edges} in Appendix~\ref{app:random-num-edges}.

In our analysis, we often want to bound a value of a two-prover game based on a value of another game defined on the same question sets, alphabet sets, and predicates but differ on the distribution. Below are a couple of useful bounds to help us do so; the proofs for both lemmas can be found in Appendix~\ref{app:inq-games-dist}.

\begin{lemma} \label{lem:inq-mult}
  Let $\cG = (X, Y, \cQ, \Sigma_X, \Sigma_Y, \{P_{x, y}\}_{(x, y)})$ and $\cG' = (X, Y, \cQ', \Sigma_X, \Sigma_Y, \{P_{x, y}\}_{(x, y)})$ be two games on the same set of questions, alphabets, and predicates. If $\cQ(x, y) \leq \alpha \cdot \cQ'(x, y)$ for some $\alpha$ for all $x \in X, y \in Y$, then $val(\cG) \leq \alpha \cdot val(\cG')$.

  In particular, when $\cQ$ and $\cQ'$ are uniform distributions on some $E \subseteq E'$ respectively, $val(\cG) \leq \frac{|E'|}{|E|} \cdot val(\cG')$.
\end{lemma}

\begin{lemma} \label{lem:inq-cond}
  Let $\cG = (X, Y, \cQ, \Sigma_X, \Sigma_Y, \{P_{x, y}\}_{(x, y)})$ be any two player game and let $A$ be any event occuring with non-zero probability $1 - p$ (with respect to $\cQ$). Let $\cQ'$ be the conditional probability $\cQ$ given $A$, i.e., $\cQ'(\tilde x, \tilde y) = \Pr_{(x, y) \sim \cQ}[x = \tilde x \wedge y = \tilde y \mid A]$.
  For the game $\cG' = (X, Y, \cQ', \Sigma_X, \Sigma_Y, \{P_{x, y}\}_{(x, y)})$, we have $val(\cG) - p \leq val(\cG') \leq val(\cG) + 2p$.
\end{lemma}

\section{Birthday Repetition Theorem} \label{sec:birthday}
In this section, we prove our birthday repetition theorems. We first state our main theorems formally, starting with the birthday repetition theorem for general games.

\begin{theorem} \label{thm:birthday-general}
  There is a constant $\alpha > 0$ such that the following is true. Let $\cG = (X, Y, E, \Sigma_X, \Sigma_Y, \{P_{(x, y)}\})$ be any two-prover game. Let $d_{max}$ be the maximum degree of a vertex in the graph $(X, Y, E)$. Moreover, let $val(\cG) = 1 - \varepsilon$ and $c = \log |\Sigma_X||\Sigma_Y|$. For all $0 \leq k \leq |X|$ and $0 \leq l \leq |Y|$, we have
  \begin{align*}
    val(\cG^{k \times l}) \leq 2(1 - \varepsilon/2)^{\frac{\alpha \varepsilon^5 kl|E|}{d_{max}|X||Y|c^2}}
  \end{align*}
\end{theorem}

We note that, if the graph $(X, Y, E)$ is biregular, the exponent in the theorem is at most $\frac{\alpha \varepsilon^5 kl}{n c^2}$, as stated in Theorem~\ref{thm:main}. This is because, when the graph is biregular, either $|E| = |X|d_{max}$ or $|E| = |Y|d_{max}$.

For projection games, we can get a better dependency on $\varepsilon$ and get rid of the dependency on $c$ completely.

\begin{theorem} \label{thm:birthday-proj}
  There is a constant $\alpha > 0$ such that the following is true. Let $\cG = (X, Y, E, \Sigma_X, \Sigma_Y, \{P_{(x, y)}\})$ be any projection game. Let $d_{max}$ be the maximum degree of a vertex in the graph $(X, Y, E)$ and let $val(\cG) = 1 - \varepsilon$. For all $0 \leq k \leq |X|$ and $0 \leq l \leq |Y|$, we have
  \begin{align*}
    val(\cG^{k \times l}) \leq 2(1 - \varepsilon/2)^{\frac{\alpha \varepsilon^3 kl|E|}{d_{max}|X||Y|}}
  \end{align*}
\end{theorem}

We now prove the two theorems. Roughly speaking, we will show that $\cG^{k \times l}$ has small value by ``embedding'' an $\Omega\left(\frac{kl|E|}{d_{max}|X| |Y|}\right)$-tensor game, which has low value by the parallel repetition theorem, into it.

For convenience, let $s$ denote $\frac{kl|E|}{d_{max}|X| |Y|}$, the expected number of edges in $E(S, T)$ when $S$ and $T$ are independently uniformly sampled from $\binom{X}{k}$ and $\binom{Y}{l}$ respectively. Let $s_1$ and $s_2$ be $s(1 + \delta)$ and $s(1 - \delta)$ respectively for some $\delta \in [0, 1/2]$ that will be chosen later. Finally, we will use $r = \beta s / d_{max}$ rounds of parallel repetition where, again, $\beta \in [0, \delta/40]$ will be specified later. Lastly, let $E^r = \{((x_1, \dots, x_r), (y_1, \dots, y_r)) \mid (x_1, y_1), \dots, (x_r, y_r) \in E\}$. Note that the distribution of $\cG^{\otimes r}$ is uniform over $E^r$.
\begin{remark}
  $\delta$ and $\beta$ will be chosen based on $\varepsilon$, $c$ and whether $\cG$ is a projection game. When $\varepsilon$ and $c$ are constant, both $\delta$ and $\beta$ are small constants. This is the most representative case and is good to keep in mind when reading through the proof.
\end{remark}

Our overall strategy is to reduce $\cG^{\otimes r}$ to $\cG^{k \times l}$. Since $val(\cG^{\otimes r})$ is exponentially small in $r = \Omega\left(\frac{kl|E|}{d_{max}|X| |Y|}\right)$ due to the parallel repetition theorem, such reduction would give a similar upper bound on $val(\cG^{k \times l})$. Unfortunately, we do not know how to do this in one step so we will have to go through a sequence of reductions. The sequence of games that we reduce to are $\cG^{\otimes r}_{\text{set}}, \cG^{k \times l}_{\text{em}}, \cG^{k \times l}_{\text{em}, [s_1, s_2]}$ and $\cG^{k \times l}_{[s_1, s_2]}$ respectively. The game $\cG^{\otimes r}_{\text{set}}$ share the same questions, alphabet sets and predicates with $\cG^{\otimes r}$ while $\cG^{k \times l}_{\text{em}}, \cG^{k \times l}_{\text{em}, [s_1, s_2]}$ and $\cG^{k \times l}_{[s_1, s_2]}$ share those with $\cG^{k \times l}$. The distribution of each game is defined as follows.
\begin{itemize}
\item The distribution of $\cG^{\otimes r}_{\text{set}}$ is uniform over the set $E^r_{\text{set}}$ of all $((x_1, \dots, x_r), (y_1, \dots, y_r)) \in E^r$ such that $x_1, \dots, x_r, y_1, \dots, y_r$ are all distinct. Note that this distribution is simply $\cG^{\otimes r}$'s distribution conditioned on $x_1, \dots, x_r, y_1, \dots, y_r$ being all distinct.

\item We will try to make the distribution $\cQ^{k \times l}_{\text{em}}$ of $\cG^{k \times l}_{\text{em}}$ reflect an embedding of the game $\cG^{\otimes r}$. We define $\cQ^{k \times l}_{\text{em}}$ based on the following sampling process for $(S, T) \sim \cQ^{k \times l}_{\text{em}}$. First, sample $((x_1, \dots, x_r), (y_1, \dots, y_r))$ uniformly at random from $E^r_{\text{set}}$. Then, sample $\tilde S$ and $\tilde T$ independently uniformly from $\binom{X - \{x_1, \dots, x_r\}}{k - r}$ and $\binom{Y - \{y_1, \dots, y_r\}}{l - r}$ respectively. Finally, set $S = \{x_1, \dots, x_r\} \cup \tilde S$ and $T = \{y_1, \dots, y_r\} \cup \tilde T$.
\item The distribution $\cQ^{k \times l}_{\text{em}, [s_1, s_2]}$ of $\cG^{k \times l}_{\text{em}, [s_1, s_2]}$ is the distribution $\cQ^{k \times l}_{\text{em}}$ conditioned on the number of edges between the two sets being in the range $[s_1, s_2]$. In other words, $\cQ^{k \times l}_{\text{em}, [s_1, s_2]}(S, T) = \Pr_{(S', T') \sim \cQ^{k \times l}_{\text{em}}}[S = S' \wedge T = T' \mid s_1 \leq |E(S', T')| \leq s_2].$
\item Finally, the distribution of $\cG^{k \times l}_{[s_1, s_2]}$ is uniform over the set $E^{k \times l}_{[s_1, s_2]}$ of all $(S, T)$ such that $|E(S, T)| \in [s_1, s_2]$. In other words, we ignore weights in $\cQ^{k \times l}_{\text{em}, [s_1, s_2]}$ and use the uniform distribution over $\supp(\cQ^{k \times l}_{\text{em}, [s_1, s_2]})$.
\end{itemize}

We will next give intuitions on why $val(\cG^{\otimes r}) \approx val(\cG^{\otimes r}_{\text{set}}) \approx val(\cG^{k \times l}_{\text{em}}) \approx val(\cG^{k \times l}_{\text{em}, [s_1, s_2]}) \approx val(\cG^{k \times l}_{[s_1, s_2]}) \approx val(\cG^{k \times l})$ where each $\approx$ hides some multiplicative or additive losses in each step. With the right choice of $\delta$ and $\beta$, we can ensure that each loss is significantly smaller than $val(\cG^{\otimes r})$, and, thus, we will be able to bound $val(\cG^{k \times l})$. Below, we state these losses more precisely and summarize the overview of each proof.

\begin{lemma} \label{lem:no-col}
  $val(\cG^{\otimes r}_{\text{set}}) \leq \left(\frac{1}{1 - 2\beta}\right)^r \cdot val(\cG^{\otimes r})$
\end{lemma}

{\bf Proof Idea.} From Lemma~\ref{lem:inq-mult}, it is enough for us to lower bound the ratio $|E^r_{\text{set}}|/|E^r|$. This is simply the probability that $r$ random edges from $E$ do not share any endpoints, which is not hard to bound.

\begin{lemma} \label{lem:embedding}
  $val(\cG^{k \times l}_{\text{em}}) \leq val(\cG^{\otimes r}_{\text{set}})$
\end{lemma}

{\bf Proof Idea.} Based on how $\cQ^{k \times l}_{\text{em}}$ is defined, it induces a canonical map from each strategy in $\cG^{k \times l}_{\text{em}}$ to a ``mixed strategy'' in $\cG^{\otimes r}_{\text{set}}$. We can show that each strategy $\phi$ in $\cG^{k \times l}_{\text{em}}$ has value no more than the value of the mixed strategy in $\cG^{\otimes r}_{\text{set}}$ that $\phi$ maps to, which essentially proves the lemma.

\begin{lemma} \label{lem:concen-1}
  $val(\cG^{k \times l}_{\text{em}, [s_1, s_2]}) \leq val(\cG^{k \times l}_{\text{em}}) + 8\exp\left(-\frac{\delta^2 r}{432 \beta}\right)$
\end{lemma}

{\bf Proof Idea.} $\cQ^{k \times l}_{\text{em}, [s_1, s_2]}$ is $\cQ^{k \times l}_{\text{em}}$ conditioned on the event that $S$ and $T$ has between $s_1$ and $s_2$ edges among them. From Lemma~\ref{lem:inq-cond}, it is enough for us to bound a probability of such event. From the definition of $\cQ^{k \times l}_{\text{em}}$, $S$ and $T$ can be sampled by first sampling $x_1, \dots, x_r, y_1, \dots, y_r$ according to $E^r$ and then sampling the rest of $S$ and $T$ from $X - \{x_1, \dots, x_r\}$ and $Y - \{y_1, \dots, y_r\}$ respectively. When $r$ is small enough, we can show, with the help of Lemma~\ref{lem:random-num-edges}, that, for any $x_1, \dots, x_r, y_1, \dots, y_r$, the number of edges generated by $S$ and $T$ concentrates around $s$. This gives us the desired bound.

\begin{lemma} \label{lem:matching-num}
  $val(\cG^{k \times l}_{[s_1, s_2]}) \leq \left(\frac{1 + \delta}{1 - \delta - 2\beta}\right)^{2r} \cdot val(\cG^{k \times l}_{\text{em}, [s_1, s_2]})$
\end{lemma}

{\bf Proof Idea.} We want to evoke Lemma~\ref{lem:inq-mult} to arrive at the bound. To do so, we need to show that the two distributions are (multiplicatively) close. Since the distribution of $\cG^{k \times l}_{[s_1, s_2]}$ is uniform, we only need to show that the maximum probability and the minimum (non-zero) probability in $\cQ^{k \times l}_{\text{em}, [s_1, s_2]}$ are close.

Fortunately, we know that $\cQ^{k \times l}_{\text{em}, [s_1, s_2]}$ is $\cQ^{k \times l}_{\text{em}}$ conditioned on an event. This means that,  when $\cQ^{k \times l}_{\text{em}, [s_1, s_2]}(S, T)$ is not zero, it is proportional to $\cQ^{k \times l}_{\text{em}}(S, T)$. The latter, in turn, is proportional to the number of edges $(x_1, y_1), \dots, (x_r, y_r) \in E^r$ such that $x_1, \dots, x_r, y_1, \dots, y_r$ are all distinct and $x_1, \dots, x_r \in S$ and $y_1, \dots, y_r \in T$. In other words, we want to upper bound and lower bound the number of $r$ edges in $E(S, T)$ with distinct endpoints. This is feasible since we know that $|E(S, T)| \in [s_1, s_2]$ and $r$ is so small that with a reasonable probability $r$ edges picked will not share any endpoint with each other.

\begin{lemma} \label{lem:concen-2}
  $val(\cG^{k \times l}) \leq val(\cG^{k \times l}_{[s_1, s_2]}) + 4\exp\left(-\frac{\delta^2 r}{54 \beta}\right)$
\end{lemma}

{\bf Proof Idea.} By realising that $\cG^{k \times l}_{[s_1, s_2]}$'s distribution is simply $\cG^{k \times l}$'s distribution conditioned on $|E(S, T)| \in [s_1, s_2]$, this follows immediately from Lemma~\ref{lem:random-num-edges} and Lemma~\ref{lem:inq-cond}.

Before we give full proofs of the above lemmas, let us first show how they imply the birthday repetition theorems. To avoid repeating arguments for both general games and projection games, we show the following intermediate lemma. Since the proof of the lemma consists of basically only calculations, we defer the proof to Subsection~\ref{subsec:birthday-helper}.

\begin{lemma} \label{lem:birthday-helper}
  Let $\cG$ be any game of value $1 - \varepsilon$ and $k, l, \beta, \delta$ and $r$ be as defiend above. If $val(\cG^{\otimes r}) \leq (1 - \varepsilon/2)^R$ for some $R \geq 0$ such that $R \leq r$, $\delta \leq \frac{\varepsilon R}{200 r}$ and $R \leq \frac{\delta^2 r}{1000 \beta \varepsilon}$, then $val(\cG^{k \times l}) \leq 2(1 - \varepsilon/2)^{R/10}$.
\end{lemma}

The final ingredient we need to prove the birthday repetition theorem is the parallel repetition theorem. For general games, we use Holenstein's version of the theorem~\cite{Hol09}, which is an improvement over the original theorem of Raz~\cite{Raz98}.

\begin{theorem}\cite{Hol09} \label{thm:par-general}
  There exists a global constant $C \in (0, 1]$ such that, for any two-prover game $\cG = (X, Y, \cQ, \Sigma_X, \Sigma_Y, \{P_{(x, y)}\})$ of value $1 - \varepsilon$ and for every $k > 0$, we have $val(\cG^{\otimes k}) \leq (1 - \varepsilon/2)^{C \varepsilon^2 k / \log(|\Sigma_X||\Sigma_Y|)}.$
\end{theorem}

Equipped with Lemma~\ref{lem:birthday-helper} and the parallel repetition theorem, we can now prove our birthday repetition theorems just by selecting the right $\delta$ and $\beta$.

\begin{proofof}[Theorem~\ref{thm:birthday-general}]
  Pick $\delta = \frac{\varepsilon^3 C}{10^3c}$ and $\beta = \frac{\varepsilon^3 C}{10^{10}c}$ where $C$ is the constant from the parallel repetition theorem for general games (Theorem~\ref{thm:par-general}). From Theorem~\ref{thm:par-general}, we have $val(\cG^{\otimes r}) \leq (1 - \varepsilon/2)^{C \varepsilon^2 r / c}$.
  Let $R = C \varepsilon^2 r / c$. We can see that $R, \delta, \beta$ satisfy the conditions in Lemma~\ref{lem:birthday-helper}. Hence, we can conclude that
  \begin{align*}
    val(\cG^{k \times l}) \leq (1 - \varepsilon/2)^{R/10}
    = (1 - \varepsilon/2)^{(C^2/10^{11})\left(\frac{\varepsilon^5 kl|E|}{c^2|X||Y|d_{max}}\right)}.
  \end{align*}
  This completes the proof for Theorem~\ref{thm:birthday-general} with $\alpha = C^2/10^{11}$.
\end{proofof}

In the case of projection game, we can improve dependency on $\varepsilon$ and get rid of dependency on $c$ thanks to the stronger bound in Rao's parallel repetition for projection games~\cite{Rao11}.

\begin{theorem}\cite{Rao11} \label{thm:par-proj}
  There exists a global constant $C \in (0, 1]$ such that, for any projection game $\cG = (X, Y, \cQ, \Sigma_X, \Sigma_Y, \{P_{(x, y)}\})$ of value $1 - \varepsilon$ and for every $k > 0$, we have $val(\cG^{\otimes k}) \leq (1 - \varepsilon/2)^{C \varepsilon k}.$
\end{theorem}

\begin{proofof}[Theorem~\ref{thm:birthday-proj}]
  Pick $\delta = \frac{\varepsilon^2 C}{10^3}$ and $\beta = \frac{\varepsilon^2 C}{10^{10}}$ where $C$ is the constant from the parallel repetition theorem for projection games (Theorem~\ref{thm:par-proj}). From the theorem, we have $val(\cG^{\otimes r}) \leq (1 - \varepsilon/2)^{C \varepsilon r}$. Let $R = C \varepsilon r$. By evoking Lemma~\ref{lem:birthday-helper}, we have
  \begin{align*}
    val(\cG^{k \times l}) \leq (1 - \varepsilon/2)^{R/10}
    = (1 - \varepsilon/2)^{(C^2/10^{11})\left(\frac{\varepsilon^3 kl|E|}{|X||Y|d_{max}}\right)}.
  \end{align*}
\end{proofof}

We devote the rest of this section to the proofs of unproven lemmas. We present the proofs of Lemma~\ref{lem:no-col}, Lemma~\ref{lem:concen-1}, Lemma~\ref{lem:matching-num}, Lemma~\ref{lem:concen-2} and Lemma~\ref{lem:birthday-helper} in this order, one lemma per subsection.

\subsection{$\cG^{\otimes r}$ vs $\cG^{\otimes r}_{\text{set}}$: No Collision Probability} \label{subsec:no-col}

\begin{proofof}[Lemma~\ref{lem:no-col}]
  Observe that $|E^r_{\text{set}}|/|E^r| = \Pr_{(x_1, y_1), \dots, (x_r, y_r) \sim E}[x_1, \dots, x_r, y_1, \dots, y_r \text{ are all distinct}].$

  We can further rewrite this as
  \begin{align*}
    \prod_{i=1}^r \Pr_{(x_1, y_1), \dots, (x_i, y_i) \sim E}[x_i \notin \{x_1, \dots, x_{i - 1}\} \wedge y_i \notin \{y_1, \dots, y_{i - 1}\} \mid x_1, \dots, x_{i - 1}, y_1, \dots, y_{i - 1} \text{ are all distinct}].
  \end{align*}
  Since the maximum degree of $(X, Y, E)$ is $d_{max}$, the number of edges with at least one endpoint in $\{x_1, \dots, x_{i - 1}, y_1, \dots, y_{i - 1}\}$ is at most $2(i - 1)d_{max}$. Hence, the above expression is at least
  \begin{align*}
    \prod_{i=1}^r \left(\frac{|E| - 2(i - 1)d_{max}}{|E|}\right) \geq \left(1 - \frac{2rd_{max}}{|E|}\right)^r \geq (1 - 2\beta)^r.
  \end{align*}

  Finally, from Lemma~\ref{lem:inq-mult}, we have $val(\cG^{\otimes r}_{\text{set}}) \leq \left(\frac{1}{1 - 2\beta}\right)^r \cdot val(\cG^{\otimes r})$ as desired.
\end{proofof}

\subsection{$\cG^{\otimes r}_{\text{set}}$ vs $\cG^{k \times l}_{\text{em}}$: Embedding of Parallel Repetition} \label{subsec:embedding}

\begin{proofof}[Lemma~\ref{lem:embedding}]
  Let $\phi$ be the strategy on $\cG^{k \times l}_{\text{em}}$ such that $val(\phi) = val(\cG^{k \times l}_{\text{em}})$. We will create a \emph{mixed strategy}, a distribution $\Phi$ of strategies, for $\cG^{\otimes r}$ such that the expected value of a strategy drawn from this distribution is at least $val(\phi)$. Once we have such a mixed strategy $\Phi$, at least one strategy in $\supp(\Phi)$ must have value at least $val(\phi)$. This implies that $val(\cG^{\otimes r}) \geq val(\phi) = val(\cG^{k \times l}_{\text{em}})$, proving the lemma.

  We define $\Phi$ by a random process that generates $\tilde\phi \sim \Phi$ as follows. For each $(x_1, \dots, x_r) \in X^r$ such that $x_1, \dots, x_r$ are distinct, sample a set $S$ uniformly at random among all the subsets of $X$ of size $k$ that contain $x_1, \dots, x_r$. We then set $\tilde\phi(x_1, \dots, x_r) = \phi(S)$. Similarly, for each $(y_1, \dots, y_r) \in Y^r$ such that $y_1, \dots, y_r$ are distinct, sample $T$ randomly from all subsets of $Y$ of size $l$ that contain $y_1, \dots, y_r$ and set $\tilde\phi(y_1, \dots, y_r) = \phi(T)$.  We will next show that the expected value of $\tilde\phi$ sampled in such way is at least $val(\cG^{k \times l}_{\text{em}})$.

  The expected value of $\tilde\phi$ can be written as follows.
  \begin{align*}
    \E_{\tilde\phi}[val(\tilde\phi)] &= \E_{\tilde\phi}\left[\E_{((x_1, \dots, x_r), (y_1, \dots, y_r)) \sim E^r_{\text{set}}}[P^r_{((x_1, \dots, x_r), (y_1, \dots, y_r))}(\tilde\phi(x_1, \dots, x_r), \tilde\phi(y_1, \dots, y_r))]\right] \\
    &= \E_{((x_1, \dots, x_r), (y_1, \dots, y_r)) \sim E^r_{\text{set}}}\left[\E_{\tilde\phi}[P^r_{((x_1, \dots, x_r), (y_1, \dots, y_r))}(\tilde\phi(x_1, \dots, x_r), \tilde\phi(y_1, \dots, y_r))]\right] \\
    &= \E_{((x_1, \dots, x_r), (y_1, \dots, y_r)) \sim E^r_{\text{set}}}\left[\E_{S, T}[P^r_{((x_1, \dots, x_r), (y_1, \dots, y_r))}(\phi(S), \phi(T))]\right].
  \end{align*}
  Note that $S, T$ in the last expression is sampled depending on $x_1, \dots, x_r, y_1, \dots, y_r$ as described above. Now, observe that $P^r_{((x_1, \dots, x_r), (y_1, \dots, y_r))}$ is no more than $P^{k \times l}_{(S, T)}$ on any input because $P^r_{((x_1, \dots, x_r), (y_1, \dots, y_r))}$ verifies the original game's predicates on $(x_1, y_1), \dots, (x_r, y_r)$ whereas $P^{k \times l}_{(S, T)}$ verifies every $(x, y) \in E \cap (S \times T)$, including $(x_1, y_1), \dots, (x_r, y_r)$. Based on this and the definition of $\cQ^{k \times l}_{\text{em}}$, we have
  \begin{align*}
    \E_{\tilde\phi}[val(\tilde\phi)] \leq \E_{((x_1, \dots, x_r), (y_1, \dots, y_r)) \sim E^r_{\text{set}}}\left[\E_{S, T}[P^{k \times l}_{(S, T)}(\phi(S), \phi(T))]\right]
    = \E_{S, T \sim \cQ^{k \times l}_{\text{embeded}}}[P^{k \times l}_{(S, T)}(\phi(S), \phi(T))]
    = val(\phi).
  \end{align*}
  Thus, we have completed the proof of the lemma.
\end{proofof}

\subsection{$\cG^{k \times l}_{\text{em}}$ vs $\cG^{k \times l}_{\text{em}, [s_1, s_2]}$: Concentration on Number of Edges} \label{subsec:concen-1}

\begin{proofof}[Lemma~\ref{lem:concen-1}]
  Recall that the distribution $\cQ^{k \times l}_{\text{em}, [s_1, s_2]}$ is $\cQ^{k \times l}_{\text{em}}$ conditioned on the number of edges in $E(S, T)$ is between $s_1$ and $s_2$. Let $A$ denote such event. We would like to bound $\Pr_{(S, T) \sim \cQ^{k \times l}_{\text{em}}}[A]$.

  Recall that $S, T \sim \cQ^{k \times l}_{\text{em}}$ comes from sampling $((x_1, \dots, x_r), (y_1, \dots, y_r)) \in E^r$, $\tilde S \in \binom{X - \{x_1, \dots, x_r\}}{k - r}, \tilde T \in \binom{Y - \{y_1, \dots, y_r\}}{l - r}$ and set $S = \tilde S \cup \{x_1, \dots, x_r\}$ and $T = \tilde T \cup \{y_1, \dots, y_r\}$. Fix $x_1, \dots, x_r, y_1, \dots, y_r$. Let $\tilde X = X - \{x_1, \dots, x_r\}$ and $\tilde Y = Y - \{y_1, \dots, y_r\}$. Observe that there are at most $2rd_{max} = 2\beta s$ edges with at least one endpoint in $\{x_1, \dots, x_r, y_1, \dots, y_r\}$. This has two consequences. First, we have $|E(\tilde X, \tilde Y)| \geq |E| - 2\beta s$. Second, if the number of edges in $E(\tilde S, \tilde T)$ lies in $[(1 - \delta + 2\beta)s, (1 + \delta - 2\beta)s]$, then $A$ occurs. Thus, to bound $\Pr_{(S, T) \sim \cQ^{k \times l}_{\text{em}}}[A]$, it is enough to bound the probability that $E(\tilde S, \tilde T) \in [(1 - \delta + 2\beta)s, (1 + \delta - 2\beta)s]$.

  Let $\tilde s = \frac{(k - r)(l - r)|E(\tilde X, \tilde Y)|}{(|X| - r)(|Y| - r)}$ and let $\gamma = \delta - 20\beta$. From Lemma~\ref{lem:random-num-edges}, we have
  \begin{align*}
    \Pr_{\tilde S \sim \binom{\tilde X}{k - r}, \tilde T \sim \binom{\tilde Y}{l - r}}[|E(\tilde S, \tilde T)| \in [(1 - \gamma){\tilde s}, (1 + \gamma){\tilde s}]] \geq 1 - 4\exp\left(- \frac{\gamma^2 \tilde s}{54 d_{max}}\right)
  \end{align*}

  Next, we will bound $\tilde s$. 
  It is easy to see that $r \leq \beta k, \beta l$ and $s \leq |E|$, which gives the following bound.
  \begin{align*}
    \tilde s &\geq \frac{(k - \beta k)(l - \beta l)(|E| - 2 \beta s)}{|X||Y|} \geq (1 - \beta)^2(1 - 2\beta) s \geq (1 - 4\beta) s.
  \end{align*}

  The above inequality also implies that $(1 - \gamma){\tilde s} \geq (1 - \gamma - 4\beta) s \geq (1 - \delta + 2\beta) s$.

  On the other hand, since $r \leq \beta |X|, \beta |Y|$, we have
  \begin{align*}
    (1 + \gamma)\tilde s &\leq \frac{(1 + \gamma)kl|E|}{(|X| - \beta |X|)(|Y| - \beta |Y|)} = (1 + \gamma)\left(\frac{1}{1 - \beta}\right)^2s \leq (1 + \delta - 2\beta)s
  \end{align*}
  where the last inequality comes from $\beta \leq 1/2, \gamma = \delta - 20\beta$ and $\gamma \leq 1$.

  As a result, we have
  \begin{align*}
    \Pr_{\tilde S \sim \binom{\tilde X}{k - r}, \tilde T \sim \binom{\tilde Y}{l - r}}[|E(\tilde S, \tilde T)| \in [(1 - \delta + 2\beta)s, (1 + \delta - 2\beta)s]] &\geq 1 - 4\exp\left(- \frac{\gamma^2 \tilde s}{54 d_{max}}\right) \\
    &\geq 1 - 4\exp\left(- \frac{\delta^2s}{432 d_{max}}\right) \\
    &= 1 - 4\exp\left(- \frac{\delta^2 r}{432 \beta}\right)
  \end{align*}
  where the second inequality comes from $\gamma \geq \delta / 2$ and ${\tilde s} \geq (1 - 6\beta)s \geq s/2$.

  Now, we are ready to bound the probability that event $A$ occurs.
  \begin{align*}
    \Pr_{(S, T) \sim \cQ^{k \times l}_{\text{em}}}[A] &= \Pr_{((x_1, \dots, x_r), (y_1, \dots, y_r)) \sim E^r_{\text{set}}}\left[\Pr_{\tilde S \sim \binom{\tilde X}{k - r}, \tilde T \sim \binom{\tilde Y}{l - r}}[A]\right] \\
    &\geq \Pr_{((x_1, \dots, x_r), (y_1, \dots, y_r)) \sim E^r_{\text{set}}}\left[\Pr_{\tilde S \sim \binom{\tilde X}{k - r}, \tilde T \sim \binom{\tilde Y}{l - r}}[|E(\tilde S, \tilde T)| \in [(1 - \delta + 2\beta)s, (1 + \delta - 2\beta)s]]\right] \\
    &\geq 1 - 4\exp\left(- \frac{\delta^2 r}{432 \beta}\right).
  \end{align*}

  Finally, Lemma~\ref{lem:inq-cond} immmediately yields $val(\cG^{k \times l}_{\text{em},[s_1, s_2]}) \leq val(\cG^{k \times l}_{\text{em}}) + 8\exp\left(- \frac{\delta^2 r}{432 \beta}\right)$ as desired.
\end{proofof}

\subsection{$\cG^{k \times l}_{\text{em}, [s_1, s_2]}$ vs $\cG^{k \times l}_{[s_1, s_2]}$: Number of Embeddings to Each Set} \label{subsec:matching-num}

\begin{proofof}[Lemma~\ref{lem:matching-num}]
  Our goal is to show that $\left(\frac{1 + \delta}{1 - \delta - 2\beta}\right)^{r}\cQ^{k \times l}_{\text{em}, [s_1, s_2]}(S, T) \geq 1/|E^{k \times l}_{[s_1, s_2]}|$ for every $(S, T) \in E^{k \times l}_{[s_1, s_2]}$. This together with Lemma~\ref{lem:inq-mult} immediately implies the lemma. To show this, we will first argue that $\frac{\cQ^{k \times l}_{\text{em}, [s_1, s_2]}(\tilde S, \tilde T)}{\cQ^{k \times l}_{\text{em}, [s_1, s_2]}(S, T)} \leq \left(\frac{1 + \delta}{1 - \delta - 2\beta}\right)^{r}$ for every $(S, T), (\tilde S, \tilde T) \in E^{k \times l}_{[s_1, s_2]}$.

  Since $\cQ^{k \times l}_{\text{em}, [s_1, s_2]}$ is just $\cQ^{k \times l}_{\text{em}}$ conditioned on $|E(S, T)| \in [s_1, s_2]$, we have
  \begin{align*}
    \frac{\cQ^{k \times l}_{\text{em}, [s_1, s_2]}(\tilde S, \tilde T)}{\cQ^{k \times l}_{\text{em}, [s_1, s_2]}(S, T)}
    = \frac{\cQ^{k \times l}_{\text{em}}(\tilde S, \tilde T)}{\cQ^{k \times l}_{\text{em}}(S, T)}
    = \frac{|\{((x_1, \dots, x_r), (y_1, \dots, y_r)) \in E^r \mid x_1, \dots, x_r \in \tilde S, y_1, \dots, y_r \in \tilde T\}|}{|\{((x_1, \dots, x_r), (y_1, \dots, y_r)) \in E^r \mid x_1, \dots, x_r \in S, y_1, \dots, y_r \in T\}|}
  \end{align*}
  where the latter equality comes from rearranging the definition of $\cQ^{k \times l}_{\text{em}}$.

  Recall that $((x_1, \dots, x_r), (y_1, \dots, y_r)) \in E^r$ iff $(x_1, y_1), \dots, (x_r, y_r) \in E$ and $x_1, \dots, x_r, y_1, \dots, y_r$ are distinct. Since $(\tilde S, \tilde T) \in E^{k \times l}_{[s_1, s_2]}$, $E(\tilde S, \tilde T) \leq s_1$. Hence, there are at most $s_1$ choices for each $(x_i, y_i)$. Thus,
  \begin{align*}
    |\{((x_1, \dots, x_r), (y_1, \dots, y_r)) \in E^r \mid x_1, \dots, x_r \in \tilde S, y_1, \dots, y_r \in \tilde T\}| \leq (s_1)^r.
  \end{align*}

  On the other hand, since $(S, T) \in E^{k \times l}_{[s_1, s_2]}$, $|E(S, T)| \geq s_2$. We want to lower bound the number of $(x_1, y_1),$ $\dots, (x_r, y_r) \in E$ with $x_1, \dots, x_r \in S$ and $y_1, \dots, y_r \in T$ such that $x_1, \dots, x_r, y_1, \dots, y_r$ are distinct. Let us pick $(x_1, y_1), \dots, (x_r, y_r) \in E(S, T)$ in this order and ensure the distinctness property in each step.

  Suppose that we have already picked $(x_1, y_1), \dots, (x_{i - 1}, y_{i - 1})$. We can pick any edge $(x_i, y_i)$ in $E(S, T)$ as long as its endpoints are not in $\{x_1, \dots, x_{i - 1}, y_1, \dots, y_{i - 1}\}$. Since the maximum degree in the graph $(X, Y, E)$ is $d_{max}$, the number of prohibited edges is at most $2(i - 1)d_{max}$. As a result, there are at least $s_2 - 2(i - 1)d_{max} \geq s_2 - 2rd_{max}$ valid choices for $(x_i, y_i)$. Thus, we have
  \begin{align*}
    |\{((x_1, \dots, x_r), (y_1, \dots, y_r)) \in E^r \mid x_1, \dots, x_r \in S, y_1, \dots, y_r \in T\} \geq (s_2 - 2rd_{max})^r.
  \end{align*}

  This implies that
  \begin{align*}
    \frac{\cQ^{k \times l}_{\text{em}, [s_1, s_2]}(\tilde S, \tilde T)}{\cQ^{k \times l}_{\text{em}, [s_1, s_2]}(S, T)} \leq \left(\frac{s_1}{s_2 - 2rd_{max}}\right)^r
    = \left(\frac{(1 + \delta)s}{(1 - \delta)s - 2(\beta s / d_{max})d_{max}}\right)^r
    = \left(\frac{1 + \delta}{1 - \delta - 2\beta}\right)^r.
  \end{align*}

  Finally, let $(S', T')$ be the element of $\supp(\cQ^{k \times l}_{\text{em}, [s_1, s_2]})$ with maximum $\cQ^{k \times l}_{\text{em}, [s_1, s_2]}(S', T')$. Due to our choice of $(S', T')$, we have $\cQ^{k \times l}_{\text{em}, [s_1, s_2]}(S', T') \geq \frac{1}{|\supp(\cQ^{k \times l}_{\text{em}, [s_1, s_2]})|} = \frac{1}{|E^{k \times l}_{[s_1, s_2]}|}$. Since $\frac{\cQ^{k \times l}_{\text{em}, [s_1, s_2]}(S', T')}{\cQ^{k \times l}_{\text{em}, [s_1, s_2]}(S, T)} \leq \left(\frac{1 + \delta}{1 - \delta - 2\beta}\right)^r$ for every $S, T \in E^{k \times l}_{[s_1, s_2]}$, we have completed the proof of this lemma.
\end{proofof}

\subsection{$\cG^{k \times l}_{[s_1, s_2]}$ vs $\cG^{k \times l}$: Concentration on Number of Edges} \label{subsec:concen-2}

\begin{proofof}[Lemma~\ref{lem:concen-2}]
  Observe that the distribution of $\cG^{k \times l}_{[s_1, s_2]}$ is simply the distribution $\cQ^{k \times l}$ of $\cG^{k \times l}$ conditioned on $|E(S, T)|$ is between $s_1$ and $s_2$. Let us call this event $A$. We will bound the probability of $A$ (with respect to $\cQ^{k \times l}$). Once we do so, we can apply Lemma~\ref{lem:inq-cond} to complete the proof.

  The probability that $A$ happens is simply the probability that $S \sim \binom{X}{k}$ and $T \sim \binom{Y}{l}$ have between $s_1 = (1 - \delta)s$ and $s_2 = (1 + \delta)s$ edges between them. Lemma~\ref{lem:random-num-edges} immediately gives the following bound.
  \begin{align*}
    \Pr_{(S, T) \sim \cQ^{k \times l}}[A] \geq 1 - 4\exp\left(-\frac{\delta^2 s}{54d_{max}}\right)
    &= 1 - 4\exp\left(-\frac{\delta^2 r}{54 \beta}\right)
  \end{align*}

  Hence, by Lemma~\ref{lem:inq-cond}, we can conclude that $val(\cG^{k \times l}) \leq val(\cG^{k \times l}_{[s_1, s_2]}) + 4\exp\left(-\frac{\delta^2 r}{54 \beta}\right)$ as desired.
\end{proofof}

\subsection{Proof of the Parameter Selection Lemma} \label{subsec:birthday-helper}

To prove Lemma~\ref{lem:birthday-helper}, we will use the following two well-know bounds.

\begin{lemma}[Bernoulli's inequality] \label{lem:inq-bernoulli}
  For any real number $r \geq 1$ and $x \geq -1$, $(1 + x)^r \geq 1 + rx.$
\end{lemma}

\begin{fact} \label{fact:inq-exp-to-linear}
  For any real number $0 \leq x \leq 1$, $\exp(-x) \leq (1 - x/2)$.
\end{fact}

\begin{proofof}[Lemma~\ref{lem:birthday-helper}]
  From Lemma~\ref{lem:no-col} and Lemma~\ref{lem:embedding}, we have $val(\cG^{k \times l}_{\text{em}}) \leq val(\cG^{\otimes r}_{\text{set}}) \leq \left(\frac{1}{1 - 2\beta}\right)^r (1 - \varepsilon/2)^{R}$. We can use Bernoulli's inequality to bound the right hand side term as follows.
  \begin{align*}
    \left(\frac{1}{1 - 2\beta}\right)^r (1 - \varepsilon/2)^{R}
    &\leq \left(\frac{1}{1 - 20 \beta r / R}\right)^{R/10} (1 - \varepsilon/2)^R 
    \leq (1 - \varepsilon/2)^{9R/10}
  \end{align*}
  Note that the second inequality comes from $\beta \leq \delta \leq \frac{\varepsilon R}{200 r}$.

  Thus, from Lemma~\ref{lem:concen-1}, we have
  \begin{align*}
    val(\cG^{k \times l}_{\text{em}, [s_1, s_2]}) &\leq (1 - \varepsilon/2)^{9R/10} + 8\exp\left(-\frac{\delta^2 r}{432 \beta}\right) \\
    (\text{Since } R \leq \frac{\delta^2 r}{1000 \beta \varepsilon}) &\leq (1 - \varepsilon/2)^{9R/10} + 8\left(\exp\left(-\varepsilon\right)\right)^{R} \\
    (\text{From Fact~\ref{fact:inq-exp-to-linear}}) &\leq 9(1 - \varepsilon/2)^{9R/10}.
  \end{align*}

  Now, from Lemma~\ref{lem:matching-num}, we have $val(\cG^{k \times l}_{[s_1, s_2]}) \leq 9 \left(\frac{1 + \delta}{1 - \delta - 2\beta}\right)^{2r}(1 - \varepsilon/2)^{9R/10}$.

  Again, the right hand side can be further bounded as follows.
  \begin{align*}
    9 \left(\frac{1 + \delta}{1 - \delta - 2\beta}\right)^{2r}(1 - \varepsilon/2)^{9R/10}
    &\leq 9 \left(\frac{1}{1 - 2\delta - 2\beta}\right)^{2r}(1 - \varepsilon/2)^{9R/10} \\
    (\text{Bernoulli's inequality}) &\leq 9\left(\frac{1}{(1 - (20r/R)(2\delta + 2\beta))}\right)^{R/10}(1 - \varepsilon/2)^{9R/10} \\
    (\text{Since } \beta \leq \delta \leq \frac{\varepsilon R}{200 r}) &\leq 9\left(\frac{1}{(1 - \varepsilon/2)}\right)^{R/10}(1 - \varepsilon/2)^{9R/10}
    = 9(1 - \varepsilon/2)^{8R/10}.
  \end{align*}

  From Lemma~\ref{lem:concen-2}, we get the following bound for $val(\cG^{k \times l})$
  \begin{align*}
    val(\cG^{k \times l}) &\leq 9(1 - \varepsilon/2)^{8R/10} + 4\exp\left(-\frac{\delta^2 r}{54 \beta}\right) \\
    (\text{Since } R \leq \frac{\delta^2 r}{1000 \beta \varepsilon}) &\leq 9(1 - \varepsilon/2)^{8R/10} + 4\left(\exp\left(-\varepsilon\right)\right)^{R} \\
    (\text{From Fact~\ref{fact:inq-exp-to-linear}}) &\leq 13(1 - \varepsilon/2)^{8R/10} 
  \end{align*}

  Finally, note that, if $2(1 - \varepsilon/2)^{R/10} \geq 1$, then $val(\cG^{k \times l}) \leq 1 \leq 2(1 - \varepsilon/2)^{R/10}$. Otherwise, if $2(1 - \varepsilon/2)^{R/10} \leq 1$, we also have $val(\cG^{k \times l}) \leq 13(1 - \varepsilon/2)^{8R/10} \leq (2(1 - \varepsilon/2)^{R/10})^8 \leq 2(1 - \varepsilon/2)^{R/10}$ as desired.
\end{proofof}

\section{Applications of the Birthday Repetition Theorem} \label{sec:app}

In this section, we prove several implications of our birthday repetition theorem, including hardness of approximation results and integrality gaps for dense CSPs and improved $\AM(2)$ protocol for {\sc 3SAT}.

\subsection{Lower Bounds for Fully-Dense CSPs}

Before we prove inapproximabilities and integrality gap of dense {\sc Max $k$-CSP}, we will first describe a reduction from two-prover games to fully-dense {\sc Max $k$-CSP}, which is central to all the results presented here.

\subsubsection{Reduction from Two-Prover Games to Fully-Dense {\sc Max $k$-CSP}}

It is possible to prove inapproximability of {\sc Max $k$-CSP} by first reducing a two-prover game to a free game via birthday repetition and then reducing it to {\sc Max $k$-CSP}. However, this does not result in the best possible dependency on $k$. To demonstrate this, recall that, in the $(l \times l)$-birthday repetition game $\cG^{l \times l}$, each variable corresponds to a set of $l$ variables of $\cG$. The guarantee in our birthday repetition theorem is that $val(\cG^{l \times l})$ decays exponentially in the number of edges between two of these sets, which is $\Theta(l^2/n)$ in expectation.

Now, let us reduce a free game to {\sc Max $k$-CSP} by letting variables in {\sc Max $k$-CSP} be the same as in the free game and the predicates be the naturally induced constraints. It is not hard to see that, if the value of the free game is $\delta$, then the value of the {\sc Max $k$-CSP} instance is at most $\delta^{\Omega(k)}$. It is also easy to see that, if we do not exploit any particular structure of the free game, this is the best upper bound one can hope for. Thus, with this approach, the hardness gap we get is exponential in $\Theta(kl^2/n)$.

Unfortunately, this is not the right gap; each variable in the resulting {\sc Max $k$-CSP} instance is a set of $l$ variables of the original game $\cG$, which means that, roughly speaking, each constraint of the {\sc Max $k$-CSP} instance contains $\Theta(k^2l^2/n^2)$ constraints from $\cG$ in expectation. Hence, intuitively, we should expect the value of the {\sc Max $k$-CSP} instance to decay exponentially in $\Theta(k^2l^2/n^2)$ instead of $\Theta(kl^2/n)$.

To allow us to prove a sharper bound for the value of the {\sc Max $k$-CSP}, we define the following reduction from two-prover game directly to {\sc Max $k$-CSP}.
\begin{definition} \label{def:k-csp-red}
Given a two-prover game $\cG = (X, Y, \cQ, \Sigma_X, \Sigma_Y, \{P_{(x, y)}\})$ and an integer $l \leq |X|, |Y|$, a fully-dense {\sc Max $k$-CSP} instance $\cG^{l}_{k} = (V', (V')^k, \{P'_S\})$ is defined as follows. The variables $V'$ is $\binom{X}{l} \times \binom{Y}{l}$, i.e., each variable is a tuple $(S, T)$ of a set $S$ containing $l$ questions from $X$ and a set $T$ containing $l$ questions from $Y$. The alphabet set is $\Sigma_X^l \times \Sigma_Y^l$ and each element is associated with an assignment to $S \cup T$. Finally, the predicate is defined in a natural way, i.e., $P'_{((S_1, T_1), \dots, (S_k, T_k))}(\phi_1, \dots, \phi_k) = 1$ if and only if $\phi_1, \dots, \phi_k$ are consistent and $P_{(x, y)}(\phi_1 \circ \cdots \circ \phi_k(x), \phi_1 \circ \cdots \circ \phi_k(y)) = 1$ for all $(x, y) \in ((S_1 \cup \cdots \cup S_k) \times (T_1 \cup \cdots \cup T_k)) \cap \supp(\cQ)$.
\end{definition}

We can then show that our intuition is indeed correct:

\begin{lemma} \label{lem:k-csp-value}
  There is a constant $\gamma > 0$ such that the following is true. Let $\cG = (X, Y, E, \Sigma_X, \Sigma_Y, \{P_{(x, y)}\})$ be any projection game. Let $d_{max}$ be the maximum degree of a vertex in the graph $(X, Y, E)$. Moreover, let $val(\cG) = 1 - \varepsilon$. For all $k \geq 2$ and $l \geq 0$ such that $kl \leq |X|, |Y|$, we have
  \begin{align*}
    val(\cG^l_k) \leq 2(1 - \varepsilon/2)^{\frac{\gamma \varepsilon^3 k^2l^2|E|}{d_{max}|X||Y|}}
  \end{align*}
\end{lemma}

The proof of Lemma~\ref{lem:k-csp-value} is by a reduction from our birthday repetition theorem and is deferred to Appendix~\ref{app:k-csp-soundness}. Note that a similar bound holds even when $\cG$ is not a projection game; however, since it is not needed for our purposes, we do not state it here. We will next use the lemma to prove lower bounds for dense CSPs.

\subsubsection{ETH-Based Hardness of Approximation of Fully-Dense {\sc Max $k$-CSP}}

The first application of the birthday repetition theorem we present is an ETH-based almost-polynomial ratio hardness for fully-dense {\sc Max $k$-CSP}, as stated formally below.

\begin{lemma} \label{thm:hardness-approx}
  If ETH is true, for any $k \geq 2$, there is no polynomial-time algorithm that, given any fully-dense {\sc Max $k$-CSP} instance $\cG$ of size $N$, can distinguish $val(\cG) = 1$ from $val(\cG) \leq (1/2)^{\tilde \Omega(\log N)}$.
\end{lemma}

We prove this by essentially applying Lemma~\ref{lem:k-csp-value} with $l = \tilde \Omega(n) / k$ to a two-prover game produced by the PCP Theorem. We start by stating a PCP Theorem; here we use the version proved by Dinur~\cite{Dinur07}.

\begin{theorem}(Dinur's PCP Theorem~\cite{Dinur07}) \label{dinur-pcp}
  Given a {\sc 3SAT} instance $\phi$ of size $n$, there is a polynomial-time reduction that produces a projection game $\cG_{\phi}$ of size $n \polylog n$ with the following properties.
  \begin{itemize}
  \item (Completeness) If $val(\phi) = 1$, then $val(\cG_{\phi}) = 1$.
  \item (Soundness) If $val(\phi) < 1$, then $val(\cG_{\phi}) \leq 1 - \varepsilon$ for some constant $\varepsilon > 0$ not depending on $\phi$.
  \item (Bounded Degree) Each variable in $\cG_{\phi}$ has constant degree.
  \item (Bounded Alphabet Size) $\cG_{\phi}$ has constant alphabet size.
  \end{itemize}
\end{theorem}

\begin{remark}
  Dinur's original reduction is from {\sc 3SAT} to {\sc Max $2$-CSP}, not a projection game, and the reduced instance need not have bounded degree. The former can be fixed by a well-known reduction from any {\sc Max $k$-CSP} to a projection game, which will also be described later on in Definition~\ref{def:clause-variable}. The bounded degree property can be ensured by the ``Preprocessing Lemma'' (Lemma 1.7) from Dinur's paper~\cite{Dinur07}.
\end{remark}

We use Dinur's PCP Theorem because the length of the PCP is crucial to the resulting ratio in the hardness result. In particular, Dinur's PCP is the shortest PCP with constant query and alphabet size.

\begin{proofof}[Lemma~\ref{thm:hardness-approx}]
  Given {\sc 3SAT} instance $\phi$ of size $n$. We first use Dinur's PCP Theorem (Theorem~\ref{dinur-pcp}) to reduce $\phi$ to $\cG$ with $n' = n \polylog n$ variables, $q' = O(1)$ alphabet size and maximum degree $d' = O(1)$. Consider the fully-dense {\sc Max $k$-CSP} $\cG^l_k$ from Definition~\ref{def:k-csp-red} with $l = n/(k \log^2 n)$.

  Let $\tilde n$ and $\tilde q$ be the number of variables and the alphabet size of $\cG^l_k$. We have $\tilde n \leq \binom{n'}{l}^2 \leq 2(n')^{2l} \leq 2^{O\left(\frac{n}{k \log n}\right)}$ and $\tilde q \leq (q')^{2l} \leq 2^{O\left(\frac{n}{k \log n}\right)}$. Hence, the size of $\cG^l_k$ is $\tilde N = (\tilde n \tilde q)^k \leq 2^{O(n / \log n)}$. We next analyze the completeness and soundness of the reduction.

  When $val(\phi) = 1$, from the PCP theorem, we have $val(\cG) = 1$. It is also obvious from the reduction that $val(\cG^l_k)$ is one. On the other hand, when $val(\phi) < 1$, we have $val(\cG) \leq 1 - \varepsilon$. Hence, by Lemma~\ref{lem:k-csp-value}, we have $$val(\cG^l_k) \leq 2(1 - \varepsilon/2)^{\Omega\left(\frac{\varepsilon^3 k^2l^2 (n')}{(d')(n')^2}\right)} \leq (1/2)^{\tilde \Omega(n)} = (1/2)^{\tilde \Omega(\log \tilde N)}.$$

  Thus, if a algorithm can distinguish $val(\cG^l_k) = 1$ from $val(\cG^l_k) \leq (1/2)^{\tilde \Omega(\log \tilde N)}$ in time polynomial in $\tilde N$, then it can also solve {\sc 3SAT} in time $2^{O(n / \log n)}$ time, contradicting with ETH.
\end{proofof}

\subsubsection{Improved Hardness of Approximation Result Based on ETHA}

The $\polyloglog N$ loss in the exponent of Lemma~\ref{thm:hardness-approx} is due to the quasi-linear size of the PCP and can be eliminated if we instead assume the stronger ETHA:

\begin{lemma} \label{thm:hardness-linear}
  If ETHA holds, for any $k \geq 2$ and any sufficiently large $i$, there is no algorithm that can, given any fully-dense {\sc Max $k$-CSP} $\cG$ of size $N$, distinguish $val(\cG) = 1$ from $val(\cG) \leq 1/N^{1/i}$ in time $O(N^{\tilde \Omega(i)/\log^2 k})$.
\end{lemma}

The proof of lemma~\ref{thm:hardness-linear} proceeds in two steps. First, using a known reduction, we reduce a {\sc 3SAT} instance to a projection game of linear size. This step is the difference between ETHA and ETH; with ETHA, approximating a projection game to within some constant ratio takes exponential time, which cannot be derived from ETH since no linear-size constant-query PCP is known. The second step is essentially the same as the proof of Lemma~\ref{thm:hardness-approx} except that, since the size of our starting projection game is linear, we can now use birthday repetition for $l = \Theta_{i, k}(n)$ instead of $\Theta_{i, k}(n / \polylog n)$.

The reduction we use in the first step is the so-called clause/variable reduction. The reduction is well-known and has appeared in literatures before (in e.g.~\cite{AIM}). The reduction can be stated as follows.

\begin{definition}(Clause/variable game) \label{def:clause-variable}
  For any {\sc Max $k$-CSP} instance $\cG = (V, \cQ, \{P_S\})$ such that $\cQ$ is uniform over $\supp(\cQ)$, its clause/variable game is a projection game $\cG' = (X', Y', \Sigma_X', \Sigma_Y', E', \{P'_{(x, y)}\})$ defined as follows. $X'$ is the set of constraints of $\cG$, i.e., $X' = \supp(\cQ)$. $Y'$ is $V$, the set of variables of $\cG$. $\Sigma_X'$ is $\Sigma^k$; for each constraint $S$, $\Sigma_X'$ is identified with the assignments of $S$ in $\cG$. $\Sigma_Y'$ is simply $\Sigma$. Finally, $E'$ contains all $(S, x)$ such that $x \in S$ and $P_{(S, x)}(\phi, \sigma) = 1$ iff $P_S(\phi) = 1$ and $\phi(x) = \sigma$.
\end{definition}

It is easy to see that, when $val(\cG)$ is bounded away from one, then so is $val(\cG')$, as we argue below.

\begin{proposition} \label{prop:clause-var-val}
  Let $\cG$ and $\cG'$ be as in Definition~\ref{def:clause-variable}. If $val(\cG) \leq 1 - \varepsilon$, then $val(\cG') \leq 1 - \varepsilon / k$.
\end{proposition}

\begin{proof}
  Suppose for the sake of contradiction that there is an assignment $\phi'$ of $\cG'$ such that $val(\phi') > 1 - \varepsilon/k$. Define $\phi: V \to \Sigma$ by $\phi(x) = \phi'(x)$ for every $x \in V$. Since less than $\varepsilon/k$ fraction of the edges are not satisfied by $\phi'$ in $\cG'$ and each $S \in X'$ has only $k$ edges touching it, more than $1 - \varepsilon$ fraction of $S \in X'$ touches only satisfied edges. These clauses are satisfied by $\phi$ in $\cG$. Hence, $val(\phi) > 1 - \varepsilon$, which is a contradiction.
\end{proof}

We will now prove Lemma~\ref{thm:hardness-linear}. We also need a tighter bound for the binomial coefficient $\binom{n}{k}$, which is stated below. The bound can be easily derived from Stirling-type inequalities.

\begin{fact} \label{fact:binom-approx}
  For every positive integer $n$ and $k \leq n$, $\binom{n}{k} \leq \left(\frac{en}{k}\right)^{k}$.
\end{fact}

\begin{proofof}[Lemma~\ref{thm:hardness-linear}]
  Given {\sc 3SAT} instance $\phi$ of size $n$. Let $\cG$ be its clause/variable game. Observe that $\cG$ has $n' = O(n)$ variables, $q' = O(1)$ alphabet size and maximum degree\footnote{It is not hard to see that we can assume without loss of generality that each variable in the {\sc 3SAT} formula appears only in a constant number of clauses. Without going into too much detail, this is because, if we know that the ETHA is true for some $\varepsilon$ and $c$, then we can pick a very large $d \gg 1/\varepsilon, 1/c$. Since there can be at most $n/d$ variables with degrees more than $d$, we can just enumerate all assignments to these variables and produce at most $2^{n/d}$ {\sc 3SAT} formulas where degree of every vertex is at most $d$. Since the original instance takes $2^{cn}$ time to solve, at least one of the new instances also takes $2^{(c - 1/d)n}$ time as well.} $d' = O(1)$. Consider $\cG^l_k$ from Definition~\ref{def:k-csp-red} with $l = \frac{\beta n \log i \log k}{ik}$ where $\beta$ is a small constant which will be chosen later.

  Let $\tilde n$ and $\tilde q$ be the number of variables and the alphabet size of $\cG^l_k$. We have $\tilde q \leq (q')^{2l} \leq 2^{O\left(\frac{\beta n \log i \log k}{ik}\right)}$. Moreover, when $\beta$ is sufficiently small, from Fact~\ref{fact:binom-approx}, we have $$\tilde n \leq \binom{n'}{l}^2 \leq \left(\frac{e n'}{l}\right)^{2l} = \left(O\left(\frac{ik}{\beta \log i \log k}\right)\right)^{2l} \leq 2^{O\left(\frac{\beta n \log^2 i \log^2 k \log (1/\beta)}{ik}\right)} \leq 2^{O\left(\frac{\sqrt{\beta} n \log^2 i \log^2 k}{ik}\right)}.$$

  As for the completeness and soundness of the reduction, first, it is obvious that $val(\phi) = 1$ implies $val(\cG^l_k) = 1$. Otherwise, from Proposition~\ref{prop:clause-var-val}, if $val(\phi) \leq 1 - \varepsilon$, then $val(\cG_\phi) \leq 1 - \varepsilon/3$. By Lemma~\ref{lem:k-csp-value}, we have $$val(\cG^l_k) \leq 2(1 - \varepsilon/6)^{\Omega\left(\frac{\varepsilon^3 k^2 l^2 (n')}{(d')(n')^2}\right)} \leq (1/2)^{\Omega(\beta^2 n \log^2 i \log^2 k / i^2)} \leq (1/\tilde n \tilde q)^{\Omega(\beta^2 k/ i)} = (1/\tilde N)^{\Omega(\beta^2 / i)}$$
  where $\tilde N = (\tilde n \tilde q)^k \leq 2^{O(\sqrt{\beta} n \log^2 i \log^2 k/i)}$ is the size of $\cG^l_k$.

  Finally, pick $\beta$ to be a constant small enough that $\tilde N \leq O(2^{c n \log^2 i \log^2 k/ i})$ where $c$ is the constant from the ETHA. If a algorithm can distinguish $val(\cG^l_k) = 1$ from $val(\cG^l_k) \leq (1/\tilde N)^{\Omega(1/i)}$ in $O(\tilde N^{\frac{i}{\log^2 i \log^2 k}})$ time, it can also distinguish $val(\phi) = 1$ from $val(\phi) \leq 1 - \varepsilon$ in time $O(2^{cn})$ time, contradicting with ETHA.
\end{proofof}

\subsubsection{Lasserre Integrality Gap for Fully-Dense {\sc Max $k$-CSP}}

We will now show how to get a polynomial integrality gap for the Lasserre relaxation for dense CSPs. In particular, even for $\tilde \Omega(ik)$-level of Lasserre hierarchy, the integrality gap remains $N^{1/i}$ for fully-dense {\sc Max $k$-CSP}, as stated formally below.

\begin{lemma} \label{thm:lasserre-gap}
  For any $k \geq 2$, any sufficiently large $N$ and any sufficiently large $i$, there exists a fully-dense {\sc Max $k$-CSP} instance $\cG$ of size $N$ such that $opt^{\tilde \Omega(ik)}_{Las}(\cG) = 1$ and $val(\cG) \leq (1/N)^{1/i}$.
\end{lemma}

One way to interpret Lemma~\ref{thm:lasserre-gap} is as a lower bound for SDP or LP hierarchies algorithm for dense {\sc Max $k$-CSP}. From this perspective, our result indicates that one cannot hope to use $\tilde O(ik)$-level Lasserre relaxation to approximate fully-dense {\sc Max $k$-CSP} to within a factor of $N^{1/i}$. Since the Lasserre hierarchy is stronger than the SA and the Lov\'{a}sz-Schrijver hierarchies~\cite{Lau03}, such lower bound  holds for those hierarchies as well. Interestingly, this lower bound essentially matches, up to a factor of $\polylog(ik)$ in the number of levels, our algorithmic result presented in the next section, justifying the running time of our algorithm.

On the other hand, Lemma~\ref{thm:lasserre-gap} can be viewed as an unconditional analogue of Lemma~\ref{thm:hardness-linear}. In this sense, we get rid of ETHA assumption at the expense of restricting our computational model to only Lasserre relaxation. Other than those differences, the two lemmas are essentially the same. In fact, to prove Lemma~\ref{thm:lasserre-gap}, we use an unconditional analogue of ETHA under the Lasserre hierarchy model, which is stated below.

\begin{lemma} \label{lem:lasserre-starting-instance}
  For sufficiently large $N$, there exists a projection game $\cG$ of size $N$ with the following properties.
  \begin{itemize}
    \item (Vector Completeness) $opt^{\Omega(N)}_{Las} = 1$.
    \item (Soundness) $val(\cG) = 1 - \varepsilon$ for some constant $\varepsilon > 0$.
    \item (Bounded Degree) Each variable has constant degree.
    \item (Bounded Alphabet Size) The alphabet size is constant.
  \end{itemize}
\end{lemma}

Results similar to Lemma~\ref{lem:lasserre-starting-instance} have been proven before in~\cite{BCVGZ12} and~\cite{M-thesis} by applying the clause/variable reduction to Tulsiani's {\sc Max $k$-CSP} integrality gap~\cite{Tul09}. However, both of the mentioned results consider different regimes compared to ours and cannot be used directly. Nevertheless, the same reduction still works in our setting so we defer the proof of Lemma~\ref{lem:lasserre-starting-instance} to Appendix~\ref{app:lasserre}.

With the help of Lemma~\ref{lem:lasserre-starting-instance}, the proof of Lemma~\ref{thm:lasserre-gap} proceeds in a similar fashion as that of Lemma~\ref{thm:hardness-linear}. However, while the soundness argument remains unchanged, we need to argue completeness for Lasserre solution instead. To do so, we prove the following lemma, which implies that the reductions considered in our paper preserve a complete solution of the Lasserre hierarchy, albeit at a loss in the number of levels.

\begin{lemma} \label{lem:lasserre-reduction-completeness}
  Let $\cG = (V, \cW, \{P_S\})$ be any {\sc Max $k$-CSP} instance. Let $\cG' = (V', \cW', \{P'_{S'}\})$ be an instance of {\sc Max $k'$-CSP} constructed from $\cG$ satisfying the following properties.
  \begin{itemize}
    \item Each variable in $V'$ corresponds to a set $S \subseteq V$ of size at most $d$.
    \item The alphabet set $\Sigma'$ of $\cG'$ is $\Sigma^d$ where $\Sigma$ is the alphabet set of $\cG$. For each $S \in V'$, we associate $|\Sigma|^{|S|}$ elements of $\Sigma'$ to each assignment to $S$ (in $\cG$). Note that since we do not require $S$ to be of size exactly $d$, it is possible that $|\Sigma'| > |\Sigma|^{|S|}$. In this case, we completely ignore the rest of the elements of $\Sigma'$, i.e., the predicate is zero when evaluated on an assignment to such elements.
    \item For every predicate $P_{(S_1, \dots, S_{k'})}$, $P_{(S_1, \dots, S_{k'})}(\phi_{S_1}, \dots, \phi_{S_{k'}}) = 1$ if $\phi_{S_1}, \dots, \phi_{S_{k'}}$ are consistent and $P_{(x_1, \dots, x_k)}(\phi_{S_1} \circ \cdots \circ \phi_{S_{k'}}(x_1), \dots, \phi_{S_1} \circ \cdots \circ \phi_{S_{k'}}(x_k)) = 1$ for every $x_1, \dots, x_k \in S_1 \cup \cdots \cup S_{k'}$.
  \end{itemize}
  Suppose that $opt^r_{Las}(\cG) = 1$ for some $r \geq k, dk'$. Then, $opt^{\lfloor r/d \rfloor}_{Las}(\cG') = 1$.
\end{lemma}

Since the proof of the lemma mainly involves trivial calculations, we defer the proof to Appendix~\ref{app:lasserre-reduction-completeness}.

It is easy to see that the the reduction from Lemma~\ref{lem:k-csp-value} satisfies the condition required in the above lemma. Hence, we immediately arrive at the following corollary.

\begin{corollary} \label{cor:lasserre-k-csp-completeness}
  For any two-prover game $\cG$, if $opt^{r}_{Las}(\cG) = 1$ for some $r \geq kl$, then $opt^{\Omega(r/l)}_{Las}(\cG^l_k) = 1$.
\end{corollary}

Now, we are ready to prove Lemma~\ref{thm:lasserre-gap}.

\begin{proofof}[Lemma~\ref{thm:lasserre-gap}]
  We start with a projection game $\cG$ from Lemma~\ref{lem:lasserre-starting-instance} of size $N$ with $n \leq N$ variables, $q = O(1)$ alphabet size and maximum degree $d = O(1)$. Consider the fully-dense {\sc Max $k$-CSP} $\cG^l_k$ from the reduction in Definition~\ref{def:k-csp-red} on $\cG$ with $l = \frac{\beta n \log i \log k}{ik}$ where $\beta$ is a small constant which will be chosen later.

  Let $\tilde n$ and $\tilde q$ be the number of variables and the alphabet size of $\cG^l_k$. We have $\tilde q \leq q^{2l} \leq 2^{O\left(\frac{\beta N \log i \log k}{ik}\right)}$. Moreover, when $\beta$ is sufficiently small, from Fact~\ref{fact:binom-approx}, we have $\tilde n \leq \binom{n}{l}^2 \leq 2^{\Omega\left(\frac{\sqrt{\beta} N \log^2 i \log^2 k}{ik}\right)}.$

  Furthermore, from Corollary~\ref{cor:lasserre-k-csp-completeness} and from $opt^{\Omega(N)}_{Las}(\cG) = 1$, we have $opt^{\Omega(N/l)}_{Las}(\cG^l_k) = opt^{\tilde \Omega(ik)}_{Las}(\cG^l_k) = 1$. Finally, by Lemma~\ref{lem:k-csp-value}, we have $val(\cG^l_k) \leq 2(1 - \varepsilon/2)^{\Omega\left(\frac{\varepsilon^3 k^2 l^2 n}{dn^2}\right)} \leq (1/\tilde n \tilde q)^{\Omega(\beta^2 k/ i)} = (1/\tilde N)^{\Omega(\beta^2 / i)}$ where $\tilde N = (\tilde n \tilde q)^k$ is the size of $\cG^l_k$. This completes our proof of Lemma~\ref{thm:lasserre-gap}.
\end{proofof}

\subsection{Almost Optimal $\AM(2)$ Protocol for {\sc 3SAT}}

In~\cite{AIM}, Aaronson et al. provided an $\AM(2)$ protocol of length $\tilde O(\sqrt{n})$ for {\sc 3SAT} with completeness 1 and soundness $\delta$ for \emph{some} constant $\delta < 1$. However, since they did not prove that birthday repetition can amplify soundness, they could not get a similar result for arbitrarily small $\delta$. In that case, they invoke Moshkovitz-Raz PCP~\cite{MR10}, which, incontrast to Dinur's PCP, gives arbitrarily small soundness. However, due to the length of Moshkovitz-Raz PCP, their protocol length is $n^{1/2 + o(1)}\poly(1/\delta)$. Since we have proved that the birthday repetition amplifies the soundness, we overcome this obstacle and we can prove Lemma~\ref{lem:am-protocol} easily as follows.

\begin{proofof}[Lemma~\ref{lem:am-protocol}]
  Given a {\sc 3SAT} instance $\phi$ of size $n$, the protocol works as follows. Arthur uses Dinur's PCP Theorem to reduce $\phi$ to $\cG$ with $n' = n \polylog n$ variables, $q' = O(1)$ alphabet size and maximum degree $d' = O(1)$. He then produces a free game $G^{l \times l} = (X, Y, X \times Y, \Sigma_X, \Sigma_Y, \{P_{(x, y)}\})$, the $(l \times l)$-birthday repetition of $\cG$, with $l = c \log^d n \sqrt{n \log(1/\delta)}$ for some constant $c$ and $d$ to be chosen later.

  Arthur then sends independent random questions to the Merlins where the questions for first and second Merlins are drawn from $X$ and $Y$ respectively. The proof of each Merlin is an assignment to the variable he is given. Finally, if the two Merlins receive questions $x \in X, y \in Y$, Arthur uses the predicate $P_{(x, y)}$ to check whether the assignments he received satisfy the predicate. If so, Arthur accepts. Otherwise, he rejects.

  It is obvious that, when $\phi \in $ 3SAT, i.e., $\phi$ is satisfiable, $\cG^{l \times l}$ is satisfiable and Arthur always accepts if Merlins answer according to a satisfying assignment of $\cG^{l \times l}$. On the other hand, if $\phi \notin$ 3SAT, $val(\cG^{l \times l}) \leq 2(1 - \varepsilon/2)^{\Omega\left(\frac{\varepsilon^3 l^2 (n')}{(d')(n')^2}\right)} \leq \delta$ when $c$ and $d$ are chosen to be sufficiently large. Hence, the soundness of the protocol is at most $\delta$. Finally, observe that the protocol has length $2l \log n = \tilde O(\sqrt{n \log(1/\delta)})$ as desired.
\end{proofof}

\section{Improved Approximation Algorithm for Dense CSPs} \label{sec:alg}


Before describing our algorithm, we first explain ingredients central in conditioning-based algorithms: a conditioning operator and a rounding procedure.

{\bf Conditioning Sherali-Adams Solution.} Let $\mu = \{\mathcal{X}_S\}$ be a solution of an $r$-level SA relaxation of a {\sc Max $k$-CSP} instance. For any set $T \subseteq V$ of size at most $r - k$ and for any $\phi_T \subseteq \Sigma^T$ such that $\mathcal{X}_T(\phi_T) > 0$, $\mu$ conditioned on $\phi_T$ is $\mu|\phi_T = \{\mathcal{\tilde X}_S\}_{|S| \leq r - |T|}$ defined as
\begin{align*}
  \mathcal{\tilde X}_S(\phi_S) =
  \begin{cases}
    \mathcal{X}_{S \cup T}(\phi_S \circ \phi_T)/\mathcal{X}_T(\phi_T) & \text{ if } \phi_S \text{ is consistent with } \phi_T, \\
    0 & \text{ otherwise.}
  \end{cases}
\end{align*}
It is not hard to see that $\mu|\phi_T$ is an $(r - |T|)$-level SA solution.

{\bf (Derandomized) Independent Rounding.} A naive way to arrive at an actual solution to the {\sc Max $k$-CSP} instance from a SA relaxation solution $\{\mathcal{X}_S\}_{|S| \leq r}$ is to independently sample each variable $x$ based on the distribution $\mathcal{X}_x$. Observe that the rounded solution expected value is at least\footnote{Note that we assume without loss of generality that, when $S$ contains a repeated variable, $P_S$ is zero on any assignment that assigns the same variable differently. This means that, when $S = (x_{i_1}, \dots, x_{i_k})$ contains repeated variables, $\E[P_S(\phi_S)]$ when $\phi_S$ is sampled based upon $\mathcal{X}_{i_1} \times \cdots \times \mathcal{X}_{i_k}$ is no more than $\E[P_S(\phi_S)]$ when $\phi_S$ is sampled according to independent rounding.} $\E_{S = (x_{i_1}, \dots, x_{i_k}) \sim \cW} \left[\E_{\phi_S \sim \mathcal{X}_{i_1} \times \cdots \times \mathcal{X}_{i_k}}\left[P_S(\phi_S)\right]\right].$
It is obvious that, by a conditional expectation argument, independent rounding can be derandomized so that the value of the output is at least the expectation.

Without going into too much detail, conditioning-based algorithms typically proceed as follows. First, solve a LP/SDP relaxation of the problem. As long as the solution has large ``total correlation'', try conditioning it on an assignment to a random variable. Once the solution has small total correlation, use independent rounding on the solution to get the desired assignment. The intuition behind such algorithms is that, if the solution has large total correlation, conditioning on one variable substantially reduces the total correlation. Hence, after a certain number of rounds of conditioning, the total correlation becomes small. At this point, the solution is quite independent and independent rounding gives a good approximation.

Our algorithm will also follow this framework. In fact, our algorithm remains largely unchanged from~\cite{YZ14} with the exception that we will use a stronger relaxation to reduce our work in arguing about the value of conditioned solutions. However, our main contribution lies in the analysis: we will show that independent rounding does well even when the total correlation is large (super-constant). This is in contrast to the previously known conditioning-based algorithms~\cite{BRS11, RT12, YZ14}, all of which require their measures of correlation to be small constants to get any meaningful result.

The new relaxation that we will used is the following. For convenience, we call this the $r$-level relaxation Sherali-Adams with Conditioning
(SAC) relaxation of {\sc Max $k$-CSP}.
\begin{align*}
  \text{maximize }
  &\lambda \\
  \text{subject to }
  & \{\mathcal{X}_S\}_{|S| \leq r} \text{ is a valid } r\text{-level SA solution} \\
  & \E_{S \sim \cW} [\E_{\phi_S \sim (\mu|\phi_T)}[P_S(\phi_S)]] \geq \lambda & \forall T, \phi_T \text{ s.t. } |T| \leq r - k, \mathcal{X}_T(\phi_T) > 0.
\end{align*}

At a glance, the program above may not look like a linear program. Fortunately, $\E_{S \sim \cW} [\E_{\phi_S \sim (\mu|\phi_T)}[P_S(\phi_S)]] \geq \lambda$ can be written as $\E_{S \sim \cW}[\sum_{\phi_S \in \Sigma^S} \mathcal{X}_{S \cup T}(\phi_S \circ \phi_T)P_{S \cup T}(\phi_S \circ \phi_T)] \geq \lambda \mathcal{X}_T(\phi_T)$, which is linear when $\lambda$ is a constant rather than a variable. As a result, we can solve the optimization problem above by binary search on $\lambda$: for a fixed $\lambda$, we can check whether the inequalities is feasible using a typical polynomial-time algorithm for LP. Hence, we can approximate $\lambda$ to within arbitrarily small additive error in polynomial time. To compute $\lambda$ exactly, observe that $\cW$ is part of the input and is expressible in polynomial number of bits. This means that there are only exponentially many choices for $\lambda$; in particular, if all probabilities in $\cW$ has only $b$ digits after decimal point, then so does $\lambda$. Hence, the described binary search can find $\lambda$ in $(nq)^{O(r)}$ time.

We now state our algorithm. In summary, we first solve an $O(\frac{k^2i}{\Delta} + k)$-level SAC relaxation for the instance. We then try every possible conditioning (i.e., on every set $T \subseteq V$ of size at most $k^2i/\Delta$ and every assignment to $T$). For each conditioned solution, we use independent rounding to arrive at an assignment. Finally, output the best such assignment. The pseudo-code for the full algorithm is shown below in Figure~\ref{fig:alg-dense-csp}.

  \begin{figure}[h!]
    \centering
\begin{minipage}{0.68\textwidth}
  \begin{algorithm}[H]
    \caption{Approximation Algorithm for Dense CSPs \label{alg:dense-csp}}
            {\bf Input:} a $\Delta$-dense {\sc Max $k$-CSP} instance $\cG = (V, \cW, \{P_S\})$, an integer $i$\\
            {\bf Output:} An assignment $\phi: V \rightarrow \Sigma$
            \begin{algorithmic}
              \State $r \leftarrow (k^2 i/\Delta + k)$
              \Do
              \State $r \leftarrow r + 1$
              \State $\mu \leftarrow$ solution of $r$-level of SAC relaxation for $\cG$.
              \State $\lambda \leftarrow$ value of $\mu$
              \doWhile{$(r - k) \lambda < k^2 i / \Delta$ and $r < n$}
              \State $\phi \leftarrow \emptyset$
              \For{$T \subseteq V$ of size at most $r - k$}
              \For{$\phi_T \in \Sigma^T$}
              \State $\phi' \leftarrow$ independent rounding of $\mu|\phi_T$
              \If{$val(\phi') > val(\phi)$}
              \State $\phi \leftarrow \phi'$.
              \EndIf
              \EndFor
              \EndFor
              \State \textbf{return} $\phi$
            \end{algorithmic}
  \end{algorithm}
\end{minipage}
\caption{Pseudo-code of Our Approximation Algorithm for Dense CSPs. The only difference between this pseudo-code and the above summary of our algorithm is that we need to iteratively increase the number of levels of the hierarchy. This is due to the fact that, as we will see in Lemma~\ref{lem:indrounding}, the number of levels needed depends on the value of the solution. More specifically, we want $r \geq k^2 i / (\Delta \lambda) + k$ \label{fig:alg-dense-csp}}
\end{figure}

  The rest of the section is organized as follows. In Subsection~\ref{subsec:totalcor}, we formally define total correlation and state a bound from~\cite{YZ14} on the total correlation of conditioned solutions. Next, in Subsection~\ref{subsec:cor-to-alg}, we state and prove our main contribution of this section, i.e., that even when the total correlation is super-constant, we can still get a non-trivial approximation from independent rouding. Finally, in Subsection~\ref{subsec:alg}, we put these together and prove the approximation guarantee for our algorithm.

\subsection{Total Correlation of Conditioned Sherali-Adams Relaxation Solution} \label{subsec:totalcor}

We start by defining the total correlation of a SA solution. For a $k$-level SA solution $\mu = \{\mathcal{X}_S\}$ and for tuple $S = (x_{i_1}, \dots, x_{i_j}) \in V^j$ of size $j \leq k$, the total correlation among $x_{i_1}, \dots, x_{i_j}$ is $C_\mu(x_S) = C(\sigma_{i_1}; \dots; \sigma_{i_j})$ where $\sigma_{i_1}, \dots, \sigma_{i_j}$ are jointly sampled from $\mathcal{X}_{\{x_{i_1}, \dots, x_{i_j}\}}$. The total correlation of $\mu$ is then defined as $C(\mu) = \E_{S \sim \cW} [C_\mu(x_S)]$. We call $\mu$ \emph{a $\kappa$-independent solution} if its total correlation is at most $\kappa$.

Yoshida and Zhou show that, for any $l > 0$ and any $(l + k)$-level SA solution $\mu$, there exists a subset $T$ of size at most $l$ and an assignment $\phi_T \in \Sigma^T$ such that the total correlation of $(\mu|\phi_T)$ is at most $3^k\log q/(l\Delta)$ where $\Delta$ is the density of the instance. Here we are able to show a slightly improved bound as stated below. Since the proof is similar to that of~\cite{YZ14} with only a few additional steps, we defer the proof to Appendix~\ref{app:corr}.

\begin{lemma} \label{lem:corr-decrease}
  Let $\mu$ be any $r$-level SA solution of a $\Delta$-dense {\sc Max $k$-CSP} instance $\cG = (V, \cW, \{P_S\})$. Then, for any $0 < l \leq r - k$, there exists $t \leq l$ such that
  $\E_{T \sim V^t, \phi_T \sim \Sigma^T}[C(\mu|\phi_T)] \leq \frac{k^2 \log q}{l \Delta}.$
\end{lemma}

\subsection{New Bound on Rounding $\kappa$-independent Solution} \label{subsec:cor-to-alg}

In this subsection, we prove our main lemma for this section. For the known conditioning-based algorithms, once the solution is fairly independent, it is easy to show that independent rounding gives a good solution. In particular, Raghavendra-Tan~\cite{RT12} and Yoshida-Zhou\cite{YZ14} proofs, whose measures of correlation are the same as ours\footnote{In~\cite{RT12}, only 2-CSPs were studied and they measure correlation by mutual information of the variables in the constraints.}, conclude this step by using the Pinsker's inequality, which states that, for any distributions $\mathcal{X}$ and $\mathcal{Y}$, $D_{KL}(\mathcal{X}\|\mathcal{Y}) \geq (2 \log 2) \|\mathcal{X} - \mathcal{Y}\|_1^2$ where $\|\mathcal{X} - \mathcal{Y}\|_1 = \sum_{\theta \in \Theta} |\mathcal{X}(\theta) - \mathcal{Y}(\theta)|$ is the $L^1$-distance between $\mathcal{X}$ and $\mathcal{Y}$. Roughly speaking, $\mathcal{X}$ is going to be the distribution in the LP solution whereas $\mathcal{Y}$ is the distribution resulting from independent rounding. Hence, when they bound $D_{KL}(\mathcal{X}\|\mathcal{Y})$ to be at most a small constant $\varepsilon$, it follows immediately that any predicate $f$ with domain $\supp(\mathcal{X})$ in $[0, 1]$ satisfies $|\E_{x \sim \mathcal{X}}[f(x)] - \E_{y \sim \mathcal{Y}}[f(y)]| \leq \sqrt{\varepsilon/(2 \log 2)}$. Thus, if $\E_{x \sim \mathcal{X}}[f(x)]$, the value of the LP solution, is large, then $\E_{y \sim \mathcal{Y}}[f(y)]$, the expected value of a solution from independent rouding, is also large.

While this works great for small constant $\varepsilon$, it does not yield any meaningful bound when $\varepsilon$ is larger than a certain constant. A natural question is whether one can prove any non-trivial bound for super-constant $\varepsilon$. In this regard, we prove the following lemma, which positively answers the question. For convenience, $0^0$ is defined to be 1 throughout this and next subsections and, whenever we write the expression $(\delta^\delta e^{-\kappa})^{\frac{1}{1 - \delta}} (1 - \delta)$ with $\delta = 1$, we define it to be 0.

\begin{lemma} \label{lem:funcbound}
Let $\mathcal{X}$ and $\mathcal{Y}$ be any two probability distributions over a finite domain $\Theta$ such that $D_{KL}(\mathcal{X} \| \mathcal{Y}) \leq \kappa$ and let $f: \Theta \to [0, 1]$ be any function. If $\E_{x \sim \mathcal{X}}[f(x)] = 1 - \delta$, then $\E_{y \sim \mathcal{Y}}[f(y)] \geq \left(\delta^\delta e^{-\kappa}\right)^{\frac{1}{1 - \delta}} (1 - \delta)$.
\end{lemma}

\begin{proofof}[Lemma~\ref{lem:funcbound}]
  We assume without loss of generality that $\delta \notin \{0, 1\}$ since, when $\delta = 0$, we can modify $f$ infinitesimally small and take the limit of the bound and, when $\delta = 1$, the bound is trivial.

  Let $\mathcal{Z}$ and $\mathcal{T}$ be two probability distributions on $\Theta$ such that $\mathcal{Z}(\theta) = \frac{\mathcal{X}(\theta)f(\theta)}{1 - \delta}$ and $\mathcal{T}(\theta) = \frac{\mathcal{X}(\theta)(1 - f(\theta))}{\delta}$. Observe that $\mathcal{Z}$ and $\mathcal{T}$ are indeed valid distributions on $\Theta$ since $\E_{\theta \sim \mathcal{X}}[f(\theta)] = 1 - \delta$. Observe that $\supp(\mathcal{Z}), \supp(\mathcal{T}) \subseteq \supp(\mathcal{X})$, which is in turn contained in $\supp(\mathcal{Y})$ since $D_{KL}(\mathcal{X} \| \mathcal{Y}) \ne \infty$.

From Weighted A.M.-G.M. inequality, we have
\begin{align*}
\E_{y \sim \mathcal{Y}}[f(y)] = \sum_{\theta \in \Theta} \mathcal{Y}(\theta)f(\theta)
&\geq \sum_{\theta \in \supp(\mathcal{Z})} \mathcal{Z}(\theta)\left(\frac{\mathcal{Y}(\theta)f(\theta)}{\mathcal{Z}(\theta)}\right) \\
(\text{Weighted A.M.-G.M. inequality}) &\geq \prod_{\theta \in \supp(\mathcal{Z})} \left(\frac{\mathcal{Y}(\theta)f(\theta)}{\mathcal{Z}(\theta)}\right)^{\mathcal{Z}(\theta)} \\
&= (1 - \delta) \left(\prod_{\theta \in \supp(\mathcal{Z})} \left(\frac{\mathcal{Y}(\theta)}{\mathcal{X}(\theta)}\right)^{\mathcal{X}(\theta)f(\theta)}\right)^{\frac{1}{1 - \delta}}. \\
\end{align*}

We will next bound $\prod_{\theta \in \supp(\mathcal{Z})} \left(\frac{\mathcal{Y}(\theta)}{\mathcal{X}(\theta)}\right)^{\mathcal{X}(\theta)f(\theta)}$ by writing it in term of $D_{KL}(\mathcal{X}\|\mathcal{Y})$ and a small term which will be bounded later.
\begin{align*}
\prod_{\theta \in \supp(\mathcal{Z})} \left(\frac{\mathcal{Y}(\theta)}{\mathcal{X}(\theta)}\right)^{\mathcal{X}(\theta)f(\theta)}
&= \left(\prod_{\theta \in \supp(\mathcal{X})} \left(\frac{\mathcal{Y}(\theta)}{\mathcal{X}(\theta)}\right)^{\mathcal{X}(\theta)}\right)\left(\prod_{\theta \in \supp(\mathcal{\mathcal{T}})} \left(\frac{\mathcal{X}(\theta)}{\mathcal{Y}(\theta)}\right)^{\mathcal{X}(\theta)(1 - f(\theta))} \right) \\
&= \frac{1}{e^{D_{KL}(\mathcal{X}\|\mathcal{Y})}} \left(\prod_{\theta \in \supp(\mathcal{\mathcal{T}})} \left(\frac{\mathcal{X}(\theta)}{\mathcal{Y}(\theta)}\right)^{\mathcal{X}(\theta)(1 - f(\theta))} \right) \\
(\text{Since } D_{KL}(\mathcal{X}\|\mathcal{Y}) \leq \kappa)&\geq e^{-\kappa} \left(\prod_{\theta \in \supp(\mathcal{\mathcal{T}})} \left(\frac{\mathcal{X}(\theta)}{\mathcal{Y}(\theta)}\right)^{\mathcal{X}(\theta)(1 - f(\theta))} \right) \\
\end{align*}

Intuitively, the term $\prod_{\theta \in \supp(\mathcal{\mathcal{T}})} \left(\frac{\mathcal{X}(\theta)}{\mathcal{Y}(\theta)}\right)^{\mathcal{X}(\theta)(1 - f(\theta))}$ should not be much smaller than one since the sum of the exponent is just $\sum_{\theta \in \supp(\mathcal{\mathcal{T}})} \mathcal{X}(\theta)(1 - f(\theta)) = \delta$. Indeed, this term is small as we can bound it as follows:
\begin{align*}
\prod_{\theta \in \supp(\mathcal{\mathcal{T}})} \left(\frac{\mathcal{X}(\theta)}{\mathcal{Y}(\theta)}\right)^{\mathcal{X}(\theta)(1 - f(\theta))}
&= \left(\prod_{\theta \in \supp(\mathcal{\mathcal{T}})} \left(\frac{\delta}{1 - f(\theta)} \cdot \frac{\mathcal{T}(\theta)}{\mathcal{Y}(\theta)}\right)^{\mathcal{T}(\theta)}\right)^{\delta} \\
&\geq \left(\prod_{\theta \in \supp(\mathcal{\mathcal{T}})} \left(\delta \cdot \frac{\mathcal{T}(\theta)}{\mathcal{Y}(\theta)}\right)^{\mathcal{T}(\theta)}\right)^{\delta} \\
&= \delta^\delta \left(e^{D_{KL}(\mathcal{T}\|\mathcal{Y})}\right)^\delta \\
&\geq \delta^\delta
\end{align*}
The last inequality comes from the fact that the informational divergence of any two distributions is no less than zero.

Combining the three inequalities, we have
$\E_{\theta \sim \mathcal{Y}}[f(\theta)]
\geq (1 - \delta) \left(e^{-\kappa}\delta^\delta\right)^{\frac{1}{1 - \delta}},$
as desired.
\end{proofof}

Now, we will use Lemma~\ref{lem:funcbound} to give a new bound for the value of the output from independent rounding on a $k$-level $\kappa$-independent solution of the Sherali-Adams Hierarchy.

\begin{lemma} \label{lem:indrounding}
If $\{\mathcal{X}_S\}$ is a $k$-level $\kappa$-independent SA solution of value $1 - \delta$ for an instance $(V, \cW, \{P_S\})$ of {\sc Max $k$-CSP}, then independent rounding gives an assignment of value at least $(\delta^\delta e^{-\kappa})^{\frac{1}{1 - \delta}} (1 - \delta)$.
\end{lemma}

\begin{proof}
  Again, we assume without loss of generality that $\delta \notin \{0, 1\}$.
  
  For each $k$-tuple $S = (x_{i_1}, \dots, x_{i_k})$, let $\kappa_S = D_{KL}(\mathcal{X}_S \| \mathcal{X}_{i_1} \times \cdots \times \mathcal{X}_{i_k})$ and $\delta_S = 1 - \E_{\phi_S \sim \mathcal{X}_S}[P_S(\phi_S)]$. Recall that the value of $\{\mathcal{X}_S\}$ in the SA relaxation is $\E_{S \sim \cW}[\E_{\phi_S \sim \mathcal{X}_S}[P_S(\phi_S)]] = (1 - \delta)$. Hence, we have $\E_{S \sim \cW}[\delta_S] = \delta$. Moreover, since $\{\mathcal{X}_S\}$ is a $\kappa$-independent solution, we have $\E_{S \sim \cW}[\kappa_S] \leq \kappa$.

  As stated earlier, the independent rounding algorithm gives an assignment of value at least $$\E_{S = (x_{i_1}, \dots, x_{i_k}) \sim \cW} \left[\E_{\phi_S \sim \mathcal{X}_{i_1} \times \cdots \times \mathcal{X}_{i_k}}\left[P_S(\phi_S)\right]\right].$$ From Lemma~\ref{lem:funcbound}, we have $\E_{\phi_S \sim \mathcal{X}_{i_1} \times \cdots \times \mathcal{X}_{i_k}}\left[P_S(\phi_S)\right] \geq (\delta_S^{\delta_S}e^{-\kappa_S})^{\frac{1}{1 - \delta_S}}(1 - \delta_S).$ Thus, the assignment from the rounding procedure has value at least $\E_{S \sim \cW}[(\delta_S^{\delta_S}e^{-\kappa_S})^{\frac{1}{1 - \delta_S}}(1 - \delta_S)]$. 

Next, let $\mathcal{Y}$ and $\mathcal{Z}$ be distributions on $V^k$ defined by $\mathcal{Y}(S) = \frac{\cW(S)(1 - \delta_S)}{(1 - \delta)}$ and $\mathcal{Z}(S) = \frac{\cW(S)\delta_S}{\delta}$. $\mathcal{Y}$ and $\mathcal{Z}$ are valid distributions since $\E_{S \sim \cW}[\delta_S] = \delta$.

We can now bound $\E_{S \sim \cW}[(\delta_S^{\delta_S}e^{-\kappa_S})^{\frac{1}{1 - \delta_S}}(1 - \delta_S)]$ as follows:
\begin{align*}
\E_{S \sim \cW}[(\delta_S^{\delta_S}e^{-\kappa_S})^{\frac{1}{1 - \delta_S}}(1 - \delta_S)]
&= \sum_{S \in V^k} \cW(S)(\delta_S^{\delta_S}e^{-\kappa_S})^{\frac{1}{1 - \delta_S}}(1 - \delta_S) \\
&= (1 - \delta) \sum_{S \in \supp(\mathcal{Y})} \mathcal{Y}(S)(\delta_S^{\delta_S}e^{-\kappa_S})^{\frac{1}{1 - \delta_S}} \\
(\text{Weighted A.M.-G.M. inequality}) &\geq (1 - \delta) \prod_{S \in \supp(\mathcal{Y})} \left(\delta_S^{\delta_S}e^{-\kappa_S}\right)^\frac{\mathcal{Y}(S)}{1 - \delta_S} \\
&= (1 - \delta) \left(\prod_{S \in \supp(\mathcal{Y})} \left(\delta_S^{\delta_S}e^{-\kappa_S}\right)^{\cW(S)}\right)^\frac{1}{1 - \delta} \\
(\text{Since } \E_{S \sim \cW}[\kappa_S] = \kappa \text{ and } \supp(\mathcal{Y}) \subseteq \supp(\cW)) &\geq (1 - \delta) \left(e^{-\kappa} \prod_{S \in \supp(\mathcal{Y})} \delta_S^{\cW(S)\delta_S}\right)^\frac{1}{1 - \delta} \\
&= (1 - \delta) \left(e^{-\kappa} \prod_{S \in \supp(\mathcal{Z})} \delta_S^{\cW(S)\delta_S}\right)^\frac{1}{1 - \delta}
\end{align*}
The last equality is true because $\delta_S = 1$ for every $S \in \supp(\mathcal{Z}) - \supp(\mathcal{Y})$ and $\delta_S = 0$ for every $S \in \supp(\mathcal{Y}) - \supp(\mathcal{Z})$.

We can now write $\prod_{S \in \supp(\mathcal{Z})} \delta_S^{\cW(S)\delta_S}$ as
\begin{align*}
\prod_{S \in \supp(\mathcal{Z})} \delta_S^{\cW(S)\delta_S}
&= \left(\prod_{S \in \supp(\mathcal{Z})} \left(\delta \cdot \frac{\mathcal{Z}(S)}{\cW(S)}\right)^{\mathcal{Z}(S)} \right)^{\delta} \\
&= \delta^\delta (e^{D_{KL}(\mathcal{Z} \| \mathcal{X})})^\delta \\
(\text{Since } D_{KL}(\mathcal{Z} \| \mathcal{X}) \geq 0)&\geq \delta^\delta.
\end{align*}

Combining the two inequality yields $\E_{S \sim \cW}[(\delta_S^{\delta_S}e^{-\kappa_S})^{\frac{1}{1 - \delta_S}}(1 - \delta_S)] \geq (1 - \delta)(e^{-\kappa}\delta^\delta)^\frac{1}{1 - \delta}$, which completes the proof of the lemma.
\end{proof}

\subsection{New Approximation Guarantee for the Algorithm} \label{subsec:alg}

With Lemma~\ref{lem:corr-decrease} and Lemma~\ref{lem:indrounding} set up, we now prove the algorithmic guarantee for Algorithm~\ref{alg:dense-csp}.

\begin{theorem} \label{thm:alg-dense-csp}
  For any {\sc Max $k$-CSP} instance $\cG$  of value $1 - \delta > 0$ and density $\Delta > 0$, Algorithm~\ref{alg:dense-csp} runs in time $N^{O\left(\frac{k i}{(1 - \delta)\Delta}\right)}$ and outputs an assignment of value at least $(1 - \delta)\delta^{\frac{\delta}{1 - \delta}}/q^{1/i}$.
\end{theorem}

\begin{proof}
  Observe that the running time is $(nq)^{O(r)}$ where $r$ is the maximum level of the SAC relaxation solved by the algorithm. Since the program is a relaxation of {\sc Max $k$-CSP}, $\lambda$ is always at least $1 - \delta$. By the condition of the loop, $r$ is at most $1 + k + \frac{k^2 i}{(1 - \delta)\Delta}$. Hence, the running time of the algorithm is $N^{O\left(\frac{k i}{(1 - \delta)\Delta}\right)}$.

  Next, we will argue about the value of the output assignment. From Lemma~\ref{lem:corr-decrease}, there exists a set $T \subseteq V$ of size at most $\frac{k^2 i}{\lambda\Delta}$ and an assignment $\phi_T \in \Sigma^T$ such that $\mu|\phi_T$ is an $(\lambda\log q / i)$-independent solution. Moreover, from how SAC program is defined, we know that $val_{SA}(\mu|\phi_T) \geq \lambda$. As a result, from Lemma~\ref{lem:indrounding}, independent rounding on $\mu|\phi_T$ gives an assignment of value at least $$((1 - \lambda)^{1 - \lambda} e^{-\lambda\log q / i})^{\frac{1}{\lambda}}\lambda = \lambda(1 - \lambda)^{\frac{1 - \lambda}{\lambda}}/q^{1/i}.$$ Finally, since $|T| \leq \frac{k^2 i}{\lambda\Delta} \leq r - k$, it is considered in the conditioning step of the algorithm. Thus, the output assignment is of value at least $\lambda(1 - \lambda)^{\frac{1 - \lambda}{\lambda}}/q^{1/i} \geq (1 - \delta)\delta^{\frac{\delta}{1 - \delta}}/q^{1/i}$.
\end{proof}

Observe that, when the instance is satisfiable, $\delta = 0$ and the value of the output assignment is at least $1/q^{1/i}$. By taking $i$ to be large enough, one arrives at a quasi-polynomial time approximation scheme (QPTAS) for dense {\sc Max $k$-CSP}, as stated below. We note that our algorithm unfortunately does not give a QPTAS for the nonsatisfiable case since we also lose an additional factor of $\delta^{\frac{\delta}{1 - \delta}}$ in the value of the output solution.

\begin{corollary} \label{thm:qptas-dense-csp}
  There exists an algorithm that, given a satisfiable $\Delta$-dense {\sc Max $k$-CSP} instance $\cG$ and any constant $1/2 > \varepsilon > 0$, runs in $N^{O\left(\frac{k \log q}{\varepsilon\Delta}\right)}$ time and output an assignment to $\cG$ of value at least $1 - \varepsilon$.
\end{corollary}

\begin{proof}
  Run Algorithm~\ref{alg:dense-csp} with $i = \log q / \log (1 + \varepsilon)$. From Theorem~\ref{thm:alg-dense-csp}, the output assignment has value at least $q^{1/i} = 1/(1 + \varepsilon) \geq 1 - \varepsilon$ while the running time is $N^{O\left(\frac{k i}{\Delta}\right)}$. Finally, we conclude by observing that $i = \log q / \log (1 + \varepsilon) \leq O(\log q / \varepsilon)$ where the inequality follows from the Bernoulli's inequality.
\end{proof}

\subsection{Approximation Algorithm for {\sc Densest $k$-Subhypergraph}}

In this subsection, we provide our algorithm for {\sc Densest $k$-Subhypergraph}, as stated below.

\begin{corollary} \label{cor:dense-hypergraph}
  There exists a randomized algorithm that, given an integer $i > 0$ and a $d$-uniform hypergraph whose densest $k$-subhypergraph has density $\Delta$, runs in time $n^{2^{O(d \log d)}i/\Delta}$ and outputs a $k$-subhypergraph with density at least $\frac{\Delta}{2^{O(d \log d)}n^{1/i}}$ with high probability.
\end{corollary}

Charikar \etal~\cite{CHK} discovered a simple randomized polynomial-time reduction from {\sc Densest $k$-Subgraph} to {\sc Max 2-CSP} that preserves approximation by a constant factor. This reduction was used in~\cite{MM15} to give an approximation algorithm for {\sc Densest $k$-Subgraph} when the optimal subgraph is sufficiently dense. By modifying the reduction slightly, we can also turn our algorithm to an approximation algorithm for {\sc Densest $k$-Subhypergraph} on $d$-uniform hypergraph whose optimal subhypergraph is sufficiently dense. The properties of the reduction are described in the following lemma.

\begin{lemma} \label{lem:dense-hyper-reduction}
  There is a randomized polynomial-time algorithm that, given a $d$-uniform hypergraph $(V, E)$ on $n$ vertices and an integer $k \leq n$, produces a fully-dense {\sc Max $d$-CSP} instance $\cG = (V', (V')^d, \{P_S\})$ such that
  \begin{itemize}
    \item the alphabet size of $\cG$ is $n$ and the number of variables $|V'|$ is $k$,
    \item there is a polynomial-time algorithm that, given any assignment $\phi$ of $\cG$, outputs $k$-subhypergraph of $(V, E)$ whose density is at least $val_{\cG}(\phi)$, and,
    \item if $k \geq 8d^2$, then, with probability at least $1/2$, $val(\cG) \geq \Delta/2^{O(d \log d)}$.
  \end{itemize}
\end{lemma}

Since the proof the lemma consists of only simple probabilistic arguments, we defer it to Appendix~\ref{app:red-dense-hyper}. We will now show that the reduction, together with our algorithm for denses {\sc Max $k$-CSP}, imply Corollary~\ref{cor:dense-hypergraph}.

\begin{proofof}[Corollary~\ref{cor:dense-hypergraph}]
  The algorithm works on input $(V, E)$ as follows:
  \begin{enumerate}
  \item If $k < 8d^2$, use brute-force search to find the optimal subgraph.
  \item Let $\tau = 1$.
  \item As long as the algorithm has not output, do the following steps:
    \begin{enumerate}
      \item Repeat the following processes $n$ times:
      \begin{enumerate}
      \item Use the reduction from Lemma~\ref{lem:dense-hyper-reduction} to reduce $(V, E)$ to a {\sc Max $d$-CSP} instance $\cG$.
      \item Run Algorithm~\ref{alg:dense-csp} (from Theorem~\ref{thm:alg-dense-csp}) on $\cG$ but if the algorithm tries to solve SAC relaxation of more than $d^2i/\tau + d$ level (i.e. $\lambda$ in the algorithm is less than $\tau$), abort the algorithm. \label{step:hyper-alg-run}
      \item If the algorithm in the previous step was not aborted, let the output assignment be $\phi$. Use Lemma~\ref{lem:dense-hyper-reduction} to turn $\phi$ into a $k$-subhypergraph of $(V, E)$. Output the subhypergraph.
      \end{enumerate}
    \item Set $\tau \leftarrow \tau / 2$.
    \end{enumerate}
  \end{enumerate}

  To see that the algorithm have the desired properties, first observe that if $val(\cG) \geq \tau$, then Step~\ref{step:hyper-alg-run} is never aborted. Hence, from the last property of the reduction in Lemma~\ref{lem:dense-hyper-reduction}, we know that the above algorithm ends while $\tau$ is still at least $\Delta/2^{O(d \log d)}$ with probability $1 - 2^{-n}$. In this case, from Theorem~\ref{thm:alg-dense-csp}, we know that the output subgraph has value at least $\Omega(\tau/n^{1/i})$ and it is obvious that the running time of the algorithm is at most $n^{2^{O(d \log d)}i/\Delta}$ as desired.
\end{proofof}

\section{Conclusion and Open Problems} \label{sec:open}

We prove that birthday repetition can amplify gap in hardness of approximation. This has several interesting consequences to the approximability of dense {\sc Max $k$-CSP}. First, we prove almost-polynomial ratio polynomial-time ETH-hardness for the problem. Second, we show, assuming the stronger ETHA, that it is impossible to approximate dense {\sc Max $k$-CSP} to within factor $N^{1/i}$ in time $N^{\tilde O_k(i)}$. Third, we prove a similar integrality gap for Lasserre relaxation of the problem. Moreover, we provide an approximation algorithm that matches our lower bound based on ETHA and the Lasserre integrality gap.

While our results settle down the approximability of dense {\sc Max $k$-CSP} up to the dependency on $k$ and a factor of $\polylog i$ in the exponent, our work also raises many interesting questions, which we list below.

\begin{itemize}
  \item {\em Can the birthday repetition theorem be used to prove almost-polynomial ratio hardness for other problems?} As stated earlier, the birthday repetition with $k = l = \tilde O(\sqrt{n})$ has inspired quasi-polynomial running time lower bounds for many problems. Can we use our technique to prove running time lower bounds for almost-polynomial hardness ratio similar to those for dense {\sc Max $k$-CSP} we achieved?

    A concrete candidate problem is {\sc Densest $k$-Subgraph} with perfect completeness, i.e., the optimal solution is a $k$-clique. Similar to dense CSPs, the best known polynomial time algorithms take $n^{O(i)}$ time and give an $O(n^{1/i})$-approximation solution~\cite{FPK01, ST05,MM15}.  
    

  \item {\em Is there a birthday repetition theorem for low-value game?} It has been shown that, if one starts with a game $\cG$ of subconstant value, the $r$-parallel repetition $\cG^{\otimes r}$ has value roughly $val(\cG)^{\Omega(r)}$~\cite{DS14,BG15}. A natural question is whether it is possible to prove such theorem for birthday repetition as well. Our technique fails to show such a theorem; in particular, our proof has two steps that produce additive errors (i.e. Lemma~\ref{lem:concen-1} and Lemma~\ref{lem:concen-2}), which prevents us to reduce the value beyond $\exp(-kl/n)$. 
  \item {\em What is the right dependency on $\varepsilon$ and $c$ in the birthday repetition theorem?} It is likely that the dependency of $\varepsilon$ and $c$ in our birthday repetition is not tight. In particular, parallel repetition for general games only has $1/c$ factor in the exponent whereas our theorem has $1/c^2$; would it be possible to reduce the dependency to $1/c$ in birthday repetition? Similar question also applies to $\varepsilon$.
  \item {\em Can our approximation algorithm for dense {\sc $k$-CSP} be made to run in $q^{O_k(i)} + N^{O(1)}$ time?} As stated earlier, Yaroslavtev's algorithm~\cite{Yar14} runs in $q^{O_k(\log q/\varepsilon^2)} + N^{O(1)}$ time and provides an $\varepsilon$ additive approximation to the problem. As for our algorithm, we can, in fact, turn the condioning step into a randomized algorithm where we just randomly pick a set and an assignment to condition\footnote{This is because the bound proved in Appendix~\ref{app:corr} on total correlation of conditioned solution is on its expectation among uniformly random tuple $T$ of size at most $l$ and random $\phi_T$ sampled according to the marginal distribution $\mathcal{X}_T$.}, which takes only linear time. The bottleneck, however, is solving the linear program (SAC relaxation), which takes $N^{\Omega(r)}$ time where $r$ is the number of rounds. Related to this, Barak \etal~\cite{BRS11} showed that their Lasserre hierarchy-based algorithm runs in $2^rN^{O(1)}$ instead of $N^{O(r)}$ time\footnote{Note here that the number of rounds $r$ used in Barak \etal's algorithm is polynomial in the alphabet size $q$.}. It is an interesting question to ask whether our algorithm can also be sped up using their technique.
  \item {\em Can Lemma~\ref{lem:funcbound} be used to prove new approximation guarantees for other problems?} Lemma~\ref{lem:funcbound} is a generic bound on the (multiplicative) difference of expectations of a function on two distributions based on their informational divergence. Hence, it may yield new approximation guarantees for other correlation-based algorithms as well. 
  \item {\em Is it possible to prove a result similar to Lemma~\ref{lem:indrounding} without losing a constant factor?} Lemma~\ref{lem:indrounding} at the heart of our approximatin algorithm has one drawback: when $\delta$ is not zero, we always lose a factor of $\delta^{\frac{\delta}{1 - \delta}}$. While the loss here is only constant (since it is minimized when $\delta \rightarrow 1$ which gives $\delta^{\frac{\delta}{1 - \delta}} \geq 0.367$), it prevents us from getting a QPTAS for non-satisfiable dense {\sc Max $k$-CSP}. If this factor can be removed, we can establish the number of levels needed for any approximation ratio from as large as polynomial in $q$ to as small as any constant.
   \item {\em What is the right dependency on $k$ in the running time of approximation algorithms for dense {\sc Max $k$-CSP}?} While $k$ is typically viewed as a constant, it is still interesting to understand what the best dependency on $k$ in the running time. In particular, our algorithm takes $N^{\Omega(ki)}$ time to approximate {\sc Max $k$-CSP} to within factor of $N^{1/i}$ but the running time lower bound we proved for such approximation ratio is only $N^{\tilde \Omega(i) / \log^2 k}$. Can we close the gap of $k \log^2 k$ in the exponent?
\end{itemize}

\section*{Acknowledgements}
PM would like to thank Aviad Rubinstein and Dana Moshkovitz for useful discussions. He also thanks Grigory Yaroslavtsev for providing insights on his algorithm for dense CSPs in~\cite{Yar14} and Madhur Tulsiani for bringing~\cite{BRS11,GS11} to his attention. In addition, we would like to thank Irit Dinur for sharing with us her result based on gap-ETH.

\bibliographystyle{alpha}
\bibliography{citations}

\appendix

\section{Upper Bound on Value of $\cG^l_k$} \label{app:k-csp-soundness}

We devote this section to the proof of Lemma~\ref{lem:k-csp-value}. We will start by introducing a new notation and proving a simple result that is helpful in proving Lemma~\ref{lem:k-csp-value}.

\subsection{Upper Bound on Value of a Convex Combination of {\sc Max $k$-CSP} Instances}

We first define a notion of a convex combination of distributions and prove a simple lemma regarding value of a {\sc Max $k$-CSP} instance whose distribution is a convex combination of other distributions.

\begin{definition} \label{def:convex-dist}
  Let $\mathcal{X}, \mathcal{Y}_1, \dots, \mathcal{Y}_m$ be distributions on $\Theta$. We write $\mathcal{X} = \alpha_1 \mathcal{Y}_1 + \cdots + \alpha_m \mathcal{Y}_m$ for some $\alpha_1, \dots, \alpha_m \in [0, 1]$ if $\mathcal{X}(\theta) = \alpha_1 \mathcal{Y}_1(\theta) + \cdots + \alpha_m \mathcal{Y}_m(\theta)$ for every $\theta \in \Theta$.
\end{definition}

The above definition almost immediately yield the following upper bound on the value of a game whose distribution is a convex combination of other distributions.

\begin{lemma} \label{lem:convex-dist}
  Let $\cG = (V, \cQ, \{P_S\})$ be any {\sc Max $k$-CSP} instance. Let $\cG_1 = (V, \cQ_1, \{P_S\}), \dots, \cG_m = (V, \cQ_m, \{P_S\})$ be {\sc Max $k$-CSP} instances on the same variables, alphabet set and predicates as $\cG$. If $\cQ = \alpha_1 \cQ_1 + \cdots + \alpha_m \cQ_m$ for some $\alpha_1, \dots, \alpha_m \in [0, 1]$, then $val(\cG) \leq \alpha_1 val(\cG_1) + \cdots + \alpha_m val(\cG_m)$.
\end{lemma}

\begin{proof}
  Let $\phi: V \rightarrow \Sigma$ be the optimal assignment of $\cG$, i.e., $val_{\cG}(\phi) = val(\cG)$. From $\cQ = \alpha_1 \cQ_1 + \cdots + \alpha_m \cQ_m$, it is trivial to see that $val_{\cG}(\phi) = \alpha_1 val_{\cG_1}(\phi) + \cdots + \alpha_m val_{\cG_m}(\phi)$. Hence, we have $val(\cG) = val_{\cG}(\phi) = \alpha_1 val_{\cG_1}(\phi) + \cdots + \alpha_m val_{\cG_m}(\phi) \leq \alpha_1 val(\cG_1) + \cdots + \alpha_m val(\cG_m),$ as desired.
\end{proof}

Lemma~\ref{lem:convex-dist} implies the following corollary, which can be seen as an analogue of Lemma~\ref{lem:inq-cond} for {\sc Max $k$-CSP}.

\begin{corollary} \label{cor:k-cond}
  Let $\cG = (V, \cQ, \{P_S\})$ be any {\sc Max $k$-CSP} instance and let $A$ be any event occurring with non-zero probability $1 - p$ (with respect to $\cQ$). Let $\cQ'$ be the conditional probability $\cQ$ given $A$, i.e., $\cQ'(\tilde x_1, \dots, \tilde x_k) = \Pr_{(x_1, \dots, x_k) \sim \cQ}[x_1 = \tilde x_1 \wedge \cdots \wedge x_k = \tilde x_k \mid A]$. For the game $\cG' = (V, \cQ', \{P_S\})$, we have $val(\cG) \leq val(\cG') + p$.
\end{corollary}

\subsection{Proof of Lemma~\ref{lem:k-csp-value}}

Similar to the proof of our birthday repetition theorems, the proof of Lemma~\ref{lem:k-csp-value} consists of several simple reductions. We will start by defining the intermediate games that we reduce $\cG^l_k$ to. For convenience, let $s = kl/40$. We define the games $\cG^l_{k \mid \geq s}$ and, for every $s < s_1 \leq |X|, s < s_2 \leq |Y|$, $\cG^l_{k \mid s_1, s_2}$ on the same questions, alphabet sets and predicates as $\cG^l_k$; on the other hand, their distributions are defined as follows:
\begin{itemize}
\item The distribution $\cQ^l_{k \mid \geq s}$ of $\cG^l_{k \mid \geq s}$ is the uniform distribution on $\{(S_1, T_1), \dots, (S_k, T_k) \in \binom{X}{l} \times \binom{Y}{l} \mid |S_1 \cup \cdots \cup S_{\lfloor k / 2 \rfloor}|, |T_{\lfloor k / 2 \rfloor + 1} \cup \cdots \cup T_k| \geq s\}$. Notice that $\cQ^l_{k \mid \geq s}$ is simply the distribution $\cQ'$ of $\cG^l_k$ conditioned on the event that $|S_1 \cup \cdots \cup S_{\lfloor k / 2 \rfloor}|, |T_{\lfloor k / 2 \rfloor + 1} \cup \cdots \cup T_k| \geq s$.
\item The distribution $\cQ^l_{k \mid s_1, s_2}$ of $\cG^l_{k \mid s_1, s_2}$ is the uniform distribution on $\{(S_1, T_1), \dots, (S_k, T_k) \in \binom{X}{l} \times \binom{Y}{l} \mid |S_1 \cup \cdots \cup S_{\lfloor k / 2 \rfloor}| = s_1 \wedge |T_{\lfloor k / 2 \rfloor + 1} \cup \cdots \cup T_k| = s_2\}$. Again, $\cQ^l_{k \mid s_1, s_2}$ is simply the distribution $\cQ'$ of $\cG^l_k$ conditioned on the event that $|S_1 \cup \cdots \cup S_{\lfloor k / 2 \rfloor}| = s_1$ and $|T_{\lfloor k / 2 \rfloor + 1} \cup \cdots \cup T_k| = s_2$.
\end{itemize}

We now describe how the values of these games are related and provide proof overviews for such relations.

\begin{lemma} \label{lem:fully-dense-em}
  $val(\cG^l_{k \mid s_1, s_2}) \leq val(\cG^{s_1 \times s_2})$.
\end{lemma}

{\bf Proof Idea.} We can naturally create a mixed strategy $\phi'$ in $\cG^{s_1 \times s_2}$ from any strategy $\phi$ in $\cG^l_{k \mid s_1, s_2}$. It is easy to show that $\phi'$ has value with respect to $\cG^{s_1 \times s_2}$ more than the value of $\phi$ with respect to $\cG^l_{k \mid s_1, s_2}$.

\begin{lemma} \label{lem:fully-dense-comb}
  $val(\cG^l_{k \mid \geq s}) \leq \max_{s \leq s_1 \leq |X|, s \leq s_2 \leq |Y|} val(\cG^l_{k \mid s_1, s_2})$.
\end{lemma}

{\bf Proof Idea.} The proof of this lemma is simply to realize the fact that $\cQ^l_{k \mid \geq s}$ and be written as convex combination of $\cQ^l_{k \mid s_1, s_2}$ for $s \leq s_1 \leq |X|, s \leq s_2 \leq |Y|$. This observation, combined with Lemma~\ref{lem:convex-dist}, immediately yields Lemma~\ref{lem:fully-dense-comb}.

\begin{lemma} \label{lem:fully-dense-cond}
  $val(\cG^l_k) \leq val(\cG^l_{k \mid \geq s}) + 2\exp\left(-\frac{kl}{8}\right)$
\end{lemma}

{\bf Proof Idea.} Thanks to Corollary~\ref{cor:k-cond}, we only need to show that the probability that $|S_1 \cup \cdots \cup S_{\lfloor k / 2 \rfloor}|, |T_{\lfloor k / 2 \rfloor + 1} \cup \cdots \cup T_k| \geq s$ is at most $2\exp\left(-\frac{kl}{8}\right)$. This can easily be shown using standard probabilistic techniques.

Before we give the full proofs of the above three lemmas, we first give a proof of Lemma~\ref{lem:k-csp-value}.

\begin{proofof}[Lemma~\ref{lem:k-csp-value}]
  By the three lemmas above, $val(\cG^l_k)$ can be bounded as follows:
  \begin{align*}
    val(\cG^l_k) 
    &\leq 2\exp\left(-\frac{kl}{8}\right) + \max_{s \leq s_1 \leq |X|, s \leq s_2 \leq |Y|} val(\cG^{s_1 \times s_2}) \\
    (\text{Theorem}~\ref{thm:birthday-proj}) 
    &\leq 2\exp\left(-\frac{kl}{8}\right) + 2(1 - \varepsilon/2)^{\frac{\alpha \varepsilon^3 s^2 |E|}{d_{max}|X||Y|}}. \\
    (\text{Since } s \leq |X|, |Y|) &\leq 4(1 - \varepsilon/2)^{\frac{\alpha \varepsilon^3 s^2 |E|}{d_{max}|X||Y|}}
    = \left(2(1 - \varepsilon/2)^{\frac{\gamma \varepsilon^3 k^2l^2 |E|}{d_{max}|X||Y|}}\right)^2
  \end{align*}
  where $\gamma = \alpha / 3200$. Note that, if $2(1 - \varepsilon/2)^{\frac{\gamma \varepsilon^3 k^2l^2 |E|}{d_{max}|X||Y|}} \geq 1$, $val(\cG^l_k) \leq 1 \leq 2(1 - \varepsilon/2)^{\frac{\gamma \varepsilon^3 k^2l^2 |E|}{d_{max}|X||Y|}}$. Otherwise, $val(\cG^l_k) \leq (2(1 - \varepsilon/2)^{\frac{\gamma \varepsilon^3 k^2l^2 |E|}{d_{max}|X||Y|}})^2 \leq 2(1 - \varepsilon/2)^{\frac{\gamma \varepsilon^3 k^2l^2 |E|}{d_{max}|X||Y|}}$. This completes our proof for Lemma~\ref{lem:k-csp-value}.
\end{proofof}

We now turn our attention to the proofs of Lemma~\ref{lem:fully-dense-em},~\ref{lem:fully-dense-comb}, and~\ref{lem:fully-dense-cond}.

\subsubsection{$\cG^{s_1 \times s_2}$ vs $\cG^l_{k \mid s_1, s_2}$: Embedding of Birthday Repetition}

\begin{proofof}[Lemma~\ref{lem:fully-dense-em}]
  Let $\phi$ be the optimal assignment of $\cG^l_{k \mid s_1, s_2} = (V', \cQ^l_{k \mid s_1, s_2}, \{P'_S\})$. We define a mixed strategy $\Phi$ of $\cG^{s_1 \times s_2}$ by defining a sampling process for $\tilde \phi \sim \Phi$ as follows. For each $S \in \binom{X}{s_1}$, sample subsets $S_1, \dots, S_k \in \binom{X}{l}$ such that $S_1 \cup \cdots \cup S_{\lfloor k / 2 \rfloor} = S$ uniformly among all such subsets and sample $T_1, \dots, T_{\lfloor k / 2 \rfloor}$ uniformly independently at random from $\binom{Y}{l}$. If $\phi((S_1, T_1)), \dots, \phi((S_{\lfloor k / 2 \rfloor}, T_{\lfloor k / 2 \rfloor}))$ are consistent, we set $\tilde \phi(S)$ to be $(\phi((S_1, T_1)) \circ \cdots \circ \phi((S_{\lfloor k / 2 \rfloor}, T_{\lfloor k / 2 \rfloor})))|_S$. Otherwise, we set $\tilde \phi(S)$ arbitrarily. Similarly, for each $T \in \binom{Y}{s_2}$, sample subsets $T_{\lfloor k / 2 \rfloor + 1}, \dots, T_k \in \binom{Y}{l}$ such that $T_{\lfloor k / 2 \rfloor + 1} \cup \cdots \cup T_k = T$ uniformly among all such subsets and $S_{\lfloor k / 2 \rfloor + 1}, \dots, S_k \sim \binom{X}{l}$ independently. Then, set $\tilde \phi(T) = (\phi((S_{\lfloor k / 2 \rfloor + 1}, T_{\lfloor k / 2 \rfloor + 1})) \circ \cdots \circ \phi((S_k, T_k)))|_T$ if $\phi((S_{\lfloor k / 2 \rfloor + 1}, T_{\lfloor k / 2 \rfloor + 1})), \dots, \phi((S_k, T_k))$ are consistent and arbitrarily otherwise.

  We will now show that the expected value of $\tilde \phi$ is at least $val(\cG^l_{k \mid s_1, s_2})$, which immediately implies Lemma~\ref{lem:fully-dense-em}. For convenient, below we write $\phi((S_1, T_1)) \circ \cdots \circ \phi((S_{\lfloor k / 2 \rfloor}, T_{\lfloor k / 2 \rfloor}))$ even when $\phi((S_1, T_1)), \dots, \phi((S_{\lfloor k / 2 \rfloor}, T_{\lfloor k / 2 \rfloor}))$ are inconsistent; in this case, $\phi((S_1, T_1)) \circ \cdots \circ \phi((S_{\lfloor k / 2 \rfloor}, T_{\lfloor k / 2 \rfloor}))$ just represents an arbitrary assignment to $S_1 \cup \cdots \cup S_{\lfloor k / 2 \rfloor} \cup T_1 \cup \cdots \cup T_{\lfloor k / 2 \rfloor}$. Same notation applies for $\phi((S_{\lfloor k / 2 \rfloor + 1}, T_{\lfloor k / 2 \rfloor + 1})) \circ \cdots \circ \phi((S_k, T_k))$. We can rewrite $\E_{\tilde \phi}[val(\tilde \phi)]$ as follows:
  \begin{align*}
    \E_{\tilde \phi}[val(\tilde \phi)]
    = \E_{\tilde \phi}[\E_{(S, T) \sim \cQ^{s_1 \times s_2}}[P^{s_1 \times s_2}_{(S, T)}(\tilde \phi(S), \tilde \phi(T))]]
    = \E_{(S, T) \sim \cQ^{s_1 \times s_2}}[\E_{\tilde \phi}[P^{s_1 \times s_2}_{(S, T)}(\tilde \phi(S), \tilde \phi(T))]].
  \end{align*}

  $\E_{\tilde \phi}[P^{s_1 \times s_2}_{(S, T)}(\tilde \phi(S), \tilde \phi(T))]$ can be further written as
  \begin{align*}
    \E_{S_1, \dots, S_k, T_1, \dots, T_k}[P^{s_1 \times s_2}_{(S, T)}((\phi((S_1, T_1)) \circ \cdots \circ \phi((S_{\lfloor k / 2 \rfloor}, T_{\lfloor k / 2 \rfloor})))|_S, (\phi((S_{\lfloor k / 2 \rfloor + 1}, T_{\lfloor k / 2 \rfloor + 1})) \circ \cdots \circ \phi((S_k, T_k)))|_T)]
  \end{align*}

  where $S_1, \dots, S_k, T_1, \dots, T_k$ are drawn depending on $S$ and $T$ as described earlier.

  Now, observe that $P^{s_1 \times s_2}_{(S, T)}((\phi((S_1, T_1)) \circ \cdots \circ \phi((S_{\lfloor k / 2 \rfloor}, T_{\lfloor k / 2 \rfloor})))|_S, (\phi((S_{\lfloor k / 2 \rfloor + 1}, T_{\lfloor k / 2 \rfloor + 1})) \circ \cdots \circ \phi((S_k, T_k)))|_T) \geq P'_{((S_1, T_1), \dots, (S_k, T_k))}(\phi((S_1, T_1)), \dots, \phi((S_k, T_k)))$ because $P'$ checks both that $\phi((S_1, T_1)), \dots, \phi((S_k, T_k))$ are consistent and that all the constraints between $S_1 \cup \cdots \cup S_k$ and $T_1 \cup \cdots \cup T_k$ are satisfied. Thus, we have
  \begin{align*}
    \E_{\tilde \phi}[val(\tilde \phi)] &\geq  \E_{(S, T) \sim \cQ^{s_1 \times s_2}}[\E_{S_1, \dots, S_k, T_1, \dots, T_k}[P'_{((S_1, T_1), \dots, (S_k, T_k))}(\phi((S_1, T_1)), \dots, \phi((S_k, T_k)))]].
  \end{align*}

  Finally, notice that $S_1, \dots, S_k, T_1, \dots, T_k$ sampled in the above expression is exactly the same as sampled by $((S_1, T_1), \dots, (S_k, T_k)) \sim \cQ^l_{k \mid s_1, s_2}$. As a result, we can conclude that
  \begin{align*}
    \E_{\tilde \phi}[val(\tilde \phi)] \geq  \E_{((S_1, T_1), \dots, (S_k, T_k)) \sim \cQ^l_{k \mid s_1, s_2}}[P'_{((S_1, T_1), \dots, (S_k, T_k))}(\phi((S_1, T_1)), \dots, \phi((S_k, T_k)))]
    = val_{\cG^l_{k \mid s_1, s_2}}(\phi) = val(\cG^l_{k \mid s_1, s_2}),
  \end{align*}
  completing the proof of this lemma.
\end{proofof}

\subsubsection{$\cG^l_{k \mid s_1, s_2}$ vs $\cG^l_{k \mid \geq s}$: Convex Combination of Distributions}

\begin{proofof}[Lemma~\ref{lem:fully-dense-comb}]
  Observe that $\cQ^l_{k \mid \geq s}$ is uniform on $\supp(\cQ^l_{k \mid \geq s})$ and $\cQ^l_{k \mid s_1, s_2}$ is uniform on $\supp(\cQ^l_{k \mid s_1, s_2})$. Moreover, it is easy to see that $\supp(\cQ^l_{k \mid \geq s}) = \bigcup_{s \leq s_1 \leq |X|, s \leq s_2 \leq |Y|} \supp(\cQ^l_{k \mid s_1, s_2})$ and that $\supp(\cQ^l_{k \mid s_1, s_2})$'s are disjoint for all $s \leq s_1 \leq |X|, s \leq s_2 \leq |Y|$. Hence, we can write $\cQ^l_{k \mid \geq s}$ as
  \begin{align*}
    \cQ^l_{k \mid \geq s} &= \sum_{s \leq s_1 \leq |X|, s \leq s_2 \leq |Y|} \alpha_{s_1, s_2} \cQ^l_{k \mid s_1, s_2}
  \end{align*}
  where $\alpha_{s_1, s_2} = \frac{|\supp(\cQ^l_{k \mid s_1, s_2})|}{|\supp(\cQ^l_{k \mid \geq s})|}$. As a result, from Lemma~\ref{lem:convex-dist}, we have
  \begin{align*}
    val(\cG^l_{k \mid \geq s}) \leq \sum_{s \leq s_1 \leq |X|, s \leq s_2 \leq |Y|} \alpha_{s_1, s_2} val(\cG^l_{k \mid s_1, s_2})
    \leq \max_{s \leq s_1 \leq |X|, s \leq s_2 \leq |Y|} val(\cG^l_{k \mid s_1, s_2})
  \end{align*}
  as desired.
\end{proofof}

\subsubsection{$\cG^l_{k \mid \geq s}$ vs $\cG^l_k$: Lower Bound on Size of Union of Random Subsets}

Before we prove Lemma~\ref{lem:fully-dense-cond}, we will prove the following lemma, which is central to the proof of Lemma~\ref{lem:fully-dense-cond}.

\begin{lemma} \label{lem:union-size}
  Let $U$ be a set and let $a, b$ be any non-negative integers such that $ab \leq |U|$. If $U_1, \dots, U_a$ are independently randomly sampled from $\binom{U}{b}$, then $\Pr[|U_1 \cup \cdots \cup U_a| \leq \frac{ab}{10}] \leq \exp\left(-\frac{ab}{2}\right)$.
\end{lemma}

\begin{proof}
  We can rewrite $\Pr[|U_1 \cup \cdots \cup U_a| \leq \frac{ab}{10}]$ as follows:
  \begin{align*}
    \Pr[|U_1 \cup \cdots \cup U_a| \leq \frac{ab}{10}] 
    &= \Pr\left[\bigvee_{U^* \in \binom{U}{\lfloor ab / 10 \rfloor}} U_1 \subseteq U^*, \dots, U_a \subseteq U^*  \right] \\
    (\text{Union Bound}) &\leq \sum_{U^* \in \binom{U}{\lfloor ab / 10 \rfloor}} \Pr[U_1 \subseteq U^*, \dots, U_a \subseteq U^*] \\
    (U_1, \dots, U_a \text{ are chosen independently}) &= \sum_{U^* \in \binom{U}{\lfloor ab / 10 \rfloor}} \Pr[U_1 \subseteq U^*] \cdots \Pr[U_a \subseteq U^*] \\
    &= \binom{|U|}{\lfloor ab / 10 \rfloor}\left(\frac{\binom{\lfloor ab / 10 \rfloor}{b}}{\binom{|U|}{b}}\right)^a \\
    (\text{Fact}~\ref{fact:binom-approx}) &\leq \left(\frac{e|U|}{\lfloor ab / 10 \rfloor}\right)^{\lfloor ab / 10 \rfloor}\left(\frac{\lfloor ab / 10 \rfloor}{|U|}\right)^{ab} \\
    &\leq \exp(-ab/2).
  \end{align*}
\end{proof}

Now, we are ready to prove Lemma~\ref{lem:fully-dense-cond}.

\begin{proofof}[Lemma~\ref{lem:fully-dense-cond}]
  As mentioned earlier, $\cQ^l_{k \mid \geq s}$ is the distribution $\cQ'$ of $\cG^l_k$ conditioned on $|S_1 \cup \cdots \cup S_{\lfloor k / 2 \rfloor}|, |T_{\lfloor k / 2 \rfloor + 1} \cup \cdots \cup T_k| \geq s$. Let us call this event $A$. Thus, from Corollary~\ref{cor:k-cond}, it is enough to show that $\Pr[\neg A] \leq 2\exp\left(-\frac{kl}{8}\right)$. This can easily be proved as follows.
  \begin{align*}
    \Pr[\neg A] 
    &\leq \Pr[|S_1 \cup \cdots \cup S_{\lfloor k / 2 \rfloor}| < s] + \Pr[|T_{\lfloor k / 2 \rfloor + 1} \cup \cdots \cup T_k| < s] \\
    (\text{Our Choice of } s) &\leq \Pr[|S_1 \cup \cdots \cup S_{\lfloor k / 2 \rfloor}| < \frac{\lfloor k / 2 \rfloor l}{10}] + \Pr[|T_{\lfloor k / 2 \rfloor + 1} \cup \cdots \cup T_k| < \frac{\lceil k / 2 \rceil l}{10}] \\
    (\text{Lemma}~\ref{lem:union-size}) &\leq \exp\left(-\frac{\lfloor k / 2 \rfloor l}{2}\right) + \exp\left(-\frac{\lceil k / 2 \rceil l}{2}\right) \leq 2\exp\left(-\frac{kl}{8}\right).
  \end{align*}
\end{proofof}

\section{Improved Bound on Total Correlation of Conditioned Sherali-Adams Solution} \label{app:corr}

To prove Lemma~\ref{lem:corr-decrease}, it is not hard to see that, if we can prove the bound of the fully-dense case, then the $\Delta$-dense case follows easily. More specifically, we will prove the following lemma.

\begin{lemma} \label{lem:corr-decrease-fully-dense}
  Let $\mu$ be any $r$-level SA solution of a fully-dense {\sc Max $k$-CSP} instance $\cG = (V, V^k, \{P_S\})$. Then, for any $0 < l \leq r - k$, there exists $t \leq l$ such that
  $\E_{T \sim V^t, \phi_T \sim \Sigma^T}[C(\mu|\phi_T)] \leq (k^2 \log q)/l.$
\end{lemma}

Lemma~\ref{lem:corr-decrease-fully-dense}  implies Lemma~\ref{lem:corr-decrease} since, if $\cW$ is a distribution of a $\Delta$-dense {\sc Max $k$-CSP} instance, then $\E_{S \sim \cW}[C_{\mu}(x_S)] \leq \frac{1}{\Delta} \E_{S \sim V^k}[C_{\mu}(x_S)]$ where the inequality comes from $\Delta \cdot \cW(S) \leq 1/n^k$ for all $S \in V^k$.

Before we prove Lemma~\ref{lem:corr-decrease-fully-dense}, we will define entropy and mutual information for a tuple of variables in a similar fashion as we did for total correlation. More specifically, for tuple $S = (x_{i_1}, \dots, x_{i_j}) \in V^j$ of size $j \leq k$, $H_\mu(x_S)$ and $I_\mu(x_S)$ are defined as $H(\sigma_{i_1}, \dots, \sigma_{i_j})$ and $I(\sigma_{i_1}; \dots; \sigma_{i_j})$ respectively where $\sigma_{i_1}, \dots, \sigma_{i_j}$ are jointly sampled from $\mathcal{X}_{\{x_{i_1}, \dots, x_{i_j}\}}$. We also define conditioned entropy, mutual information and total correlation in similar manner; for $S = (x_{i_1}, \dots, x_{i_j})$ and $T = (x_{i_{j + 1}}, \dots, x_{i_{j + l}})$ where $j + l \leq k$, $H_\mu(x_S | x_T), I_\mu(x_S | x_T)$ and $C_\mu(x_S | x_T)$ are defined as $H(\sigma_{i_1}, \dots, \sigma_{i_j} | \sigma_{i_{j + 1}}, \dots, \sigma_{i_{j + l}}), I(\sigma_{i_1}; \dots; \sigma_{i_j} | \sigma_{i_{j + 1}}, \dots, \sigma_{i_{j + l}})$ and $C(\sigma_{i_1}; \dots; \sigma_{i_j} | \sigma_{i_{j + 1}}, \dots, \sigma_{i_{j + l}})$ where $\sigma_{i_1}, \dots, \sigma_{i_{j + l}}$ are jointly sampled from $\mathcal{X}_{\{x_{i_1}, \dots, x_{i_{j + l}}\}}$.

To help us prove Lemma~\ref{lem:corr-decrease-fully-dense}, we will state another lemma, which was proved implicitly in~\cite[Lemma 3.3]{YZ14}. It can be proven easily by rearranging the identity in Lemma~\ref{lem:total-cor-mutual-info}. We do not provide a full proof here.

\begin{lemma} \label{lem:expected-total-cor}
  For any $t, d > 0$, we have $\E_{S \sim V^d}[C_{\mu}(x_S)] = \sum_{2 \leq r \leq d} \binom{d}{r} \E_{R \sim V^r} [I_\mu(x_R)]$.
\end{lemma}

Now, we are ready to prove Lemma~\ref{lem:corr-decrease-fully-dense}.

\begin{proofof}[Lemma~\ref{lem:corr-decrease-fully-dense}]
  Consider $\sum_{0 \leq t \leq l} \E_{T \sim V^t, \phi_T \sim \Sigma^T}[C(\mu|\phi_T)]$. We can rewrite it as
  \begin{align*}
    \sum_{0 \leq t \leq l} \E_{T \sim V^t, \phi_T \sim \Sigma^T}[C(\mu|\phi_T)] &= \sum_{0 \leq t \leq l} \E_{T \sim V^t} \E_{S \sim V^k} [C_\mu(x_S | x_T)] \\
    (\text{Lemma}~\ref{lem:expected-total-cor}) &= \sum_{2 \leq r \leq k} \binom{k}{r} \sum_{0 \leq t \leq l} \E_{T \sim V^t} \E_{R \sim V^r} [I_\mu(x_R | x_T)].
  \end{align*}

  Moreover, by Lemma~\ref{lem:cond-mutual-info}, we have
  \begin{align*}
    \sum_{0 \leq t \leq l} \E_{T \sim V^t} \E_{R \sim V^r} [I_\mu(x_R | x_T)]
    &= \sum_{0 \leq t \leq l} \left(\E_{T \sim V^t} \E_{R \sim V^{r - 1}} [I_\mu(x_R | x_T)] - \E_{T \sim V^{t + 1}} \E_{R \sim V^{r - 1}} [I_\mu(x_R | x_T)]\right) \\
    &= \E_{R \sim V^{r - 1}} [I_\mu(x_R)] - \E_{T \sim V^{l + 1}} \E_{R \sim V^{r - 1}} [I_\mu(x_R | x_T)].
  \end{align*}

  Hence, we have
  \begin{align*}
    \sum_{0 \leq t \leq l} \E_{T \sim V^t, \phi_T \sim \Sigma^T}[C(\mu|\phi_T)] = \sum_{2 \leq r \leq k} \binom{k}{r} \E_{R \sim V^{r - 1}} [I_\mu(x_R)] - \sum_{2 \leq r \leq k} \binom{k}{r} \E_{T \sim V^{l + 1}} \E_{R \sim V^{r - 1}} [I_\mu(x_R | x_T)].
  \end{align*}

  Using Pascal's identity, i.e., $\binom{k}{r} = \binom{k - 1}{r - 1} + \binom{k - 1}{r} = \cdots = \binom{k - 1}{r - 1} + \cdots + \binom{r - 1}{r - 1}$, the term $\sum_{2 \leq r \leq k} \binom{k}{r} \E_{R \sim V^{r - 1}} [I_\mu(x_R)]$ can be further rearranged as follows.
  \begin{align*}
    \sum_{2 \leq r \leq k} \binom{k}{r} \E_{R \sim V^{r - 1}} [I_\mu(x_R)] &= \binom{k}{2} \E_{i \sim V} [I_\mu(x_i)] + \sum_{3 \leq r \leq k} \binom{k}{r} \E_{R \sim V^{r - 1}} [I_\mu(x_R)] \\
    (\text{Pascal's identity}) &= \binom{k}{2} \E_{i \sim V} [H_\mu(x_i)] + \sum_{2 \leq d \leq k - 1} \sum_{3 \leq r \leq d + 1} \binom{d}{r - 1} \E_{R \sim V^{r - 1}} [I_\mu(x_R)] \\
    (\text{Lemma}~\ref{lem:expected-total-cor}) &= \binom{k}{2} \E_{i \sim V} [H_\mu(x_i)] + \sum_{2 \leq d \leq k - 1} \E_{S \sim V^d} [C_\mu(x_S)].
  \end{align*}

  Similarly, we can write $\sum_{2 \leq r \leq k} \binom{k}{r} \E_{T \sim V^{l + 1}} \E_{R \sim V^{r - 1}} [I_\mu(x_R | x_T)]$ as
  \begin{align*}
    \sum_{2 \leq r \leq k} \binom{k}{r} \E_{T \sim V^{l + 1}} \E_{R \sim V^{r - 1}} [I_\mu(x_R | x_T)]
    = \binom{k}{2} \E_{T \sim V^{l + 1}} \E_{i \sim V} [H_\mu(x_i | x_T)] + \sum_{2 \leq d \leq k - 1} \E_{T \sim V^{l + 1}} \E_{S \sim V^d} [C_\mu(x_S | x_T)]
    \geq 0
  \end{align*}
  where the inequality comes from non-negativity of entropy and total correlation.

  Hence, we can conclude that
  \begin{align*}
    \sum_{0 \leq t \leq l} \E_{T \sim V^t, \phi_T \sim \Sigma^T}[C(\mu|\phi_T)] &\leq \binom{k}{2} \E_{i \sim V} [H_\mu(x_i)] + \sum_{2 \leq d \leq k - 1} \E_{S \sim V^d} [C_\mu(x_S)] \\
    (\text{Lemma}~\ref{lem:total-cor-entropy}) &\leq \binom{k}{2} \E_{i \sim V} [H_\mu(x_i)] + \sum_{2 \leq d \leq k - 1} \E_{S \sim V^d} [\sum_{i \in S} H_\mu(x_i)] \\
    &\leq \binom{k}{2} \log q + \sum_{2 \leq d \leq k - 1} d \log q \\
    &\leq k^2 \log q.
  \end{align*}
  Note that the second-to-last inequality comes from a well-known fact that, for any random variable $y$ drawn according to any distribution $\mathcal{Y}$, $H(y) \leq \log |\supp(\mathcal{Y})|$. In particular, we have $H(x_i) \leq \log q$ for every $i \in V$.

  Hence, there exists $t \leq l$ such that $\E_{T \sim V^t, \phi_T \sim \Sigma^T}[C(\mu|\phi_T)] \leq (k^2 \log q) / l$, as desired.
\end{proofof}

\section{Reduction from {\sc Densest $k$-Subhypergraph} to CSP} \label{app:red-dense-hyper}

In this section, we describe a variant of Charikar \etal's~\cite{CHK} reduction and prove Lemma~\ref{lem:dense-hyper-reduction}.

\begin{proofof}[Lemma~\ref{lem:dense-hyper-reduction}]
  Given a $d$-uniform hypergraph $(V, E)$, the reduction simply proceeds as follows. First, we randomly partition the set of vertices $V$ into $k$ subsets $S_1, \dots, S_k$ by, for each vertex $v \in V$, select $i_v \in [k]$ uniformly independently at random and put $v$ into $S_{i_v}$. Then, the set of variables $V'$ is $[k]$ and the alphabet set $\Sigma$ of the resulting instance is the vertex set $V$. Finally, $P_{(i_1, \dots, i_d)}(v_1, \dots, v_d)$ is one if and only if $\{v_1, \dots, v_d\} \in E$ and $v_1 \in S_1, \dots, v_d \in S_d$, and is zero otherwise.

  The first property in Lemma~\ref{lem:dense-hyper-reduction} is obviously satisfied. The second property can also be achieved easily: given an assignment $\phi$, pick the set $S$ to be $\{\phi(1), \dots, \phi(k)\}$. If $\phi(i)$'s are not all distinct, add arbritary vertices into $S$ to make its size $k$. Clearly, the density of subhypergraph induced on $S$ is at least $val(\phi)$.

  Next, assume that $k \geq 8d^2$ and that the densest $k$-subhypergraph of $(V, E)$ has density $\Delta$. Let $S^*$ and $E^*$ be the set of vertices and edges of the subhypergraph. Let $t = 16d$ and $E^*_t = \{\{v_1, \dots, v_d\} \in E^* \mid |S^* \cap S_{i_{v_1}}|, \dots, |S^* \cap S_{i_{v_d}}| \leq t \text{ and } i_{v_1}, \dots, i_{v_d} \text{are all distinct}\}$. Observe that $val(\cG) \geq \frac{|E^*_t|}{t^dk^d}$ because, if we define a mix strategy $\tilde \phi$ where $\tilde \phi(i)$ is randomly uniformly picked from $S^* \cap S_i$ (and arbitrarily if the intersection is empty), then the expected number of constraints satisfied by $\tilde \phi$ is at least $|E^*_t|/t^d$. Since $\frac{|E^*|}{t^dk^d} = \frac{\Delta}{t^d d!} = \Delta/2^{O(d \log d)}$, to prove the last property, it is enough to show that $\Pr[|E^*_t| \geq |E^*|/2] \geq 1/2$.

  In other words, we want to show that $\Pr[|E^* - E^*_t| > |E^*|/2] \leq 1/2$. By Markov's inequality, it is enough for us to show that $\E[|E^* - E^*_t|] \leq |E^*|/4$. Again, to prove this, it is enough to show that for each $(v_1, \dots, v_d) \in E^*$, the probability that $(v_1, \dots, v_d) \notin E^*_t$ is at most $1/4$. This probability can be bounded as follows.
  \begin{align*}
    \Pr[(v_1, \dots, v_d) \notin E^*_t]
    &= \Pr[(|S^* \cap S_{i_{v_1}}| > t) \vee \cdots \vee (|S^* \cap S_{i_{v_k}}| > t) \vee (i_{v_1}, \dots, i_{v_d} \text{are not all distinct})] \\
    (\text{Union Bound}) &\leq \left(\sum_{j \in [d]} \Pr[|S^* \cap S_{i_{v_j}}| > t]\right) + \Pr[i_{v_1}, \dots, i_{v_d} \text{are not all distinct}].
  \end{align*}

  For each $j \in [d]$, Markov's inequality gives the following upper bound on $\Pr[|S^* \cap S_{i_{v_j}}| > t]$.
  \begin{align*}
    \Pr[|S^* \cap S_{i_{v_j}}| > t] \leq \frac{\E[|S^* \cap S_{i_{v_j}}|]}{t} = \frac{1 + (k - 1)/k}{t} \leq \frac{1}{8d}.
  \end{align*}

  Moreover, the probability that $i_{v_1}, \dots, i_{v_d}$ are not all distinct can be easily computed as follows.
  \begin{align*}
    \Pr[i_{v_1}, \dots, i_{v_d} \text{are not all distinct}]
    &= 1 - \Pr[i_{v_1}, \dots, i_{v_d} \text{are all distinct}] \\
    &= 1 - \frac{k(k - 1) \dots (k - d)}{k^d} \\
    &\leq 1 - \left(\frac{k - d}{k}\right)^d \\
    (\text{Bernoulli's inequality}) &\leq 1 - (1 - d^2/k) \\
    (\text{From } k \geq 8d^2) &\leq 1/8.
  \end{align*}

  Hence, we have $\Pr[(v_1, \dots, v_d) \notin E^*_t] \leq \left(\sum_{j \in [d]} \frac{1}{8d}\right) + 1/8 = 1/4$, completing the proof of the lemma.
\end{proofof}

\section{Reductions in the Lasserre Hierarchies}

In this section, we provide proofs for simple results we use regarding the Lasserre hierarchies.

\subsection{Completeness of Lasserre Solution Through Reductions} \label{app:lasserre-reduction-completeness}

In this section, we prove Lemma~\ref{lem:lasserre-reduction-completeness}. While the lemma is not explicitly stated anywhere, it was proven implicitly before (e.g. in~\cite{Tul09}). Before we prove the lemma, we will prove a few useful facts regarding a vector solution of Lasserre relaxations of {\sc Max $k$-CSP}.

Throughout the following lemmas, suppose that $\{U_{(S, \phi_S)}\}$ is a (not necessary complete) vector solution of $r$-level Lasserre relaxation of a {\sc Max $k$-CSP} instance $\cG = (V, \cW, \{P_S\})$. Note that $\cG$ here is different than $\cG$ in Lemma~\ref{lem:lasserre-reduction-completeness}. For $x \in V, \sigma \in \Sigma$, we use notation $x \to \sigma$ to denote an element $f$ of $\Sigma^{\{x\}}$ such that $f(x) = \sigma$. For brevity, we also abbreviate $U_{(\{x\}, x \to \sigma)}$ as $U_{(x, \sigma)}$.

\begin{lemma} \label{lem:lasserre-exc-var}
  For every $S \in \binom{V}{[r - 1]}$, $\phi_S \in \Sigma^S, x \in V - S$, we have $\sum_{\sigma \in \Sigma} \|U_{(S \cup \{x\}, \phi_S \circ (x \to \sigma))}\|^2 = \|U_{(S, \phi_S)}\|^2$.
\end{lemma}

\begin{proof}
  First, we will show that, for any $x \in V$, $\sum_{\sigma \in \Sigma} U_{(x,  \sigma)} = U_{(\emptyset, \emptyset)}$. This can be done by rearranging $\|\sum_{\sigma \in \Sigma} U_{(x,  \sigma)} - U_{(\emptyset, \emptyset)}\|^2$ as follows.
  \begin{align*}
    \|\sum_{\sigma \in \Sigma} U_{(x,  \sigma)} - U_{(\emptyset, \emptyset)}\|^2 &= \sum_{\sigma \in \Sigma} \|U_{(x,  \sigma)}\|^2 + 2\sum_{\sigma, \sigma' \in \Sigma \atop \sigma \ne \sigma'} \langle U_{(x,  \sigma)}, U_{(x,  \sigma')}\rangle  - 2\sum_{\sigma \in \Sigma} \langle U_{(x,  \sigma)}, U_{(\emptyset, \emptyset)}\rangle  + \|U_{(\emptyset, \emptyset)}\|^2 \\
    &= \sum_{\sigma \in \Sigma} \|U_{(x,  \sigma)}\|^2 - 2\sum_{\sigma \in \Sigma} \|U_{(x,  \sigma)}\|^2 + \|U_{(\emptyset, \emptyset)}\|^2 \\
    &= 1 - 2 + 1 = 0.
  \end{align*}

  As a result, we can write $\sum_{\sigma \in \Sigma} \|U_{(S \cup \{x\}, \phi_S \circ (x \to \sigma))}\|^2$ as
  \begin{align*}
    \sum_{\sigma \in \Sigma} \|U_{(S \cup \{x\}, \phi_S \circ (x \to \sigma))}\|^2 &= \sum_{\sigma \in \Sigma} \langle U_{(S, \phi_S)}, U_{(x, \sigma)}\rangle
    = \langle U_{(S, \phi_S)}, \sum_{\sigma \in \Sigma} U_{(x, \sigma)}\rangle
    = \langle U_{(S, \phi_S)}, U_{(\emptyset, \emptyset)}\rangle
    = \|U_{(S, \phi_S)}\|^2.
  \end{align*}
\end{proof}

The following lemma is almost immediate from Lemma~\ref{lem:lasserre-exc-var}.

\begin{lemma} \label{lem:lasserre-square-sum}
  For every $S \in \binom{V}{r}$, we have $\sum_{\phi_S \in \Sigma^S} \|U_{(S, \phi_S)}\|^2 = 1$
\end{lemma}

\begin{proof}
  Suppose that $S = \{x_1, \dots, x_l\}$ for some $l \leq r$. Let $S_i = \{x_1, \dots, x_i\}$ for every $i \in [l]$. We can use Lemma~\ref{lem:lasserre-exc-var} to write $\sum_{\phi_S \in \Sigma^S} \|U_{(S, \phi_S)}\|^2$ as
  \begin{align*}
    \sum_{\phi_S \in \Sigma^S} \|U_{(S, \phi_S)}\|^2
    &= \sum_{\phi_{S_{l - 1}} \in \Sigma^{S_{l - 1}}} \sum_{\sigma \in \Sigma} \|U_{(S_{l - 1} \cup \{x_l\}, \phi_{S_{l - 1} \circ (x_l \to \sigma)})}\|^2 \\
    (\text{Lemma~\ref{lem:lasserre-exc-var}}) &= \sum_{\phi_{S_{l - 1}} \in \Sigma^{S_{l - 1}}} \|U_{(S_{l - 1}, \phi_{S_{l - 1}})}\|^2 \\
    &\vdots \\
    &= \sum_{\phi_{S_1} \in \Sigma^{S_1}} \|U_{(S_1, \phi_{S_1})}\|^2 = 1.
  \end{align*}
\end{proof}

The last intermediate lemma we prove is a lemma regarding characterization of a complete solution.

\begin{lemma} \label{lem:lasserre-complete-char}
  $\{U_{(S, \phi_S)}\}$ is complete iff $U_{(S, \phi_S)} = 0$ for every $S \in \supp(\cW), \phi_S \in \Sigma^S$ such that $P_S(\phi_S) \ne 1$.
\end{lemma}

\begin{proof}
  From Lemma~\ref{lem:lasserre-square-sum}, we have
  \begin{align*}
    1 - \E_{S \sim \cW}[\sum_{\phi_S \in \Sigma^S} ||U_{(S, \phi_S)}||^2 P_{S}(\phi_S)]
    &= \E_{S \sim \cW}[1 - \sum_{\phi_S \in \Sigma^S} ||U_{(S, \phi_S)}||^2 P_{S}(\phi_S)] \\
    (\text{Lemma}~\ref{lem:lasserre-square-sum}) &= \E_{S \sim \cW}[\sum_{\phi_S \in \Sigma^S} ||U_{(S, \phi_S)}||^2 - \sum_{\phi_S \in \Sigma^S} ||U_{(S, \phi_S)}||^2 P_{S}(\phi_S)] \\
    &= \E_{S \sim \cW}[\sum_{\phi_S \in \Sigma^S} ||U_{(S, \phi_S)}||^2 (1 - P_{S}(\phi_S))].
  \end{align*}
  Since $P_{S}(\phi_S) \in [0, 1]$, we can conclude that $\E_{S \sim \cW}[\sum_{\phi_S \in \Sigma^S} ||U_{(S, \phi_S)}||^2 P_{S}(\phi_S)] = 1$ if and only if $U_{(S, \phi_S)} = 0$ for every $S \in \supp(\cW), \phi_S \in \Sigma^S$ such that $P_S(\phi_S) \ne 1$, completing the proof of the lemma.
\end{proof}

We will now prove Lemma~\ref{lem:lasserre-reduction-completeness}.

\begin{proofof}[Lemma~\ref{lem:lasserre-reduction-completeness}]
  Suppose that $opt^r_{Las}(\cG) = 1$ for some $r \geq k, d k'$. Let $\{U_{(S, \phi_S)}\}$ be a complete vector solution of the $r$-level Lasserre relaxation of $\cG$.

  For brevity, let $l = \lfloor r / d \rfloor$. We will define a complete solution for $l$-level relaxation of $\cG'$ as follows. For every $T' = \{S_1, \dots, S_i\} \subseteq V'$ of size at most $l$ and every $\phi'_{T'} \in \Sigma'^{T'}$, let $S_{T'}$ denote $S_1 \cup \cdots \cup S_i$ and let $\phi_{T'}$ denote $\phi'_{T'}(S_1) \circ \cdots \circ \phi'_{T'}(S_i)$. We then define $U'_{(T', \phi'_{T'})}$ as follows.
  \begin{align*}
    U'_{(T', \phi'_{T'})} =
    \begin{cases}
      U_{(S(T'), \phi_{T'})} & \text{ if } \{\phi'_{T'}(S)\}_{S \in T'} \text{ are consistent,} \\
      0 & \text{ otherwise.}
    \end{cases}
  \end{align*}

  Next, we will show that $\{U'_{(T', \phi'_{T'})}\}$ is indeed a valid vector solution of the $l$-level Lasserre relaxation of $\cG'$. Consider any $T'_1, T'_2, T'_3, T'_4 \in \binom{V'}{l}$ and any $\phi'_1 \in \Sigma'^{T'_1}, \phi'_2 \in \Sigma'^{T'_2}, \phi'_3 \in \Sigma'^{T'_3}, \phi'_4 \in \Sigma'^{T'_4}$. For convenience, we use $S_i$ to denote $\bigcup_{S \in T'_i} S$ and  $\phi_i$ to denote $\bigcirc_{S \in T'_i} \phi'_{T'_i}(S)$. We can prove the following properties.
  \begin{itemize}
    \item $\|U'_{(\emptyset, \emptyset)}\|^2 = \|U_{(\emptyset, \emptyset)}\|^2 = 1$.
    \item Since each of $U'_{(T'_1, \phi'_1)}, U'_{(T'_2, \phi'_2)}$ is either 0 or $U_{(S, \phi_S)}$ for some $S, \phi_S$, we have $\langle U'_{(T'_1, \phi'_1)}, U'_{(T'_2, \phi'_2)}\rangle \geq 0$.
    \item Suppose that $\phi'_1, \phi'_2$ are inconsistent. If at least one of $U_{(T'_1, \phi'_1)}, U_{(T'_2, \phi'_2)}$ is 0, we have $\langle U_{(T'_1, \phi'_1)}, U_{(T'_2, \phi'_2)}\rangle = 0$. Otherwise, we have $U'_{(T'_1, \phi'_1)} = U_{(S_1, \phi_1)}$ and $U'_{(T'_2, \phi'_2)} =  U_{(S_2, \phi_2)}$.  Since $\phi'_1$ and $\phi'_2$ are inconsistent, $\phi_1$ and $\phi_2$ are also inconsistent. Hence, we also have $\langle U'_{(T'_1, \phi'_1)}, U'_{(T'_2, \phi'_2)}\rangle = \langle U_{(S_1, \phi_1)}, U_{(S_2, \phi_2)}\rangle = 0$.
    \item Suppose that $T'_1 \cup T'_2 = T'_3 \cup T'_4$ and $\phi'_1 \circ \phi'_2 = \phi'_3 \circ \phi'_4$. It is not hard to see that $\{\phi'_1(S)\}_{S \in T'_1}, \{\phi'_2(S)\}_{S \in T'_2}$ are inconsistent if and only if $\{\phi'_3(S)\}_{S \in T'_3}, \{\phi'_4(S)\}_{S \in T'_4}$ are inconsistent. In the case that they are inconsistent, we have $\langle U'_{(T'_1, \phi'_1)}, U'_{(T'_2, \phi'_2)}\rangle = 0 = \langle U'_{(T'_3, \phi'_3)}, U'_{(T'_4, \phi'_4)}\rangle$.

      On the other hand, if $\{\phi'_1(S)\}_{S \in T'_1}, \{\phi'_2(S)\}_{S \in T'_2}$ are consistent, then $U'_{(T'_1, \phi'_1)} = U_{(S_1, \phi_1)}, U'_{(T'_2, \phi'_2)} = U_{(S_2, \phi_2)}, U'_{(T'_3, \phi'_3)} = U_{(S_3, \phi_3)}$ and $U'_{(T'_4, \phi'_4)} = U_{(S_4, \phi_4)}$. Since $T'_1 \cup T'_2 = T'_3 \cup T'_4$, we have $S_1 \cup S_2 = S_3 \cup S_4$. Moreover, since $\phi'_1 \circ \phi'_2 = \phi'_3 \circ \phi'_4$, we also have $\phi_1 \circ \phi_2 = \phi_3 \circ \phi_4$. Hence, we have $\langle U'_{(T'_1, \phi'_1)}, U'_{(T'_2, \phi'_2)}\rangle = \langle U_{(S_1, \phi_1)}, U_{(S_2, \phi_2)}\rangle = \langle U_{(S_3, \phi_3)}, U_{(S_4, \phi_4)}\rangle = \langle U'_{(T'_3, \phi'_3)}, U'_{(T'_4, \phi'_4)}\rangle$.

      As a result, we have $\langle U'_{(T'_1, \phi'_1)}, U'_{(T'_2, \phi'_2)}\rangle = \langle U'_{(T'_3, \phi'_3)}, U'_{(T'_4, \phi'_4)}\rangle$ in both cases.
    \item For any $S \in V'$, from Lemma~\ref{lem:lasserre-square-sum}, we have $\sum_{\phi_S \in \Sigma^{S}} \|U'_{(S, \phi_S)}\|^2 = 1$ as desired.
  \end{itemize}

  Hence, $\{U'_{(T', \phi'_{T'})}\}$ is indeed a vector solution of the $l$-level Lasserre relaxation of $\cG'$. Finally, we show that it is also complete. From Lemma~\ref{lem:lasserre-complete-char},  it is enough to show that  $U'_{(T', \phi'_{T'})} = 0$ for all $T' \in \binom{V}{k'}$ and $\phi'_{T'} \in \Sigma^{T'}$ such that $P'_{T'}(\phi'_{T'}) \ne 1$. Consider any such $T'$ and $\phi'_{T'}$.

  If $\{\phi'_{T'}(S)\}_{S \in T'}$ are inconsistent, then $U'_{(T', \phi'_{T'})} = 0$. Otherwise, $U'_{(T', \phi'_{T'})} = U_{(S(T'), \phi_{T'})}$ where $S(T')$ and $\phi_{T'}$ are defined similarly as defined earlier in the proof. Observe that, from how $\cG'$ is defined and since $P_{T'}(\phi'_{T'}) \ne 1$, we know that there exists $x_1, \dots, x_k \in S(T')$ such that $P_{\{x_1, \dots, x_k\}}(\phi_S|_{\{x_1, \dots, x_k\}}) \ne 1$. Again, from Lemma~\ref{lem:lasserre-complete-char} and from $\{U_{(S, \phi_S)}\}$ is complete, we have $U_{(\{x_1, \dots, x_k\}, \phi_S|_{\{x_1, \dots, x_k\}})} = 0$.

  As a result, $\|U_{(S(T'), \phi_{T'})}\|^2 = \langle U_{(S(T'), \phi_{T'})}, U_{(\{x_1, \dots, x_k\}, \phi_S|_{\{x_1, \dots, x_k\}})}\rangle = 0$, which implies that $U'_{(T', \phi'_{T'})} = 0$.

  In both cases, we have $U'_{(T', \phi'_{T'})} = 0$. Hence, $\{U'_{(T', \phi'_{T'})}\}$ is a complete vector solution of the $l$-level Lasserre relaxation of $\cG'$ as desired.
\end{proofof}

\subsection{$\Omega(n)$-Level Lasserre Integrality Gap for Projection Games} \label{app:lasserre}

In this subsection, we prove Lemma~\ref{lem:lasserre-starting-instance} by a reduction from Schoenebeck's Lasserre integrality gap for random {\sc Max 3-XOR}~\cite{Sch08}. We note that a similar integrality gap was discovered prior to Schoenebeck by Grigoriev~\cite{Gri01} but Schoenebeck's result, which is stated below, is formulated in such a way that is easiler for us to use.

\begin{theorem}\cite{Sch08} \label{thm:xor-gap}
  There exists constants $d > 1$ and $1 > \varepsilon, \alpha > 0$ such that the following holds. Let $\cG = (V, \cW, \{P_{S}\})$ be a {\sc Max 3-CSP} instance constructed randomly as follows.
  \begin{itemize}
    \item $V$ is a set $\{x_1, \dots, x_n\}$.
    \item The alphabet set $\Sigma$ is $\{0, 1\}$.
    \item Randomly create $d n$ constraints as follows. Pick $i_1, i_2, i_3$ uniformly at random without replacement from $[n]$ and pick $j$ uniformly at random from $\{0, 1\}$. The predicate $P_{(x_1, x_2, x_3)}(\phi)$ is one if $\phi(x_{i_1}) \oplus \phi(x_{i_1}) \oplus \phi(x_{i_1}) = j$ and is zero otherwise.
    \item $\cW$ is simply the uniform distribution over predicates constructed in the previous step.
  \end{itemize}
  With probability $1 - o(1)$, $opt^{\alpha n}_{Las}(\cG) = 1$ and $val(\cG) \leq 1 - \varepsilon$.
\end{theorem}

We are now ready to prove Lemma~\ref{lem:lasserre-starting-instance}.

\begin{proofof}[Lemma~\ref{lem:lasserre-starting-instance}]
  First, we use a clause/variable (or constraint/variable) reduction (Definition~\ref{def:clause-variable}) to turn $\cG$ to a projection game $\cG' = (X', Y', \Sigma_X', \Sigma_Y', E', \{P'_{(x, y)}\})$. Lemma~\ref{lem:lasserre-reduction-completeness} immediately implies that $opt^{\Omega(N)}_{Las}(\cG') = 1$ with probability $1 - o(1)$. Moreover, Proposition~\ref{prop:clause-var-val} tells us that, if $val(\cG) \leq 1 - \varepsilon$, then $val(\cG') \leq 1 - \varepsilon / 3$. Thus, with probability $1 - o(1)$, $val(\cG') \leq 1 - \varepsilon/3$.

  As a result, with probability $1 - o(1)$, $\cG'$ satisfies all the desired properties in Lemma~\ref{lem:lasserre-starting-instance} except that the degrees of vertices in $Y'$ may not be bounded. To fix this, we will form a game $\hat{\cG}$ by simply removing all the vertices from $\cG'$ that has degree more than $\Delta = 100d / \varepsilon$. This immediately ensure that all vertices in $\hat{\cG}$ have bounded degrees.

  Moreover, it is obvious that, when $opt^{\Omega(N)}_{Las}(\cG')$ is one, $opt^{\Omega(N)}_{Las}(\hat{\cG})$ is also one. Thus, to prove Lemma~\ref{lem:lasserre-starting-instance}, it is enough for us to show that, with probability at least $1/2 - o(1)$, $val(\hat{\cG}) \leq 1 - \varepsilon/6$ since this would imply that, with probability $1/2 - o(1)$, $\hat{\cG}$ satisfies all the properties specified in Lemma~\ref{lem:lasserre-reduction-completeness}.

  For convenience, let $C_1, \dots, C_{dn}$ denote the clauses of $\cG$ and, for each variable $x$, let $deg(x)$ denote the degree of $x$. Observe that $\hat{\cG}$ can be constructed by conditioning $\cG'$ on the event that the variable $x$ has appears in at most $\Delta$ clauses. From Lemma 4, if at most $\varepsilon / 12$ fraction of $(C, x) \in E'$ has $deg(x) \geq \Delta$, then we have $val(\hat{\cG}) \leq val(\cG') + \varepsilon / 6$. Hence, to show that $val(\hat{\cG}) \leq 1 - \varepsilon / 6$ with probability at least $1/2 - o(1)$, it is enough to show that, with probability at least $1/2$, at most $\varepsilon / 12$ fraction of $(C, x) \in E'$ has $deg(x) \geq \Delta$. Let such event be $A$.

  For each variable $x_i$, let $D_i$ denote the event $deg(x_i) \geq \Delta$. We can rewrite $\Pr_{\cG}[A]$ as $\Pr_{\cG}[A] = \Pr_{\cG}[\sum_{(C_j, x_i) \in E'} \mathds{1}[D_i] \leq |E'|(\varepsilon / 12)]$. By Markov's inequality, to show that $\Pr_{\cG}[A] \geq 1/2$, it is enough to show that $\E_{\cG}[\sum_{(C_j, x_i)u \in E} \mathds{1}[D_i]] \leq |E'|(\varepsilon / 24) = \varepsilon dn/8$.

    For each variable $x_i$ and each clause $C_j$, let $Z_{i, j}$ denote the event that $x_i$ is involved in $C_j$. We have
    \begin{align*}
      \E_{\cG}[\sum_{(C_j, x_i) \in E} \mathds{1}[D_i]]
      &= \E_{\cG}[\sum_{i \in [n], j \in [dn]} \mathds{1}[Z_{i, j} \wedge D_i]] \\
      &= \sum_{i \in [n], j \in [dn]} \E_{\cG}[\mathds{1}[Z_{i, j} \wedge D_i]] \\
      &= \sum_{i \in [n], j \in [dn]} \Pr_{\cG}[Z_{i, j} \wedge D_i] \\
      &= \sum_{i \in [n], j \in [dn]} \Pr_{\cG}[Z_{i, j}]\Pr_{\cG}[D_i \mid Z_{i, j}] \\
      &= \sum_{i \in [n], j \in [dn]}(3/n)\Pr_{\cG}[D_i \mid Z_{i, j}].
    \end{align*}

    Now, consider $\Pr_{\cG}[D_i \mid Z_{i, j}]$. Observe that $deg(x_i)$ is simply $\sum_{l \in [dn]} \mathds{1}[Z_{i, l}]$. and that, from the sampling process of $\cG$, $Z_{i, 1}, \dots, Z_{i, dn}$ are all independent. Hence, $\Pr_{\cG}[D_i \mid Z_{i, j}] = \Pr_{\cG}[\sum_{l \in [dn] - \{j\}} \mathds{1}[Z_{i, l}] \geq \Delta - 1]$. Moreover, we know that $\E_{\cG}[\sum_{l \in [dn] - \{j\}} \mathds{1}[Z_{i, l}]] = \sum_{l \in [dn] - \{j\}} \E_{\cG}[\mathds{1}[Z_{i, l}]] = (dn - 1)(3/n) \leq 3d$. Thus, by Markov's inequality, we have $\Pr_{\cG}[D_i \mid Z_{i, j}] \leq (3d)/(\Delta - 1) \leq (3d)/(99d/\varepsilon) = \varepsilon/33$. As a result, we have $\E_{\cG}[\sum_{(C_j, x_i) \in E} \mathds{1}[D_i]] \leq \sum_{i \in [n], j \in [dn]}(3/n)(\varepsilon/33) = \varepsilon dn / 11 < \varepsilon dn / 8$ as desired.

    Hence, with probability $1/2 - o(1)$, $\hat{\cG}$ satisfies all properties stated in Lemma~\ref{lem:lasserre-starting-instance}, completing our proof.
\end{proofof}

\section{Concentration Bound on Number of Edges in Random Subgraph} \label{app:random-num-edges}

In this section, we prove Lemma~\ref{lem:random-num-edges}. We start by stating the following standard inequality regarding concentration of sum of random variables in sampling without replacement setting, which will be useful for our proof. (See, e.g.~\cite{Bard2015}, for more details about the Bernstein bound.)

\begin{theorem}(Bernstein bound) \label{thm:bernst}
  Let $X = \{x_1, \dots, x_m\}$ be a set of $m$ real numbers all lie in $[0, a]$. For $n \leq m$, sample $X_1, \dots, X_n$ without replacement uniformly at random from $X$. Let $\sigma^2$ and $\mu$ be the variance and mean of $X$ respectively. Then,
  \begin{align*}
    \Pr\left[\left|\sum_{i=1}^n X_i - n \mu\right| \leq \varepsilon\right] \leq 2\exp\left(-\frac{\varepsilon^2}{2 n \sigma^2 + a\varepsilon}\right)
  \end{align*}
\end{theorem}

Before we prove Lemma~\ref{lem:random-num-edges}, we will prove a concentration bound on the number of edges for the case where a set from only one side of the bipartite graph is sampled and the other side of the graph remains the same. This is in contrast to Lemma~\ref{lem:random-num-edges}, in which two sets are randomly sampled from both sides. The concentration bound is stated and proved below.

\begin{lemma} \label{lem:num-edges-helper}
  Let $(X, Y, E)$ be any bipartite graph where each vertex has degree at most $d_{max}$. For any non-negative integer $k \leq |X|$, let $s = \frac{k|E|}{|X|}$. For any non-negative real number $\gamma < 1/2$, we have
  \begin{align*}
    \Pr_{S \sim \binom{X}{k}}[|E(S, Y)| \notin [(1 - \gamma)s, (1 + \gamma)s]] \leq 2\exp\left(-\frac{\gamma^2 s}{3 d_{max}}\right).
  \end{align*}
\end{lemma}

\begin{proof}
  Let $d_i$ be the degree of vertex $i$ for each $i \in X$. Observe that $|E(S, Y)|$ is simply $\sum_{v \in S} d_v$. Since $S$ is sampled uniformly at random from $\binom{X}{k}$, we can view each $d_v$ in the sum as a random variable from sampling with out replacement uniformly at random from $\{d_i\}_{i \in X}$. The Bernstein bound yields the following inequality for any $\varepsilon > 0$.
  \begin{align*}
    \Pr\left[\left|\sum_{v \in S} d_v - k \mu\right| \leq \varepsilon\right] \leq 2\exp\left(-\frac{\varepsilon^2}{2k \sigma^2 + d_{max}\varepsilon}\right)
  \end{align*}
  where $\sigma^2$ and $\mu$ are the varince and mean of $\{d_i\}_{i \in X}$. Observe that $k \mu = s$ and that
  \begin{align*}
    \sigma^2 \leq \frac{1}{|X|} \sum_{i \in X} d_i^2 \leq \frac{1}{|X|} \sum_{i \in X} d_id_{max} = \frac{s d_{max}}{k}.
  \end{align*}

  Substituting $\varepsilon = \gamma s$ into the inequality from the Berstein bound completes the proof of Lemma~\ref{lem:num-edges-helper}:
  \begin{align*}
    \Pr\left[\left||E(S, Y)| - s\right| \leq \gamma s\right] &\leq 2\exp\left(-\frac{\gamma^2s^2}{2 k (s d_{max}/k) + d_{max}\gamma s}\right) \\
    (\text{Since } \gamma < 1) &\leq 2\exp\left(-\frac{\gamma^2 s^2}{3s d_{max}}\right) = 2\exp\left(-\frac{\gamma^2 s}{3 d_{max}}\right).
  \end{align*}
\end{proof}

We are now ready to prove Lemma~\ref{lem:random-num-edges}.

\begin{proofof}[Lemma~\ref{lem:random-num-edges}]
  We prove this lemma by simply applying Lemma~\ref{lem:num-edges-helper} twice. First, we use the inequality to bound the probability that $|E(S, Y)|$ is far away from its expected value. Then, we use it again on the graph induced by $S$ and $Y$ to bound the probability that $|E(S, T)|$ is far away from its expected value.

  Let $A$ denote an event that $|E(S, T)| \notin [(1 - \gamma)s, (1 + \gamma)s]$ and $B$ denote an event that $|E(S, Y)| \notin [(1 - \gamma/3)\tilde s, (1 - \gamma/3)\tilde s]$ where $\tilde s = \frac{k|E|}{|X|}$. From Lemma~\ref{lem:num-edges-helper} and from $\tilde s \geq s$, we have
  \begin{align*}
    \Pr[B] \leq 2\exp\left(-\frac{\gamma^2 \tilde s}{27 d_{max}}\right) \leq 2\exp\left(-\frac{\gamma^2 s}{27 d_{max}}\right).
  \end{align*}

  For each $S$, let $\hat{s} = \frac{l|E(S, Y)|}{|Y|}$. Applying Lemma~\ref{lem:num-edges-helper} on the graph $(S, Y, E(S, Y))$, we have
  \begin{align*}
    \Pr[|E(S, T)| \notin [(1 - \gamma/3)\hat{s}, (1 + \gamma/3)\hat{s}]] \leq 2\exp\left(-\frac{\gamma^2 \hat{s}}{27 d_{max}}\right)
  \end{align*}

  Moreover, observe that $A$ and $\neg B$ implies that $|E(S, T)| \notin [(1 - \gamma/3)\hat{s}, (1 + \gamma/3)\hat{s}]$. In addition, when $\neg B$, we have $\hat{s} \geq \frac{l(1 - \gamma/3)\tilde s}{|Y|} = (1 - \gamma/3)s \geq \frac{s}{2}.$ As a result, we arrive at the following bound.
  \begin{align*}
    \Pr[A \mid \neg B] \leq \Pr[|E(S, T)| \notin [(1 - \gamma/3)\hat{s}, (1 + \gamma/3)\hat{s}] \mid \neg B] \leq 2\exp\left(-\frac{\gamma^2 s}{54 d_{max}}\right)
  \end{align*}

  Finally, we can conclude our proof as follows.
  \begin{align*}
    \Pr[A] = \Pr[A \mid B]\Pr[B] + \Pr[A \mid \neg B]\Pr[\neg B] \leq \Pr[B] + \Pr[A \mid \neg B]
    \leq 4\exp\left(-\frac{\gamma^2 s}{54 d_{max}}\right)
  \end{align*}
\end{proofof}

\section{Bounds on Values of Two Games on Different Distributions} \label{app:inq-games-dist}

Below we provide simple proofs to Lemma~\ref{lem:inq-mult} and Lemma~\ref{lem:inq-cond}.

\begin{proofof}[Lemma~\ref{lem:inq-mult}]
  Let $\phi$ be the optimal strategy for $\cG$. We can write $val(\cG) = val_{\cG}(\phi)$ as
  \begin{align*}
    \sum_{x \in X, y \in Y} \cQ(x, y) P_{x, y}(\phi(x), \phi(y))
    \leq \alpha \cdot \sum_{x \in X, y \in Y} \cQ'(x, y) P_{x, y}(\phi(x), \phi(y))
    = \alpha \cdot val_{\cG'}(\phi) \leq \alpha \cdot val(\cG').
  \end{align*}
  This concludes the proof of Lemma~\ref{lem:inq-mult}.
\end{proofof}

\begin{proofof}[Lemma~\ref{lem:inq-cond}]
  First, observe that $\cQ'(x, y) \leq \cQ(x, y) / (1 - p)$ for every $x \in X, y \in Y$. From Lemma~\ref{lem:inq-mult}, we have $val(\cG') \leq val(\cG)/ (1 - p)$. If $p \leq 1/2$, we have $val(\cG') \leq val(\cG)/ (1 - p) \leq (1 + 2p)val(\cG) \leq val(\cG) + 2p$. Otherwise, if $p > 1/2$, we have $val(\cG') \leq 1 \leq val(\cG) + 2p$. Hence, in both cases, $val(\cG') \leq val(\cG) + 2p$.

  Next, let $\phi$ be the optimal strategy for $\cG$, i.e., $val_{\cG}(\phi) = val(\cG)$. We can rearrange $val_{\cG}(\phi)$ as
  \begin{align*}
    \E_{x, y \sim \cQ}[P_{x, y}(\phi(x), \phi(y))]
    = (1 - p) \E_{x, y \sim \cQ}[P_{x, y}(\phi(x), \phi(y)) \mid A] + p \E_{x, y \sim \cQ}[P_{x, y}(\phi(x), \phi(y)) \mid \neg A]
    \leq val_{\cG'}(\phi) + p.
  \end{align*}
  Hence, we have $val(\cG) = val_{\cG}(\phi) \leq val_{\cG'}(\phi) + p \leq val(\cG') + p$, completing our proof for the lemma.
\end{proofof}

\end{document}